%% file: arxiv_v1part1.tex
\documentclass[10pt,aps,onecolumn,tightenlines,superscriptaddress,notitlepage]{revtex4-2}

\newcommand{\vertiii}[1]{{\left\vert\kern-0.25ex\left\vert\kern-0.25ex\left\vert #1 
    \right\vert\kern-0.25ex\right\vert\kern-0.25ex\right\vert}}

\input{newdef}

\allowdisplaybreaks

\newcommand*{\addFileDependency}[1]{
  \typeout{(#1)}
  \@addtofilelist{#1}
  \IfFileExists{#1}{}{\typeout{No file #1.}}
}
\makeatother

\usepackage{algorithm}
\usepackage{algpseudocode}
\algrenewcommand\algorithmicrequire{\textbf{Input:}}
\algrenewcommand\algorithmicensure{\textbf{Output:}}

\begin{document}

\author{Francesco Anna Mele}
\email{francesco.mele@sns.it}
\affiliation{NEST, Scuola Normale Superiore and Istituto Nanoscienze, Piazza dei Cavalieri 7, IT-56126 Pisa, Italy}

\author{Giovanni Barbarino}
\email{giovanni.barbarino@gmail.com}
\affiliation{Mathematics and Operational Research Unit, Faculté Polytechnique de Mons, UMONS, BE-7000 Mons, Belgium}

\author{Vittorio Giovannetti}
\email{vittorio.giovannetti@sns.it}
\affiliation{NEST, Scuola Normale Superiore and Istituto Nanoscienze, Consiglio Nazionale delle Ricerche, Piazza dei Cavalieri 7, IT-56126 Pisa, Italy}

\author{Marco Fanizza}
\email{mf@math.ku.dk}
\affiliation{Department of Mathematical Sciences, University of Copenhagen, Universitetsparken 5, 2100 Denmark} \affiliation{{F\'{\i}sica Te\`{o}rica: Informaci\'{o} i Fen\`{o}mens Qu\`{a}ntics, Departament de F\'{i}sica, Universitat Aut\`{o}noma de Barcelona, ES-08193 Bellaterra (Barcelona), Spain}}

\title{Achievable rates in non-asymptotic bosonic quantum communication}

\begin{abstract}
Bosonic quantum communication has extensively been analysed in the asymptotic setting, assuming infinite channel uses and vanishing communication errors. Comparatively fewer detailed analyses are available in the \emph{non-asymptotic} setting, which addresses a more precise, quantitative evaluation of the optimal communication rate: how many uses of a bosonic Gaussian channel are required to transmit $k$ qubits, distil $k$ Bell pairs, or generate $k$ secret-key bits, within a given error tolerance $\varepsilon$? In this work, we address this question by finding easily computable lower bounds on the non-asymptotic capacities of Gaussian channels. To derive our results, we develop new tools of independent interest. In particular, we find a stringent bound on the probability $P_{>N}$ that a Gaussian state has more than $N$ photons, demonstrating that $P_{>N}$ decreases exponentially with $N$. Furthermore, we design the first algorithm capable of computing the trace distance between two Gaussian states up to a fixed precision.
\end{abstract}

\maketitle  

\section{Introduction}

Continuous-variable quantum communication is poised to play a critical role in the advancement of future quantum technologies, enabling secure communication, qubit transmission, and entanglement distribution over free-space links and optical fibres~\cite{BUCCO, Caves, weedbrook12, Pirandola20,usenko2025continuousvariablequantumcommunication}. Among its applications, continuous-variable quantum key distribution (CV-QKD) stands out as a mature quantum technology, supported by substantial experimental progress in recent years~\cite{Record1, Record2, Record3, Record4, Record5}. On the theoretical side,  the past two decades have seen significant efforts to estimate the \emph{capacities}~\cite{Sumeet_book} of relevant continuous-variable quantum channels~\cite{holwer,Caruso2006, Wolf2007, Mark2012,Mark-energy-constrained, Rosati2018, Sharma2018, Noh2019, Noh2020,fanizza2021estimating,Giova_classical_cap,LossyECEAC1, LossyECEAC2, PLOB, Davis2018, Goodenough16, TGW, MMMM, squashed_channel, Pirandola2009, Ottaviani_new_lower, Pirandola18,LL-bosonic-dephasing,mele2023optical,Mele_2024}. In the context of quantum Shannon theory~\cite{Sumeet_book,MARK,HOLEVO-CHANNELS-2,HOLEVO}, the capacities of a quantum channel quantify the ultimate achievable performance of quantum communication and are defined as the maximum achievable transmission rate --- i.e.~the ratio of distilled resources (e.g.~transmitted bits, transmitted qubits, distilled Bell pairs, or generated secret-key bits) to the number of channel uses --- in the asymptotic limit of infinitely many channel uses and vanishing communication errors. Capacities, as inherently \emph{asymptotic} concepts, have been extensively analysed for both discrete-variable and continuous-variable channels~\cite{Sumeet_book,HOLEVO-CHANNELS-2,HOLEVO}.

Over the past two decades, however, researchers --- particularly in the discrete-variable setting --- have shifted their focus toward a more precise type of analysis: the investigation of the communication performance when only a \emph{finite} number of channel uses is allowed. This has given rise to a new field known as \emph{non-asymptotic} or \emph{one-shot} quantum Shannon theory~\cite{Tomamichel2015, Sumeet_book, Tomamichel2012, Tomamichel2016, Tomamichel2008, Berta2011,cheng2024invitationsamplecomplexityquantum}. In this setting, the concept of channel capacity is replaced by the \emph{non-asymptotic capacities} (or one-shot capacities)~\cite{Sumeet_book}, which quantify the maximum resources --- such as qubits, Bell pairs, or secret-key bits --- which is distillable within a specified error tolerance after a fixed, finite number of channel uses. Estimating these non-asymptotic capacities addresses critical practical questions: How many uses of a given channel are needed to reliably transmit \(k\) qubits, distil \(k\) Bell pairs, or generate \(k\) secret-key bits within an error tolerance \(\varepsilon\)? Determining the minimum of such a number of channel uses is important because it indicates how long one needs to wait in order to complete a given quantum communication task. This is analogous to the concepts of "sample complexity" or "query complexity" well-studied in the quantum learning theory literature~\cite{anshu2023survey,cheng2024invitationsamplecomplexityquantum}.

While non-asymptotic capacities have been extensively studied in the discrete-variable (finite-dimensional) setting~\cite{Sumeet_book}, their exploration in the continuous-variable domain remains limited~\cite{MMMM, Kaur_2017, khatri2021secondorder, WildeRenes2016, LL-bosonic-dephasing}. In particular, the literature lacks explicitly computable lower bounds on the non-asymptotic capacities of continuous variable quantum channels as a function of the channel parameters, in contrast to the finite-dimensional case. This gap is ultimately attributed to the fact that existing techniques for deriving lower bounds in the finite-dimensional setting~\cite{Sumeet_book} are not immediately adaptable to the infinite-dimensional case, necessitating a more careful analysis.

In this work, we fill this gap by establishing easily computable lower bounds on relevant non-asymptotic capacities of bosonic Gaussian channels. Our strategy is to leverage the bounds derived for finite-dimensional systems~\cite{Sumeet_book} to estimate the one-shot capacities of restricted versions of bosonic channels. These restricted versions are defined by constraining the input and the output states to subspaces with a finite (fixed) maximum energy value. The challenging aspect of the derivation lies in demonstrating that, when the analysis is applied to truncated versions of Gaussian input states, the lower bound remains well-behaved and does not diverge as the energy constraint is lifted. To conduct our analysis, we introduce several technical tools that may hold independent interest. Among these, we find a stringent upper bound on the probability $P_{>N}$ that a Gaussian state has more than $N$ photons, demonstrating that $P_{>N}$ decreases exponentially with $N$. Building on this result, we introduce the first algorithm capable of computing the trace distance~\cite{NC,MARK,Sumeet_book} between two Gaussian states up to a fixed precision, contributing to the rapidly growing literature on perturbation bounds on Gaussian states~\cite{bittel2024optimalestimatestracedistance,fanizza2024efficienthamiltonianstructuretrace,mele2024learningquantumstatescontinuous,holevo2024estimatestracenormdistancequantum,holevo2024estimatesburesdistancebosonic}. In particular, we apply our lower bounds to the \emph{pure loss channel}~\cite{BUCCO}, which is arguably the most important Gaussian channel and is extensively used to model photon loss in optical fibres and free-space links, and to the \emph{pure amplifier channel}, which models phenomena such as spontaneous parametric down-conversion~\cite{Clerk_2010}, the dynamical Casimir effect in superconducting circuits~\cite{Moore1970}, the Unruh effect~\cite{Unruh1976}, and Hawking radiation~\cite{Hawking1972}.  Our results prove a conjecture posed in~\cite{Kaur_2017}, which states that the known upper bound~\cite{MMMM} on the $n$-shot capacities of the pure loss channel is nearly optimal.

The paper is structured as follows. Section~\ref{sec_notation} introduces the necessary notation and preliminaries to support the presentation of our results. In Section~\ref{Sec_bos_q_com_main}, we present our lower bounds on the non-asymptotic capacities of arbitrary bosonic Gaussian channels in Theorem~\ref{thm_lower_bound_arb_gauss_main}. We also provide a method to compute these bounds by deriving a closed formula for the conditional Petz–Rényi entropy of Gaussian states (Lemma~\ref{thm_cond_petz_GaussianM}). In Section~\ref{sec_main_pure_loss}, we apply these bounds to the pure loss channel and pure amplifier channel. Moreover, adopting a different approach w.r.t.~the one used in the previous section, we establish alternative lower bounds on the non-asymptotic capacities for the special case of pure loss channel, which in certain regimes improve the result of Theorem~\ref{thm_lower_bound_arb_gauss_main}. Finally, in Section~\ref{section_tail}, we derive an upper bound on the probability that a Gaussian state exceeds a fixed photon number. Leveraging this result, we introduce an algorithm to estimate the trace distance between two Gaussian states with controlled precision. In the Appendix, we provide the complete proofs of all results presented in the paper.

\section{Notation}\label{sec_notation}
\subsection{Asymptotic and non-asymptotic capacities}
Let us briefly give an informal definition of capacities and non-asymptotic capacities of a quantum channel $\NN$. For a rigorous definition, we refer to Subsection~\ref{subsec_n_shot}.

While in the asymptotic setting the channel can be used infinitely many times, allowing for a vanishing error in executing the quantum communication task, in the non-asymptotic setting the number $n$ of uses of the channel is finite and the error is just required to be below a certain non-zero error threshold $\varepsilon$. In such a setting, the role of the (asymptotic) capacities is played by the so-called \emph{$n$-shot capacities}. Specifically, in this work we investigate (asymptotic) capacities and $n$-shot capacities for three different tasks~\cite{MARK, Sumeet_book}: 
\begin{itemize}
    \item \emph{Qubit distribution}: The goal is to reliably transmit an arbitrary quantum state, possibly entangled with an ancillary system, across the quantum channel. The associated  capacity is called the \emph{$n$-shot quantum capacity} and it is defined as follows. The $n$-shot quantum capacity of a quantum channel $\NN$, denoted as $Q^{(\varepsilon,n)}(\NN)$, is the maximum number of qubits that can be transmitted with error $\varepsilon$ across $n$ uses of a quantum channel $\NN$. Moreover, the (asymptotic) \emph{quantum capacity}, denoted as $Q(\NN)$, is defined as
    \bb
        Q(\NN)\coloneqq \lim\limits_{\varepsilon \rightarrow 0^+} \liminf_{n \to \infty}\frac{Q^{(\varepsilon,n)}(\NN)}{n}\,.
    \ee
    \item \emph{Entanglement distribution assisted by two-way classical communication}: The goal is to generate maximally entangled states between two parties connected by the quantum channel, exploiting the additional resource of a two-way classical communication line. The associated  capacity is called the \emph{$n$-shot two-way quantum capacity} and it is defined as follows. The $n$-shot two-way quantum capacity of a quantum channel $\NN$, denoted as $Q_2^{(\varepsilon,n)}(\NN)$, is the maximum number of two-qubit Bell pairs --- dubbed \emph{ebits} --- that one can distil with an error not exceeding $\varepsilon$ by exploiting $n$ uses of the channel and arbitrary LOCC (Local Operation and Classical Communication) operations~\cite{Sumeet_book}. Moreover, the (asymptotic) \emph{two-way quantum capacity}, denoted as $Q_2(\NN)$, is defined as
    \bb
        Q_2(\NN)\coloneqq \lim\limits_{\varepsilon \rightarrow 0^+} \liminf_{n \to \infty}\frac{Q_2^{(\varepsilon,n)}(\NN)}{n}\,.
    \ee
    \item \emph{Secret-key distribution assisted by two-way classical communication}: The goal is to generate secret keys shared between two parties linked by the quantum channel, exploiting the additional resource of a two-way classical communication line. The associated  capacity is called the \emph{$n$-shot secret-key capacity} and it is defined as follows. The $n$-shot secret-key capacity of a quantum channel $\NN$, denoted as $K^{(\varepsilon,n)}(\NN)$, is the maximum number of secret-key bits that one can distil with an error not exceeding $\varepsilon$ by exploiting $n$ uses of the channel and arbitrary LOCCs~\cite{Sumeet_book}. Moreover, the (asymptotic) \emph{secret-key capacity}, denoted as $K(\NN)$, is defined as
    \bb
        K(\NN)\coloneqq \lim\limits_{\varepsilon \rightarrow 0^+} \liminf_{n \to \infty}\frac{K^{(\varepsilon,n)}(\NN)}{n}\,.
    \ee
\end{itemize}
The exact definitions of the error $\varepsilon$ associated with each of the three quantum communication tasks mentioned above are detailed in Section~\ref{subsec_n_shot}. Moreover, it is worth mentioning distilling $k$ ebits up to an error $\varepsilon$ also enables the generation of $k$ secret-key bits with the same level of precision, so that the ($n$-shot) secret-key capacity is always not smaller than the ($n$-shot) two-way quantum capacity. For more details, we refer to Section~\ref{subsec_n_shot}.

Let us now introduce the concept of \emph{energy-constrained capacities}~\cite{Davis2018,Mark-energy-constrained}. In practise, a quantum communication protocol can not exploit \emph{infinite} energy. Instead, all optical input signals used in a communication protocol have a bounded energy (bounded by e.g.~the energy budget available in lab). It is therefore important to consider only those quantum communication protocols that exploit input states satisfying a suitable \emph{energy constraint}. Specifically, in the continuous-variable setting, along with the regular (i.e.~unconstrained) capacities, it is common to consider the so-called energy-constrained capacities~\cite{Davis2018,Mark-energy-constrained}. Given a positive number $N_s$ (representing the energy budget per input signal), the energy-constrained capacities are defined in the same way as the unconstrained capacities, apart from the fact that the optimisation is performed over those protocols such that the average expected value of the photon number operator on all input signals is required to be at most $N_s$ (see definition of photon number operator in Section~\ref{sec_preliminaries_CV}). 
The energy-constrained quantum capacities of a quantum channel $\NN$ are denoted as $Q(\NN,N_s)$, $Q_2(\NN,N_s)$, and $K(\NN,N_s)$, while the $n$-shot energy-constrained capacities are denoted as $Q^{(\varepsilon,n)}(\NN,N_s)$, $Q_2^{(\varepsilon,n)}(\NN,N_s)$, and $K^{(\varepsilon,n)}(\NN,N_s)$. Their corresponding unconstrained counterparts  $Q(\NN)$, $Q_2(\NN)$, $K(\NN)$, $\cdots$, $K^{(\varepsilon,n)}(\NN)$
can then be recovered as the $N_s\rightarrow \infty$  limits of the latter. Specifically  
$Q(\NN)= \lim_{N_s\rightarrow \infty} Q(\NN,N_s)$, 
$Q_2(\NN)= \lim_{N_s\rightarrow \infty} Q_2(\NN,N_s)$, 
$\cdots$, 
$K^{(\varepsilon,n)}(\NN)= \lim_{N_s\rightarrow \infty} K^{(\varepsilon,n)}(\NN,N_s)$.
It is worth noting that, except for certain singular cases, these limits typically exist and are finite.

In this paper, we will derive \emph{easily computable} lower bounds (energy-constrained and unconstrained) on the \(n\)-shot capacities of Gaussian channels. By "easily," we mean that the lower bound can be expressed as a simple function of the channel parameters, the number of uses \(n\), and the error tolerance \(\varepsilon\), so that the lower bound can be easily inverted with respect to \(n\). This allows one to answer the above-mentioned question: How many uses of a bosonic Gaussian channel are sufficient to transmit \(k\) qubits, distil \(k\) Bell pairs, or generate \(k\) secret-key bits within a given error tolerance \(\varepsilon\)? To solve this, it suffices to ensure that the lower bounds on the \(n\)-shot capacities exceed \(k\), and then solve the resulting inequality for \(n\).

\subsection{Continuous-variable systems}
Here, we introduce the relevant notation regarding continuous-variable systems. For more details, we refer to Section~\ref{sec_preliminaries_CV}. A continuous-variable system with \( n \) \emph{modes} is characterised by the Hilbert space \( L^2(\mathbb{R}^n) \), consisting of all square-integrable, complex-valued functions over \( \mathbb{R}^n \). The first moment \( \mathbf{m}(\rho) \) and the covariance matrix \( V(\rho) \) of a quantum state \( \rho \) are defined as  
\[
\mathbf{m}(\rho) \coloneqq \Tr\!\left[\hat{\mathbf{R}}\,\rho\right] \quad \text{and} \quad V(\rho) \coloneqq \Tr\!\left[\left\{\mathbf{(\hat{R}-m\,\hat{\mathbb{1}}),(\hat{R}-m\,\hat{\mathbb{1}})}^{\intercal}\right\}\rho\right],
\]  
where \( (\cdot)^\intercal \) denotes the transpose, \( \{\cdot, \cdot\} \) is the anti-commutator, and \( \hat{\mathbf{R}} \coloneqq (\hat{x}_1, \hat{p}_1, \dots, \hat{x}_n, \hat{p}_n)^{\intercal} \) represents the quadrature operator vector. Here, \( \hat{x}_1, \hat{p}_1, \dots, \hat{x}_n, \hat{p}_n \) are the position and momentum operators for each mode~\cite{BUCCO}. The quadrature operator vector satisfies the canonical commutation relation  
$[\hat{\mathbf{R}}, \hat{\mathbf{R}}^\intercal] = i \, \Omega \, \mathbb{\hat{1}}$, where 
\bb\label{symp_form}
\Omega \coloneqq \bigoplus_{i=1}^n \begin{pmatrix} 0 & 1 \\ -1 & 0 \end{pmatrix} \,.
\ee
An $n$-mode \emph{Gaussian state} is a Gibbs state of a quadratic hamiltonian in the quadrature operator vector $\hat{\textbf{R}}$~\cite{BUCCO}. The set of \( n \)-mode Gaussian states is uniquely associated with pairs \( (V, \mathbf{m}) \), where \( V \) is a \( 2n \times 2n \) real matrix satisfying the \emph{uncertainty relation} \( V + i \Omega \geq 0 \), and \( \mathbf{m} \) is a \( 2n \)-dimensional real vector~\cite{BUCCO}. Specifically, for any such pair \( (V, \mathbf{m}) \) with \( V \) obeying the uncertainty relation, there exists a unique Gaussian state with covariance matrix \( V \) and first moment \( \mathbf{m} \)~\cite{BUCCO}. Conversely, the covariance matrix of any (Gaussian) state satisfies the uncertainty relation~\cite{BUCCO}.  

A Gaussian channel is a quantum channel~\cite{MARK,Sumeet_book} that maps the set of Gaussian states into itself~\cite{BUCCO}. Any Gaussian channel can be written in Stinespring representation in terms of a Gaussian unitary and an environmental vacuum state~\cite{BUCCO}. Arguably, the most important Gaussian channel is the \emph{pure loss channel}, which models photon loss in optical fibres and free-space links~\cite{BUCCO}. The pure loss channel is mathematically defined in as follows.
\begin{Def}[(Pure loss channel)]\label{def_pure_loss_MAIN}
Let $S$ and $E$ be two single-mode systems.  The pure loss channel of transmissivity $\lambda\in[0,1]$ is a quantum channel $\pazocal{E}_{\lambda}:S\to S$ defined as follows:
\bb \label{defGVloss} 
        \pazocal{E}_{\lambda}(\cdot)&\coloneqq\Tr_E\left[U_\lambda^{SE} \big(\cdot \otimes\ketbra{0}_E\big) (U_\lambda^{SE})^\dagger\right] \,,
\ee
where $\ketbra{0}_E$ denotes the environmental vacuum state, $U_\lambda^{SE}$ denotes the beam splitter unitary of transmissivity $\lambda$ defined as $U_{\lambda}^{S E}\coloneqq\exp\!\left[\arccos\sqrt{\lambda}\left(a^\dagger b-a\, b^\dagger\right)\right]$, with $a$ and $b$ being the annihilation operators of $S$ and $E$, respectively, and $\Tr_E$ denotes the partial trace with respect to the environment.
\end{Def}
Note that the pure loss channel is noiseless for $\lambda=1$ (it equals the identity channel), while it is completely noisy for $\lambda=0$ (it maps any input state into the vacuum).

Another notable example of a single-mode Gaussian channel is the \emph{pure amplifier channel}, which serves as a model for phenomena such as spontaneous parametric down-conversion~\cite{Clerk_2010}, the dynamical Casimir effect in superconducting circuits~\cite{Moore1970}, the Unruh effect~\cite{Unruh1976}, and Hawking radiation~\cite{Hawking1972}.
\begin{Def}[(Pure amplifier channel)]\label{def_pure_ampl_main}
Let $S$ and $E$ be two single-mode systems.  The pure amplifier channel of gain $g\ge1$ is a quantum channel $\Phi_{g}:S\to S$ defined as follows:
\bb
        \Phi_{g}(\cdot)&\coloneqq\Tr_E\left[U_g^{SE} \big(\cdot \otimes\ketbra{0}_E\big) (U_g^{SE})^\dagger\right] \,,
\ee
where $U_g^{SE}$ denotes the two-mode squeezing unitary of parameter $\lambda$ defined as $U_{g}^{S E}\coloneqq\exp\left[\arccosh\sqrt{g}\left(a^\dagger b^\dagger-a\, b\right)\right]$, with $a$ and $b$ being the annihilation operators of $S$ and $E$, respectively.
\end{Def}
Note that the pure amplifier channel is noiseless for $g=1$, while it becomes increasingly noisy as $g$ increases.

\subsubsection{Capacities of the pure loss channel and pure amplifier channel}
For all $\lambda\in[0,1]$ the unconstrained version of the quantum capacity $Q(\mathcal{E}_{\lambda})$~\cite{holwer, Wolf2007, Wolf2006}, the two-way quantum capacity $Q_2(\mathcal{E}_{\lambda})$~\cite{PLOB}, and the secret-key capacity $K(\mathcal{E}_{\lambda})$~\cite{PLOB} of the pure loss channel $\mathcal{E}_{\lambda}$ are given by:
\bb\label{capacities_pure_loss_main}
    Q(\mathcal{E}_{\lambda})&=   
        \begin{cases}
        \log_2\!\left(\frac{\lambda}{1-\lambda}\right) &\text{if $\lambda\in(\frac{1}{2},1]$ ,} \\
        0 &\text{if $\lambda\in[0,\frac{1}{2}]$ ,}
    \end{cases}\\
    Q_2(\mathcal{E}_{\lambda})&=K(\mathcal{E}_{\lambda})= \log_2\!\left(\frac{1}{1-\lambda}\right)\,.
\ee
Hence, the quantum capacity vanishes for all transmissivity values in the range $\lambda\in[0,\frac{1}{2}]$, while the two-way quantum capacity and secret-key capacity are always larger than zero for all non-zero values of the transmissivity. 
 In the energy-constrained scenario 
 the quantum capacity of the loss channel is provided instead by the expression~\cite{holwer, Caruso2006, Wolf2007, Mark2012, Mark-energy-constrained, Noh2019} 
\bb \label{qec_pur_loss}
    Q\left(\mathcal{E}_\lambda,N_s\right)=   
        \begin{cases}
        h(\lambda N_s)-h((1-\lambda)N_s) &\text{if $\lambda\in(\frac{1}{2},1]$ ,} \\
        0 &\text{if $\lambda\in[0,\frac{1}{2}]$ ,}
    \end{cases} 
\ee
where $h(x)\coloneqq (x+1)\log_2(x+1)-x\log_2 x$. The energy-constrained two-way quantum and secret-key capacities of the pure loss channel have not yet been determined. However, a lower bound is given by the reverse coherent information of the bipartite state obtained by sending an half of the two-mode squeezed vacuum state $\ketbra{\Psi_{N_s}}$ with local energy $N_s$ into the pure loss channel~\cite{Pirandola2009}:
\bb\label{q2_ec_pure_loss}
    K\left(\mathcal{E}_\lambda,N_s\right)\ge Q_2\left(\mathcal{E}_\lambda,N_s\right)\ge
     R\left(\mathcal{E}_\lambda,N_s\right) \coloneqq h(N_s)-h\!\left((1-\lambda)N_s\right)\,.
\ee 
In this work, we will establish  lower bounds for $n$-shot versions of all the above capacities.
An upper bound on the $n$-shot unconstrained 
two-way quantum and secret-key capacities of the pure loss channel is given by~\cite{MMMM}
\bb\label{upper_bound_n_shot}
    Q_2^{(\varepsilon,n)}(\mathcal{E}_\lambda)\le K^{(\varepsilon,n)}(\mathcal{E}_\lambda)\le n Q_2(\mathcal{E}_{\lambda})+\log_26+2\log_2\!\left(\frac{1+\varepsilon}{1-\varepsilon}\right)\,,
\ee
while an easily computable lower bound was not known before the present work.

For all $\Phi_g\ge1$ the quantum capacity $Q(\Phi_g)$~\cite{holwer, Wolf2007, Wolf2006}, the two-way quantum capacity $Q_2(\Phi_g)$~\cite{PLOB}, and the secret-key capacity $K(\Phi_{g})$~\cite{PLOB} of the pure amplifier channel $\Phi_{g}$ are all equal to each other and are given by:
\bb\label{capacities_pure_ampl_main}
Q(\Phi_{g})=Q_2(\Phi_{g})&=K(\Phi_{g})= \log_2\!\left(\frac{g}{g-1}\right)\,.
\ee  
Moreover, the energy-constrained quantum capacity of the pure amplifier channel has been determined and it reads~\cite{holwer, Caruso2006, Wolf2007, Mark2012, Mark-energy-constrained, Noh2019}  
\bb\label{q_ec_pure_ampl_main}
    Q\left(\Phi_g,N_s\right)&=h\!\left( gN_s + g-1\right)-h\!\left( (g-1)(N_s+1) \right)\,.
\ee
In contrast, the energy-constrained two-way quantum and secret-key capacities of the pure amplifier channel have not yet been determined. However, a lower bound is given by~\cite{PLOB,Pirandola2009}:
\bb\label{q2_ec_pure_amp_main}
    K\left(\Phi_g,N_s\right)&\ge Q_2\left(\Phi_g,N_s\right)\ge Q\left(\Phi_g,N_s\right)=h\!\left( gN_s + g-1\right)-h\!\left( (g-1)(N_s+1) \right)\,.
\ee 
An upper bound on the $n$-shot capacities of the pure amplifier channel is given by~\cite{MMMM}
\bb
    Q_2^{(\varepsilon,n)}(\Phi_g)&\le K^{(\varepsilon,n)}(\Phi_g)\le n\log_2\!\left(\frac{g}{g-1}\right)+\log_26+2\log_2\!\left(\frac{1+\varepsilon}{1-\varepsilon}\right)\,,
\ee
while a lower bound is proved in the present paper.

\section{Non-asymptotic bosonic quantum communication across arbitrary Gaussian channels}\label{Sec_bos_q_com_main}
The following theorem establishes an easily computable lower bound on the $n$-shot quantum capacity, $n$-shot two-way quantum capacity, and $n$-shot secret-key capacity of an arbitrary Gaussian channel. 
\begin{thm}[(Lower bound on the $n$-shot capacities of arbitrary Gaussian channels)]\label{thm_lower_bound_arb_gauss_main}
Let $\NN_{A'\to B}$ be a Gaussian channel, mapping the input system $A'$ into the output system $B$. Let $\varepsilon\in(0,1)$ be the error threshold, let $n\in\N$ be the number of uses of the channel, and let us assume that $n\ge2\log_2\!\left(\frac{2}{\varepsilon^2}\right)$. Then, the $n$-shot quantum capacity, the $n$-shot two-way quantum capacity, and the $n$-shot secret-key capacity of $\NN_{A'\to B}$ can be lower bounded as follows:
     \bb\label{lower_bound_CV2_main1}
        Q^{(\varepsilon,n)}(\NN)&\ge n I_c(A\,\rangle\, B)_{\Psi}-\sqrt{n}4\log_2\!\left(\sqrt{2^{H_{1/2}(A|B)_{\Psi}}}+\sqrt{2^{H_{1/2}(A|E)_{\Psi}}}+1\right)\sqrt{\log_2\!\left(\frac{2^9}{\varepsilon^2}\right)}-\log_2\!\left(\frac{2^{18}}{3\varepsilon^4}\right)\,,
    \ee
    \bb\label{lower_bound_CV2_main2}
        K^{(\varepsilon,n)}(\NN)&\ge Q_2^{(\varepsilon,n)}(\NN)\ge 
        \left\{ \begin{array}{l} 
        nI_c(A\,\rangle\, B)_{\Psi}-\sqrt{n}4\log_2\!\left(\sqrt{2^{H_{1/2}(A|B)_{\Psi}}}+\sqrt{2^{H_{1/2}(A|E)_{\Psi}}}+1\right)\sqrt{\log_2\!\left(\frac{8}{\varepsilon}\right)}-\log_2\!\left(\frac{16}{\varepsilon^2}\right)\,,\\ \\ 
        nI_c(B\,\rangle\, A)_{\Psi}-\sqrt{n}4\log_2\!\left(\sqrt{2^{H_{1/2}(B|A)_{\Psi}}}+\sqrt{2^{H_{1/2}(B|E)_{\Psi}}}+1\right)\sqrt{\log_2\!\left(\frac{8}{\varepsilon}\right)}-\log_2\!\left(\frac{16}{\varepsilon^2}\right)\;, \end{array} \right.
    \ee
for any input pure Gaussian state $\Phi_{AA'}$, where we denoted as $\Psi_{ABE}\coloneqq V_{A'\to BE}\Phi_{AA'}(V_{A'\to BE})^{\dagger}$ the output state associated with a Stinespring dilation isometry $V_{A'\to BE}$ of the channel. Moreover, the quantities $I_c$ and $H_{1/2}$ are finite quantities denoting coherent information and conditional Petz-Rényi entropies explicitly defined in Section~\ref{subsec_def_entropic_quant}.

Additionally, if the input state $\Phi_{AA'}$ satisfies the energy constraint $\Tr[\hat{N}_{A'}\Phi_{AA'}]\le N_s$, where $\hat{N}_{A'}$ is the photon number operator on $A'$ (see definition of photon number operator in Section~\ref{sec_preliminaries_CV}), then the right-hand side of \eqref{lower_bound_CV2_main1} and \eqref{lower_bound_CV2_main2} also constitute lower bounds on the energy-constrained $n$-shot capacities $Q^{(\varepsilon,n)}(\NN,N_s)$, $Q_2^{(\varepsilon,n)}(\NN,N_s)$, and $K^{(\varepsilon,n)}(\NN,N_s)$.
\end{thm}
The proof of the above theorem is provided in Section~\ref{section_AEP0} and 
relies on the so-called \emph{asymptotic equipartition property}~\cite{Tomamichel_2009,Furrer_2011,fawzi2023asymptotic,fang2024generalizedquantumasymptoticequipartition} and on the \emph{tail bounds on Gaussian states} explained in the forthcoming Section~\ref{section_tail}.
This theorem guarantees that, as long as  $n$ is a 
sufficiently large integer (i.e.~$n\ge2\log_2\!\left(\frac{2}{\varepsilon^2}\right)$),
$n$ uses of the Gaussian channel $\NN$ are sufficient to transmit either a number of qubits given by the lower bound in \eqref{lower_bound_CV2_main1} or a number of ebits and secret-key bits given by the lower bound in \eqref{lower_bound_CV2_main2} with an error not exceeding $\varepsilon$. Notably, the entropic quantities involved in such lower bounds are easily computable in terms of the covariance matrix of the output state $\Psi$. In fact, the calculation of the coherent information involves calculating von Neumann entropies of Gaussian states, which is well-known to be easily computable just in terms of covariance matrices~\cite{BUCCO}. Additionally, in the forthcoming lemma we show how to compute the conditional Petz-Rényi entropy $H_{1/2}(A|B)_{\rho_{AB}}\coloneqq2\log_2(\Tr[\sqrt{\rho_{AB}}\,\mathbb{1}_A\otimes\sqrt{\rho_B}])$ of a Gaussian state $\rho_{AB}$ in terms of its first moment and covariance matrix.
\begin{lemma}[(Conditional Petz-Rényi entropy of Gaussian states)]\label{thm_cond_petz_GaussianM}
    Let $A$ be an $m$-mode system and let $B$ be an $n$-mode system. Let $\rho_{AB}$ be a Gaussian state on $AB$. Then, its conditional Petz-Rényi entropy can be calculated as:
\bb 
    H_{1/2}(A|B)_{\rho_{AB}}=\log_2\left(\frac{\sqrt{\det\!\left(V_{\mathrm{sqrt}}(\rho_{AB})\right)\,\det\!\left( V_{\mathrm{sqrt}}(\rho_B)\right)}}{ \det\!\left(   
  \frac{ V_{\mathrm{sqrt}}(\rho_{AB})\big|_B + V_{\mathrm{sqrt}}(\rho_B)}{2}\right)}\right)\,,
\ee
where we denoted 
\bb
    V_{\mathrm{sqrt}}(\sigma)&\coloneqq \left[\mathbb{1}+\sqrt{\mathbb{1}- (iV(\sigma)\Omega)^{-2} }\right]V(\sigma)\qquad\forall\,\sigma\,,
    \ee
$V_{\mathrm{sqrt}}(\rho_{AB})\big|_B$ is the $2n\times 2n$ bottom-right block of the matrix $V_{\mathrm{sqrt}}(\rho_{AB})$, the state $\rho_A$ (resp.~$\rho_B$) denotes the reduced states of $\rho_{AB}$ on $A$ (resp.~B), and $\Omega$ is defined in \eqref{symp_form}.
\end{lemma}
The proof of Lemma~\ref{thm_cond_petz_GaussianM} is provided in Section~\ref{section_petz} in the Appendix. Lemma~\ref{thm_cond_petz_GaussianM} shows that the lower bounds on the \(n\)-shot capacities of Gaussian channels stated in Theorem~\ref{thm_lower_bound_arb_gauss_main}
 are easily computable, as they depend solely on the first moment and the covariance matrix of the tripartite Gaussian state of the channel's input, output, and environment. Moreover, by requiring that the lower bounds in Theorem~\ref{thm_lower_bound_arb_gauss_main}
 are larger than some number $k$ and then solving the resulting quadratic inequality with respect to $\sqrt{n}$, one can determine how many channel uses $n$ suffice to transmit \(k\) qubits, distil \(k\) Bell pairs, and generate \(k\) secret-key bits within a given error tolerance \(\varepsilon\) across the Gaussian channel $\NN$.

\section{Non-asymptotic bosonic quantum communication across the pure loss channel and pure amplifier channel}\label{sec_main_pure_loss}

Evaluating the r.h.s.~terms of Eqs.~(\ref{lower_bound_CV2_main1}) and (\ref{lower_bound_CV2_main2})
for the special case of the pure loss
channel $\mathcal{E}_\lambda$ of transmissivity $\lambda\in[0,1]$ and pure amplifier channel $\Phi_g$ of gain $g\ge1$, we obtain the following lower bounds on their $n$-shot capacities:
\bb \label{lower_bound_q_pure_loss_maintext}
    Q^{(\varepsilon,n)}(\mathcal{E}_\lambda)&\ge nQ(\mathcal{E}_\lambda)-\sqrt{n}4\log_2\!\left( \sqrt{\frac{1-\lambda}{\lambda}}+ \sqrt{\frac{\lambda}{1-\lambda}} +1 \right)\sqrt{\log_2\!\left(\frac{2^9}{\varepsilon^2}\right)}-\log_2\!\left(\frac{2^{18}}{3\varepsilon^4}\right)\,,
\\ 
K^{(\varepsilon,n)}(\mathcal{E}_\lambda)&\ge Q_2^{(\varepsilon,n)}(\mathcal{E}_\lambda) \ge nQ_2(\mathcal{E}_\lambda)-\sqrt{n}4\log_2\!\left(\sqrt{1-\lambda}+\sqrt{\frac{1}{1-\lambda}}+1\right)\sqrt{\log_2\!\left(\frac{8}{\varepsilon}\right)}-\log_2\!\left(\frac{16}{\varepsilon^2}\right)\,,\\
Q^{(\varepsilon,n)}(\Phi_g)&\ge nQ(\Phi_g)-\sqrt{n}4\log_2\!\left( \sqrt{\frac{g-1}{g}}+ \sqrt{\frac{g}{g-1}} +1 \right)\sqrt{\log_2\!\left(\frac{2^9}{\varepsilon^2}\right)}+\log_2\!\left(\frac{3\varepsilon^4}{2^{18}}\right)\,,\\
    K^{(\varepsilon,n)}(\Phi_g)&\ge Q_2^{(\varepsilon,n)}(\Phi_g) \ge nQ_2(\Phi_g)-\sqrt{n}4\log_2\!\left(\sqrt{\frac{g-1}{g}}+ \sqrt{\frac{g}{g-1}} +1 \right)\sqrt{\log_2\!\left(\frac{8}{\varepsilon}\right)}-\log_2\!\left(\frac{16}{\varepsilon^2}\right)\,,
\ee
for all $\varepsilon\in(0,1)$ and $n\ge2\log_2\!\left(\frac{2}{\varepsilon^2}\right)$
(see Theorem~\ref{thm_one_shot_Q_pure}, Theorem~\ref{thm_one_shot_Q2_pure}, and Theorem~\ref{thm_one_shot_Q_pure_ampl} of the Appendix).

It turns out that at least for the pure loss channel,  
an alternative, improved version of the above relations can be obtained by relying on 
the properties of the so-called \emph{entropy variance}~\cite{Kaur_2017,MMMM}.
Specifically, as shown in  
Section~\ref{sec_best_low}, one establishes the following result: 
 \begin{thm}[(Improved lower bounds on the $n$-shot capacities of the pure loss channel)]\label{thm_lower_bound_one_shot_main}
Let $\lambda\in[0,1]$, $\varepsilon\in(0,1)$, and $n\in\mathbb{N}$. The $n$-shot quantum capacity, the $n$-shot two-way quantum capacity, and the $n$-shot secret-key capacity of the pure loss channel $\mathcal{E}_\lambda$ can be lower bounded as follows:
    \bb\label{ineq_best_bounds_main01}
        Q^{(\varepsilon,n)}(\mathcal{E}_{\lambda}) 
        &\geq  
     n Q(\mathcal{E}_{\lambda}) - \log_2\!\left( \frac{2^{23}(32-\varepsilon)^2}{ (16-\varepsilon)\varepsilon^6}\right)\;,
    \ee
    \bb\label{ineq_best_bounds_main02}
        K^{(\varepsilon,n)}(\mathcal{E}_{\lambda})&\ge Q_2^{(\varepsilon,n)}(\mathcal{E}_{\lambda}) \ge  n Q_2(\mathcal{E}_{\lambda})
        - \log_2\!\left(\frac{2^{6}\,3\,(4-\sqrt{\varepsilon})^2}{(2-\sqrt{\varepsilon})\varepsilon^3 }\right)\,,
    \ee
    where $Q(\mathcal{E}_{\lambda})=\max\!\left(0,\log_2\!\left(\frac{\lambda}{1-\lambda}\right)\right)$ and $Q_2(\mathcal{E}_{\lambda})= \log_2\!\left(\frac{1}{1-\lambda}\right)$ are the (asymptotic) quantum capacity and the (asymptotic) two-way quantum capacity, respectively.
\end{thm}
The above theorem establishes that $n$ uses of the pure loss channel are sufficient to transmit the number of qubits given on the right-hand side of Eq.~\eqref{ineq_best_bounds_main01} up to an error $\varepsilon$ and to distil the number of ebits --- and hence secret-key bits given on the right-hand side of \eqref{ineq_best_bounds_main02} --- up to an error $\varepsilon$. One may notice
that at variance with the inequalities~(\ref{lower_bound_q_pure_loss_maintext}), the lower bounds of Theorem~\ref{thm_lower_bound_one_shot_main}
do not exhibit a negative contribution proportional to $\sqrt{n}$.
 In particular by comparing Eq.~\eqref{ineq_best_bounds_main02} with 
  the upper bound of Ref.~\cite{MMMM} reported in Eq.~\eqref{upper_bound_n_shot} this allows us to conclude that for constant values of $\varepsilon$ it holds that
\bb
    Q_2^{(\varepsilon,n)}(\mathcal{E}_{\lambda})=K^{(\varepsilon,n)}(\mathcal{E}_{\lambda})= n\log_2\!\left(\frac{1}{1-\lambda}\right) + \pazocal{O}(1)\,.
\ee 
This proves the conjecture posed in \cite{Kaur_2017}, which states that the upper bound on the $n$-shot capacities of the pure loss channel reported in Eq.~\eqref{upper_bound_n_shot} is nearly optimal. 
Additionally, Theorem~\ref{thm_lower_bound_one_shot_main} and Eq.~\eqref{upper_bound_n_shot} allow us to prove the following corollaries (explicitly proved in Section~\ref{sec_channel_compl}), which determine the number of uses of the pure loss channel required to distil \(k\) ebits, generate \(k\) secret-key bits, and distribute $k$ qubits within an error tolerance of \(\varepsilon\).
\begin{cor}[(One-shot entanglement and secret-key distribution across the pure loss channel)]
    A number 
\bb
    n\ge\frac{k+ \log_2\!\left(\frac{2^{6}\,3\,(4-\sqrt{\varepsilon})^2}{(2-\sqrt{\varepsilon})\varepsilon^3 }\right)}{\log_2\!\left(\frac{1}{1-\lambda}\right)} \,,
\ee
of uses of the pure loss channel $\mathcal{E}_\lambda$, together with LOCCs, suffice in order to generate $k$ ebits --- and hence $k$ secret-key bits --- with an error of at most $\varepsilon$. Conversely, if $n$ uses of the pure loss channel, together with LOCCs, allows one to generate $k$ secret key bits --- e.g.~$k$ ebits --- up to an error $\varepsilon$, then we must have
\bb n\ge \frac{k -\log_2\!\left(\frac{6(1+\varepsilon)^2}{(1-\varepsilon)^2}\right)}
{ \log_2\!\left(\frac{1}{1-\lambda}\right)  }\,.
\ee
Here, the "error" is measured in the same way as explained in Section~\ref{subsec_n_shot}.
\end{cor}
\begin{cor}[(One-shot qubit distribution across the pure loss channel)]\label{thm_activation_complexityQ_m}
    Let $\lambda\in(\frac{1}{2},1)$ be the transmissivity. Then, a number
    \bb
        n\ge \frac{k+  \log_2\!\left( \frac{2^{23}(32-\varepsilon)^2}{ (16-\varepsilon)\varepsilon^6}\right)}
        {\log_2\!\left(\frac{\lambda}{1-\lambda}\right)}\;,
    \ee
    of uses of the pure loss channel $\mathcal{E}_\lambda$ suffices in order to transmit $k$ qubits with an error of at most $\varepsilon$. Here, the error is measured as explained in Section~\ref{subsec_n_shot}. 
\end{cor}

\subsection{Energy-constrained scenario} \label{sec_main_pure_loss_energy_constrained}

As shown in the previous paragraphs, the {asymptotic equipartition property} approach of Theorem~\ref{thm_lower_bound_arb_gauss_main}
yields weaker lower bounds compared to those relying on the analysis of the so-called {entropy variance} of 
Theorem~\ref{thm_lower_bound_one_shot_main}. Both methods, however, are effective in deriving easily computable lower bounds on the \emph{energy-constrained} $n$-shot capacities of the pure loss channel. Notably, in the energy-constrained scenario, the bounds derived using the asymptotic equipartition property approach are the best one over a wide range of parameters. These latter lower bounds are stated in Theorem~\ref{thm_lower_ec_aep} and they read
\bb\label{eq_bound_1_ec}
        Q^{(\varepsilon,n)}(\mathcal{E}_\lambda,N_s)&\ge  n Q\left(\mathcal{E}_\lambda,N_s\right)-\sqrt{n}4\log_2\left(\sqrt{2^{H_{1/2}(A|B)_{\Psi_{ABE}^{(\lambda,N_s)}}}}+\sqrt{2^{H_{1/2}(A|E)_{\Psi_{ABE}^{(\lambda,N_s)}}}}+1\right)\sqrt{\log_2\!\left(\frac{2^9}{\varepsilon^2}\right)}-\log_2\!\left(\frac{2^{18}}{3\varepsilon^4}\right)\,,\\
        K^{(\varepsilon,n)}(\mathcal{E}_\lambda,N_s)&\ge Q_2^{(\varepsilon,n)}(\mathcal{E}_\lambda,N_s)\\
        &\ge n R\left(\mathcal{E}_\lambda,N_s\right) -\sqrt{n}4\log_2\left(  \sqrt{2^{H_{1/2}(B|A)_{\Psi_{ABE}^{(\lambda,N_s)}}}}+\sqrt{2^{H_{1/2}(B|E)_{\Psi_{ABE}^{(\lambda,N_s)}}}}+1  \right)\sqrt{\log_2\!\left(\frac{8}{\varepsilon}\right)}-\log_2\!\left(\frac{16}{\varepsilon^2}\right)\,,
    \ee
where $Q\left(\mathcal{E}_\lambda,N_s\right)$ is the (asymptotic) energy-constrained quantum capacity reported in \eqref{qec_pur_loss}, $R\left(\mathcal{E}_\lambda,N_s\right)$ is reported in \eqref{q2_ec_pure_loss}, and the expressions of $H_{1/2}(A|B)_{\Psi_{ABE}^{(\lambda,N_s)}}$, $H_{1/2}(A|E)_{\Psi_{ABE}^{(\lambda,N_s)}}$, $H_{1/2}(B|A)_{\Psi_{ABE}^{(\lambda,N_s)}}$, $H_{1/2}(B|E)_{\Psi_{ABE}^{(\lambda,N_s)}}$ are reported in \eqref{ingredient1_ec}. Instead, the lower bounds derived exploiting the relative entropy approach in the energy-constrained setting are stated in Theorem~\ref{thm_lower_bound_twoway_rel_ec} and they read
\bb\label{eq_bound_2_ec}
 Q^{(\varepsilon,n)}(\mathcal{E}_{\lambda},N_s)&
        \geq  n  Q\left(\mathcal{E}_\lambda,N_s\right)-\sqrt{n}4\sqrt{V(A|E)_{\Psi^{(\lambda,N_s)}}}\frac{1}{\sqrt{\varepsilon}}-  \log_2\left( \frac{2^{23}(32-\varepsilon)^2}{ (16-\varepsilon)\varepsilon^6}\right)\,, \\
        K^{(\varepsilon,n)}(\mathcal{E}_{\lambda},N_s)&\ge Q_2^{(\varepsilon,n)}(\mathcal{E}_{\lambda},N_s) \\
        &\ge n R\left(\mathcal{E}_\lambda,N_s\right)-\sqrt{n}\sqrt{2V(B|E)_{\Psi^{(\lambda,N_s)}}}\frac{1}{\varepsilon^{1/4}}
       -\log_2\!\left(\frac{2^{6}\,3\,(4-\sqrt{\varepsilon})^2}{(2-\sqrt{\varepsilon})\varepsilon^3 }\right)\,,
    \ee
where $V(A|E)_{\Psi^{(\lambda,N_s)}}$ and $V(B|E)_{\Psi^{(\lambda,N_s)}}$ are reported in \eqref{eq_variance_thm_ec}. 
As $N_s\rightarrow \infty$, the quantities $V(A|E)_{\Psi^{(\lambda,N_s)}}$ and $V(B|E)_{\Psi^{(\lambda,N_s)}}$ approach zero, 
 and Eq.~\eqref{eq_bound_2_ec} converges to the lower bounds established in  Theorem~\ref{thm_lower_bound_one_shot_main}.
However when $N_s$ is finite,  $V(A|E)_{\Psi^{(\lambda,N_s)}}$ and $V(B|E)_{\Psi^{(\lambda,N_s)}}$ do not vanish, resulting in non-vanishing terms proportional to $\sqrt{n}$ appearing in both \eqref{eq_bound_1_ec} and \eqref{eq_bound_2_ec}.  Moreover, for sufficiently small $\varepsilon$, the bounds based on the asymptotic equipartition property reported in \eqref{eq_bound_1_ec} are tighter than those based on entropy variance approach reported in \eqref{eq_bound_2_ec}, due to the scaling behaviour in $\varepsilon$ of the term proportional to $\sqrt{n}$. Specifically, in \eqref{eq_bound_1_ec}, the scaling is logarithmic in $\varepsilon$, while, in \eqref{eq_bound_2_ec}, the scaling is as the inverse of a power of $\varepsilon$. Therefore, for a given transmissivity $\lambda$ and and energy constraint $N_s$, if the error $\varepsilon$ is sufficiently small then the best bounds are those based on the asymptotic equipartition property reported in \eqref{eq_bound_1_ec}; otherwise, the bounds from the entropy variance approach in \eqref{eq_bound_2_ec} are tighter.

Finally, the asymptotic equipartition property approach can also be exploited to find the following lower bounds on the energy-constrained $n$-shot capacities of the pure amplifier channel:
    \bb
        Q^{(\varepsilon,n)}(\Phi_g,N_s)&\ge  n Q\left(\Phi_g,N_s\right)-\sqrt{n}4\log_2\left(\sqrt{2^{H_{1/2}(A|B)_{\Psi_{ABE}^{(g,N_s)}}}}+\sqrt{2^{H_{1/2}(A|E)_{\Psi_{ABE}^{(g,N_s)}}}}+1\right)\sqrt{\log_2\!\left(\frac{2^9}{\varepsilon^2}\right)}+\log_2\!\left(\frac{3\varepsilon^4}{2^{18}}\right)\,,\\
        K^{(\varepsilon,n)}(\Phi_g,N_s)&\ge Q_2^{(\varepsilon,n)}(\Phi_g,N_s)  \\
        &\ge  n Q\left(\Phi_g,N_s\right)-\sqrt{n}4\log_2\left(\sqrt{2^{H_{1/2}(A|B)_{\Psi_{ABE}^{(g,N_s)}}}}+\sqrt{2^{H_{1/2}(A|E)_{\Psi_{ABE}^{(g,N_s)}}}}+1\right)\sqrt{\log_2\!\left(\frac{8}{\varepsilon}\right)}-\log_2\!\left(\frac{16}{\varepsilon^2}\right)\,\,,
    \ee
    where $Q\left(\Phi_g,N_s\right)$ is reported in \eqref{q2_ec_pure_amp_main}, and $H_{1/2}(A|B)_{\Psi_{ABE}^{(g,N_s)}}$, $H_{1/2}(A|E)_{\Psi_{ABE}^{(g,N_s)}}$ are reported in \eqref{ingredient1_ec_amp}. 
\section{Tail bounds on Gaussian states}\label{section_tail}
A key step in the proof of the above Theorem~\ref{thm_lower_bound_one_shot_main} involves answering a simple-looking question: What is the probability that a fixed Gaussian state has more than $M$ photons? In other words, given a Gaussian state $\rho$ specified by its first moment $\textbf{m}$ and covariance matrix $V$, how to estimate the probability $P_{>M}$ that the outcome of the photon number measurement on $\rho$ is larger than $M$, as a function of $\textbf{m}$, $V$, and $M$? Estimating such a probability $P_{> M}$ is important, as the photon number measurement is a key element in the design of quantum technologies based on continuous-variable systems~\cite{BUCCO}.

By Born's rule, such a probability $P_{> M}$ is given by \bb
    P_{> M}\coloneqq\Tr[(\mathbb{1}-\Pi_M)\rho]\,,
\ee
where $\Pi_M$ is the projector onto the Hilbert space spanned by all the $n$-mode Fock states with photon number less or equal to $M$:
\bb\label{proj_fockmain}
    \Pi_M\coloneqq\sum_{\textbf{k}\in\mathbb{N}^n:\,\sum_{i=1}^n k_i\le M}\ketbra{\textbf{k}}\,,
\ee
with $\ket{\textbf{k}}=\ket{k_1}\otimes\ket{k_2}\otimes\ldots\otimes\ket{k_n}$ being the $n$-mode Fock state~\cite{BUCCO}. The following theorem, proved in Theorem~\ref{the:expdecay00}, allows us to upper bound the probability $P_{>M}$, establishing that it converges to zero exponentially fast in $M$, with exponential lifetime given by the mean photon number of the state.
\begin{thm}[(Probability that a fixed Gaussian state has more than $M$ photons)]\label{the:expdecay00_main}
Let $\rho$ be a Gaussian state with first moment $\textbf{m}$ and covariance matrix $V$. The probability $P_{>M}$ that the outcome of the photon number measurement on $\rho$ is larger than $M$ converges to zero exponentially fast in $M$. Specifically, it can be upper bounded as
\bb
    P_{>M}\le  \alpha\,e^{-M /(4N+2) }\,,
\ee
where $N$ denotes the mean photon number and $\alpha$ denotes a finite quantity, which can be expressed in terms of the first moment $\textbf{m}$ and covariance matrix $V$ as
\bb\label{def_N_alpha}
    N&=\frac{\Tr[V-\mathbb{1}]}{4}+\frac{\|\textbf{m}\|_2^2}{2}\,,\\
    \alpha&\coloneqq \frac{e^{\textbf{m}^\intercal((8N+4) \mathbb{1}-V)^{-1}\textbf{m} }}{\sqrt{\det\!\left[\frac{(8N+4) \mathbb{1}-V}{8N+3}\right]}}\leq 2^n e^{1/2}\,.
\ee 
A strictly tighter bound can be obtained in terms of the following one-parameter optimisation problem:
\bb\label{eq_87_main}
    P_{>M}\le \inf_{x>\|V\|_\infty}\frac{e^{\textbf{m}^\intercal(x \mathbb{1}-V)^{-1}\textbf{m} }}{\sqrt{\det\!\left[\frac{x \mathbb{1}-V}{x -1}\right]}}\, e^{-2\,\mathrm{arccoth}(x)\,M }\,,
\ee
where $\|V\|_\infty$ denotes the operator norm of $V$. 
\end{thm}
As an application of the above theorem, we can estimate the \emph{trace distance}~\cite{NC,MARK,Sumeet_book} between a Gaussian state $\rho$ and its projection $\rho_M$ onto the finite-dimensional subspace spanned by all the Fock states with at most $M$ photons, where $\rho_M$ is defined as
\bb\label{proj_rho}
    \rho_M\coloneqq \frac{\Pi_M\rho\Pi_M}{\Tr[\Pi_M\rho\Pi_M]}\,,
\ee
with $\Pi_M$ being the projector reported in \eqref{proj_fockmain}. This is useful because, in order to numerically deal with infinite-dimensional states (e.g.~to analyse properties of states prepared by applying general non-Gaussian operations to a Gaussian state), one usually needs to truncate the Fock basis expansion of the state up to a sufficiently large photon cut-off $M$, and, additionally,  to estimate the error incurred in such a truncation procedure. Due to the Holevo--Helstrom theorem~\cite{HELSTROM, Holevo1976}, the most meaningful way to measure such an error is given by the  \emph{trace-distance error}, defined as the trace distance between the original state $\rho$ and its truncated approximation $\rho_M$,  denoted as $\frac12\|\rho-\rho_M\|_1$. By combining the above theorem~\ref{the:expdecay00_main} with Gentle Measurement Lemma \cite[Lemma 6.15]{Sumeet_book}, one obtains that such a trace distance error can be upper bounded as follows:
\bb\label{ineq_tail_main}
    \frac{1}{2}\left\|\rho-\rho_M\right\|_1\le \sqrt{P_{>M}}\le \sqrt{\alpha} \,e^{-M /(8N+4) }\,,
\ee
where $\alpha$ and $N$ are reported in \eqref{def_N_alpha}. As an application of this inequality, we can find an algorithm to estimate the trace distance between two Gaussian states.

\subsection{Algorithm to estimate the trace distance between two Gaussian states}\label{sub_main_algo}
Calculating the trace distance between two Gaussian states is important, as the trace distance is the most meaningful notion of distance between quantum states~\cite{HELSTROM, Holevo1976} and Gaussian states are omnipresent in nature~\cite{BUCCO}. However, calculating exactly such a trace distance is impossible in general, as the difference of two Gaussian states is a non-Gaussian operator, and there is no way to calculate in general the trace norm of a non-Gaussian (infinite-dimensional) operator. Therefore, it is necessary to come up with methods to \emph{approximate}, up to a fixed precision, the trace distance between two Gaussian states. While recent works have established both lower~\cite{mele2024learningquantumstatescontinuous} and upper~\cite{bittel2024optimalestimatestracedistance,fanizza2024efficienthamiltonianstructuretrace,mele2024learningquantumstatescontinuous,holevo2024estimatestracenormdistancequantum,Banchi_2015} bounds on the trace distance between two Gaussian states, there remains no known algorithm capable of computing such a trace distance up to a fixed precision. Here, we fill this gap, by providing the first algorithm to compute the trace distance between Gaussian states up to a fixed precision. The algorithm is as follows:
\begin{algorithm}[H]\label{algo1_main}
\caption{(Computation of the trace distance $\frac12\|\rho_{\textbf{m},V}-\rho_{\textbf{t},W}\|_1$ between two $n$-mode Gaussian states $\rho_{\textbf{m},V}$ and $\rho_{\textbf{t},W}$, given their first moments $\textbf{m},\textbf{t}$ and covariance matrices $V,W$.}
\begin{algorithmic}[1]
\Require First moments $\textbf{m},\textbf{t}$; covariance matrices $V,W$; error tolerance $\epsilon > 0$.
 \Ensure Approximation $\tilde{d}$ such that $ \left|\frac12\|\rho_{\textbf{m},V}-\rho_{\textbf{t},W}\|_1- \tilde{d}\right| \leq \epsilon$
  \State Use Eq.~\eqref{def_N_alpha} to find the maximum between the mean photon number of $\rho_{\textbf{m},V}$ and the one of $\rho_{\textbf{t},W}$, denoted as $N$.
 \State Use Eq.~\eqref{ineq_tail_main} (with the determined value of $N$) to find $M$ large enough so that both the two Gaussian states are $\epsilon/3$-close to their projection on the subspace spanned by all the $n$-mode Fock states with at most $M$ photons. 
 \State Calculate the matrix elements of $\rho_{\textbf{m},V}$ with respect to those Fock states with at most $M$ photons (e.g.~via the methods introduced in \cite{Quesada_2019}). With such matrix elements, construct a finite-dimensional matrix $A_1$ and set $\rho_1\coloneqq A_1/\Tr A_1$.
\State Analogously, calculate the matrix elements of $\rho_{\textbf{t},W}$ with respect to the same basis as above, and construct the corresponding finite-dimensional matrix $A_2$. Finally, set $\rho_2\coloneqq A_2/\Tr A_2$.
 \State Compute the trace distance between the finite dimensional matrices $\rho_1$ and $\rho_2$ up to an error $\epsilon/3$.
 \State \Return the estimate of the above trace distance.
 \end{algorithmic}
 \end{algorithm}
 As proved in Theorem~\ref{thm_algo} in the Appendix, the above algorithm computes the trace distance between two $n$-mode Gaussian states up to a trace-distance error of $\epsilon$ in a runtime of $\pazocal{O}\!\left(\!\left(2^6(n+1)E\log\left(\frac{2}{\epsilon}\right)\right)^{3n}\right)$, where $E$ denotes the maximum mean energy per mode of the two states. Notably, the logarithmic dependence on the trace-distance error $\epsilon$ of this runtime is a distinctive feature of Gaussian states. In contrast, a similar "truncation algorithm" for computing the trace distance between two $n$-mode \emph{non-Gaussian} states up to a precision error of $\epsilon$ has a significantly worse runtime of $\pazocal{O}\!\left(\!\left(\frac{16E}{\epsilon^2}\right)^{3n}\right)$, as proved in Lemma~\ref{lemma_alt_alg} in the Appendix.

\section{Discussion}
In this work, we have derived easily computable lower bounds on the non-asymptotic capacities of Gaussian channels. Our results address the fundamental question: given a Gaussian channel $\Phi$, a number $k$, and an error tolerance $\varepsilon$, how many uses of $\Phi$ are sufficient to transmit $k$ qubits, distil $k$ ebits, and generate $k$ secret-key bits within the specified error tolerance $\varepsilon$?  
Our work introduces several tools of independent interest, including: an upper bound on the probability that a Gaussian state exceeds a specified number of photons (Theorem~\ref{the:expdecay00_main}); an algorithm to estimate the trace distance between two Gaussian states with a fixed precision (Table~\ref{sub_main_algo}); a closed formula for the conditional Petz–Rényi entropy of a Gaussian state (Lemma~\ref{thm_cond_petz_GaussianM}).

\medskip

\begin{acknowledgments}
\smallskip
\noindent \emph{Acknowledgements.}  We thank Ludovico Lami, Matteo Rosati, Mark Wilde, and Andreas Winter for useful discussions. F.A.M. and V.G. acknowledge financial support by MUR (Ministero dell'Istruzione, dell'Universit\`a e della Ricerca) through the following projects: PNRR MUR project PE0000023-NQSTI, PRIN 2017 Taming complexity via Quantum Strategies: a Hybrid Integrated Photonic approach (QUSHIP) Id.\ 2017SRN-BRK, and project PRO3 Quantum Pathfinder.
G.B. is member of the Research Group GNCS (Gruppo Nazionale per il Calcolo Scientifico) of INdAM (Istituto Nazionale di Alta Matematica) and  acknowledges the support by the ERC Consolidator Grant 101085607 through the Project eLinoR. M. F. is supported by the European Research Council (ERC) under Agreement 818761 and by VILLUM FONDEN via the QMATH Centre of Excellence (Grant No. 10059). M.F. was previously supported by a Juan de la Cierva Formaci\'on fellowship (Spanish MCIN project FJC2021-047404-I), with funding from MCIN/AEI/10.13039/501100011033 and European Union NextGenerationEU/PRTR, by European Space Agency, project ESA/ESTEC 2021-01250-ESA, by Spanish MCIN (project PID2022-141283NB-I00) with the support of FEDER funds, by the Spanish MCIN with funding from European Union NextGenerationEU (grant PRTR-C17.I1) and the Generalitat de Catalunya, as well as the Ministry of Economic Affairs and Digital Transformation of the Spanish Government through the QUANTUM ENIA ``Quantum Spain'' project with funds from the European Union through the Recovery, Transformation and Resilience Plan - NextGenerationEU within the framework of the "Digital Spain 2026 Agenda".
\end{acknowledgments}

\medskip


\bibliographystyle{apsrev4-1}
\nocite{apsrev41Control}
\bibliography{biblio,bibarb,revtex-custom}

\newpage
\appendix
\tableofcontents

\section{Preliminaries and notations}
For any $p\in[1,\infty)$, the \emph{Schatten $p$-norm} of a linear operator $X$ is defined as 
\bb
\|X\|_p \coloneqq \left(\Tr\left[\left(\sqrt{ X^\dagger X }\right)^p\right]\right)^{1/p}.
\ee
For $p = 1$, the Schatten $p$-norm is known as the \emph{trace norm}, while for $p = 2$, it is referred to as the \emph{Frobenius norm} or \emph{2-norm}. The \emph{operator norm} of $X$, denoted as $ \|X\|_\infty$, is the largest singular value of $X$. H\"older's inequality establishes that for all linear operators $A,B$ it holds that $|\Tr[A^\dagger B]|\le\|A\|_p\|B\|_q$ where $ p,q\in[1,\infty]$ such that $p^{-1}+q^{-1}=1$.

Given an Hilbert space $\HH$, a quantum state $\rho$ on  $\HH$ is a positive semi-definite operator on $\HH$ with unit trace. Given two Hilbert spaces $\HH_A,\HH_B$ and given a bipartite state $\rho_{AB}$ on $\HH_A\otimes \HH_B$, its partial trace with respect $A$ is denoted as $\rho_B\coloneqq \Tr_A\rho_{AB}$.

The \emph{trace distance} between two quantum states $\rho$ and $\sigma$ is defined as $\frac{1}{2}\|\rho-\sigma\|_1$ and it is considered to be the most meaningful distance between quantum states due to its operational meaning given by the Holevo--Helstrom theorem~\cite{HELSTROM, Holevo1976}. The fidelity between two states $\rho$ and $\sigma$ is defined as $F(\rho,\sigma)\coloneqq \|\sqrt{\rho}\sqrt{\sigma}\|_1^2$.

Given two quantum systems $A$ and $B$ associated with Hilbert spaces $\HH_A$ and $\HH_B$, respectively, a quantum channel $\Phi_{A\to B}$ is a linear, completely positive, trace-preserving map from the space of linear operators on $\HH_A$ to the space of linear operators on $\HH_B$. For simplicity, we sometimes omit the subscript $A \to B$ and denote the quantum channel simply as $\Phi$.
\subsection{$n$-shot quantum communication}\label{subsec_n_shot}
A fundamental problem in quantum Shannon theory~\cite{MARK,Sumeet_book} is determining the optimal achievable performances of relevant quantum communication tasks over a given quantum channel. These optimal performances are quantified by the so-called \emph{capacities} of quantum channels. Capacities of physically interesting channels, such as bosonic Gaussian channels~\cite{BUCCO} (see Section~\ref{sec_preliminaries_CV}), have been extensively analysed in the \emph{asymptotic setting}~\cite{PLOB, Rosati2018, Sharma2018, Noh2019,holwer, Noh2020,fanizza2021estimating, lower-bound, Wolf2007, Mark2012, Mark-energy-constrained, LossyECEAC1, LossyECEAC2,LL-bosonic-dephasing,Mele_2024}, while they are not so much explored in the \emph{non-asymptotic setting}~\cite{MMMM, Kaur_2017, khatri2021secondorder, WildeRenes2016}. In the asymptotic setting, the channel can be used infinitely many times, allowing for a vanishing error in executing the quantum communication task. Instead, in the non-asymptotic setting, the number $n$ of uses of the channel is finite, and the error is just required to be below a certain non-zero error threshold $\varepsilon$. In such a setting, the capacities are referred to as \emph{$n$-shot capacities}. Specifically, in this work we investigate $n$-shot capacities for three different tasks~\cite{MARK, Sumeet_book}: 
\begin{itemize}
    \item \emph{Qubit distribution}: The goal is to reliably transmit an arbitrary quantum state, possibly entangled with an ancillary system, across the quantum channel;
    \item \emph{Entanglement distribution assisted by two-way classical communication}: The goal is to generate maximally entangled states between two parties connected by the quantum channel, exploiting the additional resource of a two-way classical communication line;
    \item \emph{Secret-key distribution assisted by two-way classical communication}: The goal is to generate secret keys shared between two parties linked by the quantum channel, exploiting the additional resource of a two-way classical communication line;
\end{itemize}
In the following paragraphs we define the $n$-shot capacities for each of these tasks. For more details, please refer to~\cite{Sumeet_book}.

\subsubsection{Definition of $n$-shot quantum capacity}\label{paragraph_quantum_cap}
In this paragraph, we define the $n$-shot quantum capacity of quantum channels~\cite{Sumeet_book}. Roughly speaking,  given a quantum channel $\Phi$, an error threshold $\varepsilon$, and a number $n$ of channel uses, the $n$-shot quantum capacity $Q^{(\varepsilon,n)}(\Phi)$ is the maximum number of qubits that can be transmitted through $n$ uses of $\Phi$ with an error not exceeding $\varepsilon$. Below, we provide a precise definition of $n$-shot quantum capacity. We start by stating the definition of a quantum communication code over a quantum channel.
\begin{Def}[$(|M|,\varepsilon)$-code for quantum communication over a quantum channel~\cite{Sumeet_book}]
Let $\Phi_{A\to B}$ be a quantum channel. Let $\varepsilon\in(0,1)$ and $|M|\in\N$. If there exist an Hilbert space $M$ of dimension $|M|$ and two quantum channels $\mathcal{E}_{M\to A},\mathcal{D}_{B \to M}$ --- called the "encoding channel" and the "decoding channel", respectively --- such that the channel fidelity between the identity map $\Id_{M}$ and the channel $\mathcal{D}_{B \to M}\circ \Phi_{A\to B}  \circ \mathcal{E}_{M\to A}$ --- called the "encoded channel" ---  satisfies
\bb\label{eq_def_Q_cap}
F\!\left( \mathcal{D}_{B \to M}\circ \Phi_{A\to B}  \circ \mathcal{E}_{M\to A}, \Id_{M} \right) \geq 1-\varepsilon\,,
\ee
then the triplet $(M,\mathcal{E}_{M\to A},\mathcal{D}_{B \to M})$ is said to be a $(|M|,\varepsilon)$-code for quantum communication over $\Phi_{A\to B}$. The channel fidelity between two channels $\Phi_1$ and $\Phi_2$ is defined in terms of terms of the fidelity between states as follows:
\bb\label{channel_fidelity}
F(\Phi_1,\Phi_2) \coloneqq \inf_{\Psi_{RS}} F\!\left(\Id_R \otimes \Phi_1(\Psi_{RS}),\Id_R \otimes \Phi_2(\Psi_{RS})\right)\,,
\ee
where the optimisation is over every reference system $R$ and every bipartite state $\Psi_{RS}$ on the channel input $S$ and on the reference system $R$. 
\end{Def}
Basically, the condition in \eqref{eq_def_Q_cap} means that the channel $\mathcal{D}_{B \to M}\circ \Phi_{A\to B}  \circ \mathcal{E}_{M\to A}$ is an $\varepsilon$-approximation of the identity channel $\Id_M$ over a system of dimension $|M|$. Hence, $\varepsilon$ can be interpreted as the error incurred in the quantum communication code, while $\log_2|M|$ represents the number of qubits that can be transmitted with error $\varepsilon$ by exploiting the code.
\begin{Def}
[One-shot quantum capacity of a quantum channel~\cite{Sumeet_book}]
    Given a quantum channel $\Phi$ and an error threshold $\varepsilon\in(0,1)$, the one-shot quantum capacity $Q^{(\varepsilon)}(\Phi)$ of the channel $\Phi$ is defined as
\bb
Q^{(\varepsilon)}(\Phi) \coloneqq \sup \{\log_2 |M| : \exists \text{$(|M|,\varepsilon)$-code for quantum communication over } \Phi_{A\to B}\}\,.
\ee
\end{Def}
In other words, the one-shot quantum capacity $Q^{(\varepsilon)}(\Phi)$ is the maximum number $\log_2|M|$ of transmitted qubits among all $(|M|,\varepsilon)$-codes for quantum communication over $\Phi$. We are now ready to define the $n$-shot quantum capacity.
\begin{Def}
[$n$-shot quantum capacity of a quantum channel~\cite{Sumeet_book}]
    Given a quantum channel $\Phi$, an error threshold $\varepsilon\in(0,1)$, and a number $n$ of channel uses, the $n$-shot quantum capacity $Q^{(\varepsilon,n)}(\Phi)$ is defined as
\bb
Q^{(\varepsilon,n)}(\Phi) \coloneqq Q^{(\varepsilon)}(\Phi^{\otimes n})\,.
\ee
\end{Def}
The maximum number of qubits that can be transmitted with error $\varepsilon$ across $n$ uses of a quantum channel $\Phi$ is thus given by the $n$-shot quantum capacity $Q^{(\varepsilon,n)}(\Phi)$. The (asymptotic) \emph{quantum capacity} is then defined in terms of the $n$-shot quantum capacity as 
\bb
    Q(\Phi)\coloneqq \lim\limits_{\varepsilon \rightarrow 0^+} \liminf_{n \to \infty}\frac{Q^{(\varepsilon,n)}(\Phi)}{n}\,.
\ee
Hence, the quantum capacity of a quantum channel is the maximum achievable rate of qubits (i.e.~the ratio between the number of transmitted qubits $\log_2|M|$ and the number $n$ of uses of the channel) transmitted with vanishing error ($\varepsilon\rightarrow0^+$) in the asymptotic limit of infinitely many uses of the channel ($n\rightarrow\infty$).

\subsubsection{Definition of $n$-shot two-way quantum capacity}\label{paragraph_two_way}
In this paragraph, we define the $n$-shot two-way quantum capacity of quantum channels~\cite{Sumeet_book}. This is a figure of merit of quantum communication in scenarios where Alice (the sender at the input of the channel) and Bob (the receiver at the output of the channel) are able to implement arbitrary local operations assisted by classical communication (LOCC)~\cite{Sumeet_book,LOCC}. Roughly speaking, given a quantum channel $\Phi$, an error threshold $\varepsilon$, and a number $n$ of channel uses, the $n$-shot two-way quantum capacity $Q_2^{(\varepsilon,n)}(\Phi)$ is the maximum number of two-qubit maximally-entangled states (or \emph{ebits}) that can be distilled using LOCCs and $n$ uses of $\Phi$ with an error not exceeding $\varepsilon$. Below, we provide a precise definition of $n$-shot two-way quantum capacity.  We start by stating the definition of an LOCC-assisted quantum communication protocol over a quantum channel.
\begin{Def}[$(n,|M|,\varepsilon)$ LOCC-assisted quantum communication protocol over a quantum channel~\cite{Sumeet_book}]\label{def_locc_protocol}
Let $\Phi_{A\to B}$ be a quantum channel. Let $n\in\N$, $|M|\in\N$, and $\varepsilon\in(0,1)$. An $(n,|M|,\varepsilon)$ LOCC-assisted quantum communication protocol over $\Phi$ consists of a product state $\ket{0}_{A_0}\otimes\ket{0}_{B_0}$ and $n+1$ LOCC channels~\cite{Sumeet_book} $\mathcal{L}^{(0)}_{A_0B_0\to A'_1 A_1 B'_1}$, $\mathcal{L}^{(1)}_{A'_1 B_1 B'_1\to A'_2 A_2 B'_2}$, $\mathcal{L}^{(2)}_{A'_2 B_2 B'_2\to A'_3 A_3 B'_3}$, $\ldots$, $\mathcal{L}^{(n-1)}_{A'_{n-1} B_{n-1} B'_{n-1}\to A'_n A_n B'_n}$, $\mathcal{L}^{(n)}_{A'_n B_n B'_n\to M_A M_B}$ such that the fidelity between the final state of the protocol
\bb\label{final_state_eta}
\eta_n^{M_AM_B} &\coloneqq (\mathcal{L}^{(n)}_{A'_{n}B_{n}B'_{n}\to M_A M_B} \circ \Phi_{A_n \to B_n} \circ \mathcal{L}^{(n-1)}_{A'_{n-1}B_{n-1}B'_{n-1} \to A'_{n}A_{n}B'_{n}}\circ \cdots\\
&\qquad\qquad\qquad\circ  \mathcal{L}^{(1)}_{A'_{1}B_{1}B'_{1} \to A'_{2}A_{2}B'_{2}}\circ  \Phi_{A_1 \to B_1}\circ \mathcal{L}^{(0)}_{A_0B_0\to A'_1 A_1 B'_1})(\ketbra{0}_{A_0}\otimes\ketbra{0}_{B_0})
\ee
and the maximally entangled state $\Gamma_{|M|}^{M_A M_B}$ of Schmidt rank $|M|$ satisfies
\bb\label{def_eq_two_way_q}
F\!\left(\eta_n^{M_A M_B}, \Gamma_{|M|}^{M_A M_B}\right) \geq 1- \varepsilon.
\ee
\end{Def}
A pictorial representation of an LOCC-assisted quantum communication protocol is provided in Fig.~\ref{fig_LOCC_protocol}. Basically, the condition in \eqref{def_eq_two_way_q} means that the final state of the protocol is an $\varepsilon$-approximation of the state of $\log_2|M|$ ebits. In summary, an $(n,|M|,\varepsilon)$ LOCC-assisted quantum communication protocol allows one to generate $\log_2|M|$ ebits between Alice and Bob with an error not exceeding $\varepsilon$ by exploiting $n$ uses of the channel and arbitrary LOCCs.
\begin{figure}[!h]
	\centering
	\includegraphics[width=1\linewidth]{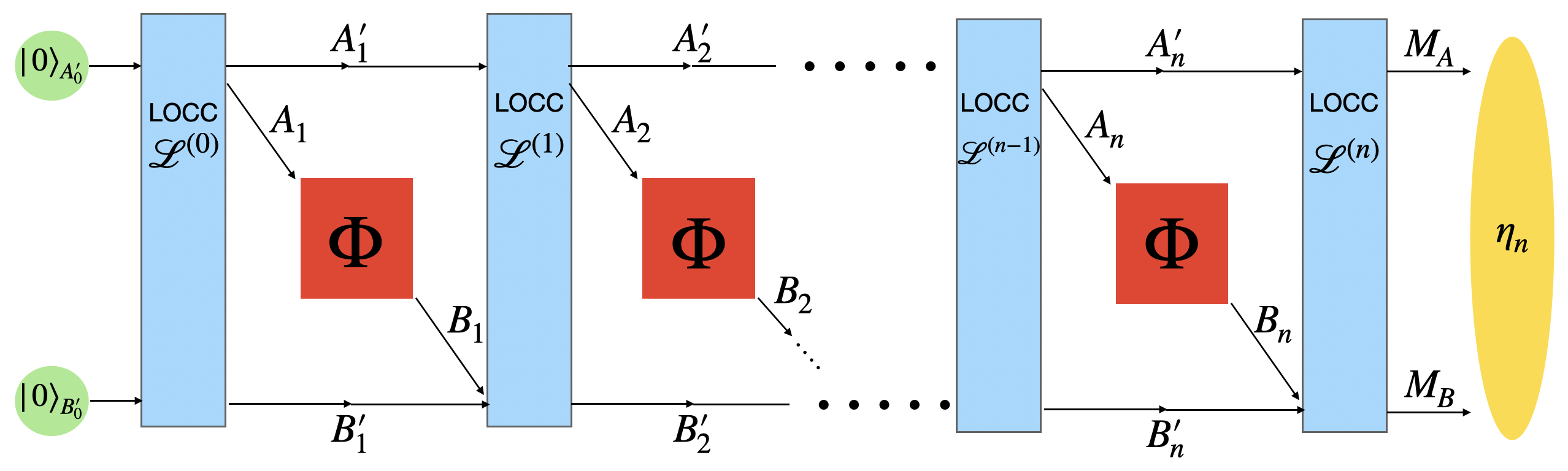}
	\caption{Pictorial representation of an LOCC-assisted quantum communication protocol over a quantum channel $\Phi$. 
	}
    \label{fig_LOCC_protocol}
\end{figure}
We are now ready to define the $n$-shot two-way quantum capacity.
\begin{Def}
[$n$-shot two-way quantum capacity of a quantum channel~\cite{Sumeet_book}]
    Given a quantum channel $\Phi$, an error threshold $\varepsilon\in(0,1)$, and a number $n$ of channel uses, the $n$-shot two-way quantum capacity $Q_2^{(\varepsilon,n)}(\Phi)$ is defined as
\bb
    Q_2^{(\varepsilon,n)}(\Phi) \coloneqq
\sup \left\{ \log_2 |M| : \exists (n,|M|,\varepsilon) \text{ LOCC-assisted quantum communication protocol over } \Phi \right \}.
\ee
\end{Def}
In other words, the $n$-shot two-way quantum capacity $Q_2^{(\varepsilon,n)}(\Phi)$ is the maximum number $\log_2|M|$ of distilled ebits among all the $(n,|M|,\varepsilon)$-codes for quantum communication over $\Phi$. The (asymptotic) \emph{two-way quantum capacity} is then defined in terms of the $n$-shot two-way quantum capacity as 
\bb
Q_{2}(\Phi) \coloneqq \lim_{\varepsilon\rightarrow0^+} \liminf_{n \to \infty}\frac{Q_2^{(\varepsilon,n)}(\Phi)}{n}\,.
\ee
Hence, the two-way quantum capacity of a quantum channel is the maximum achievable rate of ebits distilled via LOCCs with vanishing error in the asymptotic limit of infinitely many uses of the channel.

\subsubsection{Definition of $n$-shot secret-key capacity}\label{paragraph_secret_key}
In this paragraph, we define the $n$-shot secret-key capacity of quantum channels~\cite{Sumeet_book}. Roughly speaking, given a quantum channel $\Phi$, an error threshold $\varepsilon$, and a number $n$ of channel uses, the $n$-shot secret-key capacity $K^{(\varepsilon,n)}(\Phi)$ is the maximum number of secret-key bits that can be generated between two parties using LOCCs and $n$ uses of $\Phi$, such that the error does not exceed $\varepsilon$. The secret key has to be secret to any third party --- typically referred to as Eve --- with access to the environment associated with a Stinespring dilation of the channel. Mathematically, this means that Alice and Bob have to generate a \emph{private state}.
\begin{Def}[Private state~\cite{private,Horodecki_2009_secret2,Sumeet_book}]
Let $|M|\in\mathbb{N}$. A private state of dimension $|M|$ is a state the form 
\bb
\gamma_{|M|}^{M_A M_B S_A S_B} \coloneqq U_{M_A M_B S_A S_B} (\Gamma_{|M|}^{M_A M_B} \otimes \theta_{S_A S_B} )U_{M_A M_B S_A S_B}^\dag\,,
\ee
where $U_{M_A M_B S_A S_B}$ is a unitary of the form
\bb
U_{M_A M_B S_A S_B} \coloneqq \sum_{i,j} \ketbra{i}_{M_A} \otimes \ketbra{j}_{M_B} \otimes U^{i,j}_{S_A S_B},
\ee
with each $U^{i,j}_{S_A S_B}$ being a unitary, $\Gamma_{|M|}^{M_A M_B}$ representing a maximally entangled state of Schmidt rank $|M|$, and $\theta_{S_A S_B}$ being an arbitrary state. 
\end{Def}
It turns out that generating a private state of dimension $|M|$ is equivalent to generating a secret key of length $|M|$~\cite{private,Horodecki_2009_secret2}. Mathematically, the latter involves generating a \emph{tripartite secret-key state} of dimension $|M|$, which is a state of the form
\bb
    \frac{1}{|M|} \sum_{m=0}^{|M|-1}\ketbra{m}_{M_A} \otimes \ketbra{m}_{M_B} \otimes \sigma_E\,
\ee
where $\sigma_E$ is an arbitrary state on Eve's system. Thus, in a tripartite secret-key state, Eve's system $E$ is entirely decoupled from Alice's system $M_A$ and Bob's system $M_B$. Additionally, Alice and Bob are perfectly classically correlated: they share the state $\ket{m}_{M_A} \otimes \ket{m}_{M_B}$, where $m$ being a uniform random variable over $\{0, 1, \ldots, |M|-1\}$, with its actual value being the secret key. Hence, since the probability distribution of $m$ is uniform, Eve cannot infer any non-trivial information about the value of $m$.
\begin{Def}[$(n,|M|,\varepsilon)$ protocol for secret key agreement over a quantum channel~\cite{Sumeet_book}]
Let $\Phi_{A\to B}$ be a quantum channel. Let $n\in\N$, $|M|\in\N$, and $\varepsilon\in(0,1)$. An $(n,|M|,\varepsilon)$-protocol for secret key agreement over $\Phi$ consists of a product state $\ket{0}_{A_0}\otimes\ket{0}_{B_0}$ and $n+1$ LOCC channels~\cite{Sumeet_book} $\mathcal{L}^{(0)}_{A_0B_0\to A'_1 A_1 B'_1}$, $\mathcal{L}^{(1)}_{A'_1 B_1 B'_1\to A'_2 A_2 B'_2}$, $\mathcal{L}^{(2)}_{A'_2 B_2 B'_2\to A'_3 A_3 B'_3}$, $\ldots$, $\mathcal{L}^{(n-1)}_{A'_{n-1} B_{n-1} B'_{n-1}\to A'_n A_n B'_n}$, $\mathcal{L}^{(n)}_{A'_{n}B_{n}B'_{n}\to M_A M_B S_A S_B}$ such that there exists a private state $\Gamma_{|M|}^{M_AM_BS_AS_B}$ of dimension $|M|$ for which the fidelity between the final state of the protocol
\bb\label{final_state_eta2}
\eta_n^{M_AM_BS_AS_B} &\coloneqq (\mathcal{L}^{(n)}_{A'_{n}B_{n}B'_{n}\to M_A M_B S_A S_B} \circ \Phi_{A_n \to B_n} \circ \mathcal{L}^{(n-1)}_{A'_{n-1}B_{n-1}B'_{n-1} \to A'_{n}A_{n}B'_{n}}\circ \cdots\\
&\qquad\qquad\qquad\circ  \mathcal{L}^{(1)}_{A'_{1}B_{1}B'_{1} \to A'_{2}A_{2}B'_{2}}\circ  \Phi_{A_1 \to B_1}\circ \mathcal{L}^{(0)}_{A_0B_0\to A'_1 A_1 B'_1})(\ketbra{0}_{A_0}\otimes\ketbra{0}_{B_0})
\ee
and $\Gamma_{|M|}^{M_AM_BS_AS_B}$ satisfies
\bb\label{def_eq_two_way_q22}
F\!\left(\eta_n^{M_A M_BS_AS_B}, \Gamma_{|M|}^{M_A M_BS_AS_B}\right) \geq 1- \varepsilon.
\ee
\end{Def}
Therefore, a protocol for secret key agreement is similar to an LOCC-assisted quantum communication protocol as defined in Definition~\ref{def_locc_protocol}, with two differences: the final LOCC and the final state of the protocol. In a secret key agreement protocol, the final LOCC is replaced by a more general LOCC $\mathcal{L}^{(n)}_{A'_{n}B_{n}B'_{n}\to M_A M_B S_A S_B}$, which outputs not only Alice's system $M_A$ and Bob's system $M_B$, but also an additional system for Alice, $S_A$, and an additional system for Bob, $S_B$. Moreover, the final state of the protocol has to be an $\varepsilon$-approximation of a private state. Since a maximally entangled state of Schmidt rank $|M|$ is a particular example of a private state of dimension $|M|$, it follows that any $(n,|M|,\varepsilon)$ LOCC-assisted quantum communication protocol is also an $(n,|M|,\varepsilon)$ protocol for secret key agreement. In summary, an $(n,|M|,\varepsilon)$ protocol for secret key agreement allows one to generate $\log_2|M|$ secret-key bits shared by Alice and Bob with an error not exceeding $\varepsilon$ by exploiting $n$ uses of the channel and arbitrary LOCCs. We are now ready to define the $n$-shot secret-key capacity.
\begin{Def}
[$n$-shot secret-key capacity of a quantum channel~\cite{Sumeet_book}]
    Given a quantum channel $\Phi$, an error threshold $\varepsilon\in(0,1)$, and a number $n$ of channel uses, the $n$-shot secret-key capacity $K^{(\varepsilon,n)}(\Phi)$ is defined as
\bb
    K^{(\varepsilon,n)}(\Phi) \coloneqq
\sup \left\{ \log_2 |M| : \exists (n,|M|,\varepsilon) \text{ secret key agreement protocol over } \Phi \right \}.
\ee
\end{Def}
In other words, the $n$-shot secret-key capacity $K^{(\varepsilon)}(\Phi)$ is the maximum number $\log_2|M|$ of secret-key bits among all the $(n,|M|,\varepsilon)$ secret key agreement protocols over $\Phi$. The (asymptotic) \emph{secret-key capacity} is then defined in terms of the $n$-shot secret-key capacity as 
\bb
K(\Phi) \coloneqq \lim_{\varepsilon\rightarrow0^+} \liminf_{n \to \infty}\frac{K^{(\varepsilon,n)}(\Phi)}{n}\,.
\ee
Hence, the secret-key capacity of a quantum channel is the maximum achievable rate of secret-key bits generated via LOCCs with vanishing error in the asymptotic limit of infinitely many uses of the channel.

Due to the fact that a $(n,|M|,\varepsilon)$ LOCC-assisted quantum communication protocol is also a $(n,|M|,\varepsilon)$ protocol for secret key agreement, it follows that the $n$-shot two-way quantum capacity is always upper bounded by the $n$-shot secret-key capacity:
\bb\label{ineq_nshot_q_k}
    Q_2^{(n,\varepsilon)}(\Phi)\le K^{(n,\varepsilon)}(\Phi)\,.
\ee
The same holds for the asymptotic capacities:
\bb
    Q_2(\Phi)\le K(\Phi)\,.
\ee

\subsection{Entropic quantities}\label{subsec_def_entropic_quant}
In this section, we briefly define the entropic quantities that turns out to be useful for our analysis regarding non-asymptotic quantum communication. Further technical details regarding these quantities can be found in \cite{Sumeet_book}. 
\begin{Def}[von Neumann entropy]
    The von Neumann entropy of a state $\rho$ is defined as
    \bb
        S(\rho)\coloneqq -\Tr[\rho\log_2\rho]\,.
    \ee
\end{Def}
\begin{Def}[Quantum relative entropy]
    The {quantum relative entropy} between a state $\rho$ and a positive-semi-definite operator $\sigma$ is defined as
    \bb
        D(\rho\|\sigma)\coloneqq\Tr\left[\rho\left(\log_2\rho-\log_2\sigma\right)\right]\,.
    \ee
\end{Def}
\begin{Def}[Conditional entropy]
    The {conditional entropy} of a bipartite state $\rho_{AB}$ is defined as
    \bb
        S(A|B)_\rho \coloneqq S(\rho_{AB})-S(\rho_{B})=-D(\rho_{AB}\|\mathbb{1}_A\otimes\rho_{B} )\,.
    \ee
\end{Def}
\begin{Def}[Coherent information]\label{def_coh_info}
    The {coherent information} of a bipartite state $\rho_{AB}$ is defined as
    \bb
        I_c(A\,\rangle\, B)_{\rho_{AB}}\coloneqq -S(A|B)_{\rho_{AB}}= D(\rho_{AB}\|\mathbb{1}_A\otimes\rho_B)=S(\rho_B)-S(\rho_{AB})\,.
    \ee
\end{Def}
\begin{Def}[Petz Rényi relative entropy]\label{petz_renyi}
    For any $\alpha\in(0,1)\cup(1,2)$, the {Petz Rényi relative entropy} between a state $\rho$ and a positive-semi-definite operator $\sigma$ is defined as
    \bb
        D_\alpha(\rho\|\sigma)\coloneqq\frac{1}{\alpha-1}\log_2\!\left(\Tr\left[\rho^\alpha\sigma^{1-\alpha}\right]\right)\,.
    \ee
\end{Def}
\begin{Def}[Conditional Petz Rényi entropy]
    For any $\alpha\in(0,1)\cup(1,2)$, the {conditional Petz Rényi entropy} of a bipartite state $\rho_{AB}$ is defined as
    \bb
        H_\alpha(A|B)_{\rho_{AB}}\coloneqq-D_\alpha(\rho_{AB}\|\mathbb{1}_A\otimes\rho_B)\,.
    \ee
\end{Def}
Let us mention the following useful duality relation regarding conditional Petz Rényi entropy, proved in \cite[Lemma 6]{Tomamichel2008}. 
\begin{lemma}[Duality relation of the conditional Petz Rényi entropy]\label{duality_conditional_entropies}
    For any tripartite pure state $\Psi_{ABE}$ and any $\alpha\in(0,1)\cup(1,2)$, the conditional Petz Rényi entropy satisfies the following relation:
    \bb\label{eq_duality}
        H_{\alpha}(A|B)_{\Psi_{ABE}}=H_{2-\alpha}(A|E)_{\Psi_{ABE}}\,.
    \ee   
\end{lemma}
In \ref{eq_duality} we have simplified the notation by omitting partial traces in the subscript of the conditional entropy: the partial trace is performed over systems not involved in the conditional entropy, e.g.~
\bb
    H_{\alpha}(A|B)_{\Psi_{ABE}}\coloneqq H_{\alpha}(A|B)_{\Tr_E\Psi_{ABE}}\,.
\ee
This notation will be used throughout the manuscript.
\begin{Def}[Smooth max relative entropy]\label{def_max_smooth}
    Let $\varepsilon\in(0,1)$. The {smooth max relative entropy} between a state $\rho$ and a positive-semi-definite operator $\sigma$ is defined as
\bb
    D_{\max}^{\varepsilon}(\rho\|\sigma)\coloneqq \log_2\inf\left\{\lambda\,:\, \rho'\le\lambda\sigma,\, \rho'\ge 0,\,\Tr\rho'=1,\, P(\rho,\rho')\le\varepsilon\right\}\,,
\ee
where $P(\rho,\rho')$ denotes the {sine distance}, which is given by 
\bb\label{sine_distance}
    P(\rho,\rho')\coloneqq \sqrt{1-F(\rho,\rho')}\,,
\ee
with $F(\rho,\rho')\coloneqq \|\sqrt{\rho}\sqrt{\sigma}\|_1^2$ being the fidelity.
\end{Def}
\begin{Def}[Smooth conditional min entropy]\label{def_smooth_min_ent}
    Let $\varepsilon\in(0,1)$. The {smooth conditional min entropy} of a bipartite state $\rho_{AB}$ is defined as
    \bb
        H_{\min}^\varepsilon(A|B)_\rho\coloneqq -\min_{\sigma_B}D_{\max}^{\varepsilon}(\rho_{AB}\|\mathbb{1}_A\otimes\sigma_B)\,,
    \ee
    where the optimisation is over states $\sigma_B$.
\end{Def}
\begin{Def}[Smooth conditional max entropy]\label{smooth_max_conditional}
    Let $\varepsilon\in(0,1)$. The {smooth conditional max-entropy} of a bipartite state $\rho_{AB}$ is defined as
    \bb\label{def_H_min}
        H_{\max}^\varepsilon(A|B)_\rho\coloneqq \inf_{\tilde{\rho}:\, P(\rho,\tilde{\rho})\le\varepsilon}H_{\max}^0(A|B)_{\tilde{\rho}}\,,
    \ee
    where $H_{\max}^0(A|B)_{\tilde{\rho}}$ denotes the conditional max-entropy defined as
    \bb
        H^0_{\max}(A|B)_\rho\coloneqq  -H_{\min}^0(A|E)_{\psi_{ABE}}\,,
    \ee
    where $\psi_{ABE}$ is a purification of $\rho_{AB}$. The infimum in \eqref{def_H_min} is over states $\tilde{\rho}$ such that $P(\rho,\tilde{\rho})\le\varepsilon$, where $P(\rho,\tilde{\rho})$ denotes the sine distance defined in \eqref{sine_distance}.
\end{Def}
From the definitions, it follows that the {smooth conditional max-entropy} and the {smooth conditional min-entropy} are related by the {duality relation} stated in the following lemma.
\begin{lemma}\label{lemma_duality}
    For any pure tripartite state $\Psi_{ABE}$ and any $\varepsilon\in(0,1)$, it holds that
    \bb
        H_{\max}^\varepsilon(A|B)_{ \Psi_{ABE}}=-H_{\min}^\varepsilon(A|E)_{ \Psi_{ABE}}\,.
    \ee
\end{lemma}

\subsection{Continuous-variable systems}\label{sec_preliminaries_CV}
In this section, we provide an overview of quantum information with continuous variable systems~\cite{BUCCO}. By definition, a \emph{continuous variable system} is a quantum system associated with the Hilbert space $L^2(\mathbb R^n)$, which comprises all square-integrable complex-valued functions over $\mathbb{R}^n$, where $n\in\mathbb{N}$. Such an Hilbert space can be seen as the tensor product of $n$ Hilbert spaces equal to $L^2(\mathbb{R})$, each of which is called a \emph{mode}: 
\bb
    L^2(\mathbb R^n)=[L^2(\mathbb R)]^{\otimes n}\,.
\ee
Each mode constitutes a model to describe a single mode of electromagnetic radiation with definite frequency and polarisation. For each $i\in\{1,2,\ldots,n\}$, one can define the well-known \emph{position operator} $\hat{x}_i$ and \emph{momentum operator} $\hat{p}_i$ of the $i$-th mode~\cite{BUCCO}, satisfying the \emph{canonical commutation relation} $[\hat{x}_j,\hat{p}_k]=i\delta_{jk}\mathbb{1}$. By placing these operators into a vector, one obtains the \emph{quadrature vector}, which is defined as follows:
\bb
    \hat{\mathbf{R}}\coloneqq (\hat{x}_1,\hat{p}_1,\dots,\hat{x}_n,\hat{p}_n)^{\intercal}=(\hat{R}_1,\hat{R}_2,\dots,\hat{R}_{2n-1},\hat{R}_{2n})^{\intercal}\,
\ee
The canonical commutation relations can be expressed in terms of the quadrature vector as 
\bb
    [\hat{R}_k,\hat{R}_l]=i\,(\Omega )_{kl}\mathbb{\hat{1}}\qquad\forall\,k,l\in[2n]\,,
    \label{comm_rel_quadrature}
\ee
or in a more compact form as
\bb
[\hat{\mathbf{R}},\hat{\mathbf{R}}^{\intercal}]=i\,\Omega \mathbb{\hat{1}}\,,
\ee
where 
\bb\label{symplectic_form}
\Omega \coloneqq\bigoplus_{i=1}^n \left(\begin{matrix}0&1\\-1&0\end{matrix}\right) 
\ee
is called the \emph{$n$-mode symplectic form} and $\mathbb{\hat{1}}$ is the identity operator over $L^2(\mathbb R^n)$. The \emph{energy operator} $\hat{E}_n$ is defined as 
\bb\label{energy_operator}
\hat{E}_n&\coloneqq \frac{1}{2}\hat{\mathbf{R}}^\intercal\hat{\mathbf{R}}=\sum_{j=1}^n \!\left(\frac{\hat{x}_j^2}{2}+\frac{\hat{p}_j^2}{2}\right)=\sum_{j=1}^n \!\left(a_j^\dagger a_j +\frac{\mathbb{\hat{1}}}{2}\right)=\hat{N}_n+\frac{n}{2}\mathbb{\hat{1}}\,,
\ee
where we defined the \emph{total photon number operator} 
\bb\label{def_total_photon_number}
    \hat{N}_n\coloneqq \sum_{j=1}^n a_j^\dagger a_j
\ee
and the \emph{annihilation operator} $a_j$ of the $j$-th mode as 
\bb
    a_j\coloneqq \frac{\hat{x}_j+i\hat{p}_j}{\sqrt{2}}
\ee
for each $j\in\{1,2,\ldots,n\}$.
\subsubsection{Continuous-variable states}
Let us consider a single-mode system and let $a$ be its annihilation operator. The \emph{vacuum state} $\ket{0}$ is defined as the state satisfying $a\ket{0}=0$. Moreover, the $m$-th \emph{Fock state} is defined as
\begin{equation}
    \ket{m}\coloneqq\frac{(a^\dagger)^m}{\sqrt{m!}}\ket{0}\,,
\end{equation}
and it represents the quantum state with $n$ photons. Importantly, the Fock states $\{\ket{n}\}_{n\in\mathbb{N}}$ form a basis for the single-mode Hilbert space $L^2(\mathbb{R})$, meaning that the latter can be effectively seen as an infinite-dimensional qudit. The (single-mode) photon number operator con be diagonalised in terms of Fock states as follows:
\bb
    a^\dagger a=\sum_{m=0}^\infty m \ketbra{m}\,.
\ee
A \emph{single-mode state} is a quantum state on $L^2(\mathbb R)$, while an \emph{$n$-mode state} is a quantum state on $L^2(\mathbb R^n)$. The \emph{mean energy} of an $n$-mode state is defined by 
\bb\label{mean_energy_rho}
    E(\rho)\coloneqq \Tr[\hat{E}_n\rho]\,,
\ee
where $\hat{E}_n$ is the energy operator defined in \eqref{energy_operator}. Moreover, the \emph{mean photon number} of an $n$-mode state is defined as 
\bb\label{mean_phot_rho}
    N(\rho)\coloneqq\Tr[\hat{N}_n\rho]\,,
\ee
where $\hat{N}_n$ is the total photon number operator defined in \eqref{def_total_photon_number}.

Given an $n$-mode state $\rho$, one can define its \emph{characteristic function} $\chi_\rho: \mathbb{R}^{2n}\to \mathbb{C}$ as $\chi_\rho(\mathbf r)\coloneqq\Tr[ \rho  \hat{D}_{\mathbf{r}} ]$, where for all $\mathbf{r}\in \mathbb{R}^{2n}$ we introduced the \emph{displacement operator} $\hat{D}_{\mathbf{r}}\coloneqq e^{-i {\mathbf{r}}^{\intercal}\Omega  \hat{\mathbf{R}}}$. Conversely, given the characteristic function  $\chi_\rho$, one can reconstruct the quantum state via the so-called \emph{Fourier-Weyl relation}:
\bb\label{inverse_fourier_displacement}
\rho=\int_{\mathbb{R}^{2n}}\frac{\mathrm{d}^{2n}\mathbf{r}}{(2\pi)^n}\chi_\rho(\mathbf r) \hat{D}_{\mathbf{r}}^\dagger\,.
\ee
The Wigner function $ W_\rho: \mathbb{R}^{2n} \to \mathbb{R}$ of an \( n \)-mode state \( \rho \) is defined as the inverse (symplectic) Fourier transform of the characteristic function:
\begin{equation}
W_\rho(\mathbf{r})  = \frac{1}{(2\pi)^{2n}} \int_{\mathbf{r}' \in \mathbb{R}^{2n}} \mathrm{d}^{2n}\mathbf{r}' \, \chi_\rho(\mathbf{r}') e^{i {\mathbf{r}'}^{\intercal} \Omega  \mathbf{r}}.
\end{equation}
As a consequence, the characteristic function can be expressed as the (symplectic) Fourier transform of the characteristic function:
 \begin{equation}\label{Fourier_Weyl}
\chi_\rho(\mathbf{r})  = \int_{\mathbf{r}' \in \mathbb{R}^{2n}} \mathrm{d}^{2n}\mathbf{r}' \, W_\rho(\mathbf{r}') e^{-i {\mathbf{r}}^{\intercal} \Omega  \mathbf{r}'}.
\end{equation}
The \emph{first moment} of a quantum state $\rho$ is defined as $\mathbf{m}(\rho)=\Tr\!\left[\hat{\mathbf{R}}\,\rho\right]$.
Moreover, the \emph{covariance matrix} of $\rho$ is a matrix denoted as $V\!(\rho)$ and defined by
\bb
	[V\!(\rho)]_{k,l}\coloneqq \Tr\!\left[\left\{{\hat{R}_k-m_k(\rho)\mathbb{1},\hat{R}_l-m_l(\rho)\mathbb{1}}\right\}\rho\right]\qquad\forall\in\{1,2,\ldots,2n\}\,,
\ee
where $\{\hat{A},\hat{B}\}\coloneqq \hat{A}\hat{B}+\hat{B}\hat{A}$ is the anti-commutator. In a more compact form, the covariance matrix can be written as
\bb
	V\!(\rho)&=\Tr\!\left[\left\{(\hat{\textbf{R}}-\textbf{m}(\rho)\mathbb{1}),(\hat{\textbf{R}}-\textbf{m}(\rho))^{\intercal}\mathbb{1}\right\}\rho\right]\, .
\ee   
Any covariance matrix $V\!(\rho)$ satisfies the following operator inequality~\cite{BUCCO}:
\bb
V\!(\rho)+i\Omega \ge0\,,
\ee
known as \emph{uncertainty relation}. Consequently, since $\Omega $ is skew-symmetric, any covariance matrix $V\!(\rho)$ is positive semi-definite on $\mathbb{R}^{2n}$. 
\subsubsection{Gaussian states}
\begin{Def}[(Gaussian state)]\label{def_gauss_sm}
A Gaussian state is a state that can be written as a tensor product of Gibbs states of quadratic Hamiltonians. By definition, a Gibbs state of a quadratic Hamiltonian $\hat{H}$ is a state of the form
\begin{equation}
   \rho=  \frac{e^{-\beta \hat{H}}}{\Tr[e^{-\beta \hat{H}}]} \,,
\end{equation}
for some $\beta\in(0,\infty]$, where
\bb
    \hat{H}\coloneqq \frac{1}{2}(\hat{\mathbf{R}}-\mathbf{m})^{\intercal}H(\hat{\mathbf{R}}-\mathbf{m})\,,
\ee
with $\hat{\textbf{R}}$ being the quadrature vector, $H$ being a real symmetric positive-definite matrix, and $\mathbf{m}$ being a real vector.
\end{Def}
Equivalently, the set of Gaussian states can be defined as the closure, with respect to the trace norm, of the set of Gibbs states of quadratic Hamiltonians~\cite{G-resource-theories}.

The Wigner function of a Gaussian state $\rho$ is a Gaussian probability distribution with first moment $\textbf{m}(\rho)$ and covariance matrix $\frac{V(\rho)}{2}$~\cite{BUCCO}:
\bb
    W_\rho(\textbf{r})=\NN\!\left[\textbf{m}(\rho),\frac{V(\rho)}{2}\right]\!(\textbf{r})=\frac{e^{- (\textbf{r}-\textbf{m}(\rho))^\intercal [V(\rho) ]^{-1}  (\textbf{r}-\textbf{m}(\rho))}}{\pi^{n}\sqrt{\det [V(\rho) ]}}\,.
\ee
where 
\bb
    \NN[\textbf{m},V](\textbf{r})&\coloneqq\frac{e^{-\frac12 (\textbf{r}-\textbf{m})^\intercal V^{-1}  (\textbf{r}-\textbf{m})}}{(2\pi)^{n}\sqrt{\det V }}\,.
\ee
In addition, the characteristic function of a Gaussian state $\rho$ is given by~\cite{BUCCO}
\bb\label{charact_gaussian}
\chi_{\rho}(\mathbf{r})=\exp\!\left( -\frac{1}{4}(\Omega  \mathbf{r})^{\intercal}V\!(\rho)\Omega  \mathbf{r}+i(\Omega  \mathbf{r})^{\intercal}\mathbf{m}(\rho) \right)\,.
\ee
Since any quantum state can be reconstructed from its characteristic function via the Fourier-Weyl relation in \eqref{Fourier_Weyl}, it follows that any Gaussian state is uniquely identified by its first moment and covariance matrix.

An important example of Gaussian state is the so-called \emph{thermal state}. Given ${N_s}\ge0$, the thermal state $\tau_{N_s}$ is defined as
\bb
    \tau_{{N_s}}\coloneqq \frac{1}{{N_s}+1}\sum_{n=0}^\infty \left(\frac{{N_s}}{{N_s}+1}\right)^{n}\ketbra{n}\,.
    \label{eq:termal}
\ee
Note that the mean photon number of $\tau_{{N_s}}$ is given by $\Tr[a^\dagger a\,\tau_{{N_s}} ]={N_s}$. The first moment and the covariance matrix of the thermal state are given by:
\bb\label{moments_thermal}
\mathbf{m}(\tau_{{N_s}})&=(0,0)^{\intercal}\,,\\
V(\tau_{{N_s}})&=(2{N_s}+1)\mathbb{1}_2 \,.
\ee
Note that for ${N_s}=0$, the thermal state coincides with the vacuum state $\ketbra{0}$. Another important example of Gaussian state is the so-called $\emph{two-mode squeezed vacuum state}$. Given $N_s\ge0$, the two-mode squeezed vacuum state is a two-mode state defined as
\bb
    \ket{\psi_{{N_s}}}_{AB}\coloneqq \frac{1}{\sqrt{{N_s}+1}}\sum_{n=0}^\infty \left(\frac{{N_s}}{{N_s}+1}\right)^{n/2}\ket{n}_A\otimes\ket{n}_B\,.
    \label{eq:2mode_squeezed}
\ee
The first moment and covariance matrix of the two-mode squeezed vacuum are given by~\cite{BUCCO}:
\bb\label{cov_two_mode_squeezed}
\mathbf{m}(\tau_{{N_s}})&=(0,0,0,0)^{\intercal}\,,\\
V(\ketbra{\Psi_{N_s}})&=\left(\begin{matrix} (2N_s+1)\mathbb{1}_2 & 2\sqrt{N_s(N_s+1)}\sigma_z \\ 2\sqrt{N_s(N_s+1)}\sigma_z  &(2N_s+1)\mathbb{1}_2\end{matrix}\right)\,,
\ee 
where $\sigma_z\coloneqq\left(\begin{matrix}1&0\\0&-1\end{matrix}\right)$. Note that the two-mode squeezed vacuum state $\ket{\psi_{{N_s}}}$ is a purification of the thermal state $\tau_{N_s}$, that is $\Tr_{B}\ketbra{\psi_{{N_s}}}_{AB}=\tau_{N_s}$.

The von Neumann entropy of a thermal state $\tau_{N_s}$ is given by~\cite{BUCCO}
\bb\label{entropy_thermal}
    S(\tau_{N_s})=h(N_s)\,,
\ee
where $h(x)\coloneqq (x+1)\log_2(x+1) - x\log_2 x\,$. As proved in \cite{WildeRenes2016}, the \emph{entropy variance}~\cite{WildeRenes2016} of a thermal state $\tau_{N_s}$ reads
\bb\label{ent_var_th_state}
    \Tr[\tau_{N_s}(\log_2 \tau_{N_s})^2]-\Tr[\tau_{N_s}\log_2 \tau_{N_s}]^2&=N_s(N_s+1)\log_2\left(1+\frac{1}{N_s}\right)^2\,.
\ee

\subsubsection{Gaussian unitaries}
A \emph{symplectic matrix} $S\in\mathbb{R}^{2n,2n}$ is a matrix satisfying $S\Omega S^\intercal=\Omega $. Given a symplectic matrix $S\in\mathbb{R}^{2n,2n}$, one can define a suitable $n$-mode unitary $U_S$ --- dubbed \emph{symplectic unitary} --- such that~\cite{BUCCO}
\begin{equation}
U_S^\dagger \hat{R}_k U_S=\sum^{2n}_{k,l}S_{k,l}\hat{R}_l\qquad\forall k\in\{1,2,\ldots,n\}\,,
\end{equation}
In a more compact form, this can be written as $U_S^\dagger \hat{\mathbf{R}}U_S=S\hat{\mathbf{R}}$. In particular, a symplectic unitary transforms the first moment and the covariance matrix as follows:
\bb\label{law_sympl_unitary}
    &\mathbf{m}(U_S\rho U_S^\dagger)=S\mathbf{m}(\rho)\,,\\
    &V(U_S\rho U_S^\dagger)=SV\!(\rho)S^\intercal\, .
\ee
Similarly, the displacement operator transforms the quadrature vector as 
\begin{align}
    \label{eq:displ}\hat{D}_\mathbf{r}\hat{\mathbf{R}}\hat{D}_\mathbf{r}^\dagger=\hat{\mathbf{R}}+\mathbf{r}\mathbb{\hat{1}},
\end{align} 
and, in particular, a displacement operator  transforms the first moment and the covariance matrix as follows:
\bb
    &\mathbf{m}\!\left(\hat{D}_\mathbf{r}\rho \hat{D}_\mathbf{r}^\dagger\right)=\mathbf{m}(\rho)+\mathbf{r}\,,\\
    &V\!\left(\hat{D}_\mathbf{r}\rho \hat{D}_\mathbf{r}^\dagger\right)=V\!(\rho)\,.
\label{eq:bille}
\ee
\begin{Def}[(Gaussian unitary)]\label{def_gauss_sm2}
    A Gaussian unitary is a unitary given by the product of a displacement operator and a symplectic unitary.
\end{Def}
Any Gaussian unitary maps the set of Gaussian states into itself~\cite{BUCCO}. 
\subsubsection{Beam splitter unitary}
An important example of Gaussian unitary is the so-called \emph{beam splitter unitary}.
\begin{Def}[(Beam splitter unitary)]\label{def_beam_splitter}
Let $S$ and $E$ be two single-mode systems. Let $a$ and $b$ be the annihilation operators of $S$ and $E$, respectively. For all $\lambda\in[0,1]$, the beam splitter unitary of transmissivity $\lambda$ is defined by
\bb
    U_{\lambda}^{S E}\coloneqq\exp\!\left[\arccos\sqrt{\lambda}\left(a^\dagger b-a\, b^\dagger\right)\right]\,.
\ee
\end{Def}
The symplectic matrix $S_\lambda$ associated with the beam splitter unitary $U^{SE}_{\lambda}$ satisfying the relation $(U^{SE}_{\lambda})^\dagger \hat{\mathbf{R}}\,U^{SE}_{\lambda}=S_\lambda\hat{\mathbf{R}}$ is given by~\cite{BUCCO}
\bb\label{symplectic_beam_splitter}
    S_\lambda &\coloneqq \left(\begin{matrix} \sqrt{\lambda}\,\mathbb{1}_2  & \sqrt{1-\lambda}\,\mathbb{1}_2 \\
-\sqrt{1-\lambda}\,\mathbb{1}_2 & \sqrt{\lambda}\,\mathbb{1}_2 \end{matrix}\right)\,.
\ee
The beam splitter unitary $U^{SE}_{\lambda}$ acts acts on a Fock state $\ket{i}_S\otimes\ket{j}_E$ as (see e.g.~\cite{Die-Hard-2-PRA})
\bb\label{eq_beam_ij}
U_{\lambda}\ket{i}_S\otimes\ket{j}_E=\sum_{m=0}^{i+j}c_m^{(i,j)}\ket{i+j-m}_{S}\otimes\ket{m}_{E} \,,
\ee
where
\bb
    c_m^{(i,j)}\coloneqq\frac{1}{\sqrt{i!j!}}\sum_{k=\max(0,m-j)}^{\min(i,m)}(-1)^k\sqrt{m!(i+j-m)!}\binom{i}{k}\binom{j}{m-k}\lambda^{\frac{i+m-2k}{2}}(1-\lambda)^{\frac{j+2k-m}{2}}\,.
\ee
In particular, it holds that 
\bb\label{eq_beam_n0}
    U_\lambda^{SE}\ket{n}_S\otimes\ket{0}_E&=
\sum_{l=0}^n(-1)^l\sqrt{\binom{n}{l}}\lambda^{\frac{n-l}{2}}(1-\lambda)^{\frac{l}{2}}\ket{n-l}_{S}\otimes\ket{l}_E\,,\\
    U_\lambda^{SE}\ket{0}_S\otimes\ket{n}_E&=
\sum_{l=0}^n \sqrt{\binom{n}{l}}(1-\lambda)^{\frac{l}{2}}\lambda^{\frac{n-l}{2}}\ket{l}_{S}\otimes\ket{n-l}_E\,.
\ee

\subsubsection{Two-mode squeezing unitary}
Another important example of Gaussian unitary is the so-called \emph{two-mode squeezing unitary}.
\begin{Def}[(Two-mode squeezing unitary)]\label{def_two_sq}
Let $S$ and $E$ be two single-mode systems. Let $a$ and $b$ be the annihilation operators of $S$ and $E$, respectively. For all $g\ge1$, the two-mode squeezing unitary of parameter $g$ is defined by
\bb
    U_{g}^{S E}\coloneqq\exp\left[\arccosh\sqrt{g}\left(a^\dagger b^\dagger-a\, b\right)\right]\,.
\ee
\end{Def}
The symplectic matrix $S_g$ associated with the two-mode squeezing unitary $U^{SE}_{g}$ satisfying the relation $(U^{SE}_{g})^\dagger \hat{\mathbf{R}}\,U^{SE}_{g}=S_g\hat{\mathbf{R}}$ is given by~\cite{BUCCO}
\bb\label{symplectic_squez}
S_g\coloneqq	\begin{pmatrix}
		\sqrt{g}\,\mathbb{1}_2\, & \,\sqrt{g-1}\,\sigma_z \\
		\sqrt{g-1}\,\sigma_z\, &\, \sqrt{g}\,\mathbb{1}_2
	\end{pmatrix}\,.
\ee

\subsubsection{Williamson decomposition}
The covariance matrix $V\!(\rho)$ of an $n$-mode state $\rho$ satisfies the so-called \emph{Williamson decomposition}~\cite{BUCCO}: there exist $n$ real numbers $d_1,d_2,\ldots , d_n\ge1$, known as the \emph{symplectic eigenvalues} of $V\!(\rho)$, and a symplectic matrix $S$ such that
\bb
    V\!(\rho)=S\,\operatorname{diag}\left(d_1,d_1,d_2,d_2,\ldots,d_n,d_n\right)\,S^\intercal\,,
    \label{eq:will}\,.
\ee
Moreover, if $\rho$ is a Gaussian state, one can show that the above decomposition of the covariance matrix implies the following decomposition of $\rho$, known as \emph{normal decomposition}:~\cite{BUCCO} 
\bb
    \rho= \hat{D}_{\mathbf{m}(\rho)}U_S\left(\tau_{\nu_1}\otimes\tau_{\nu_2}\otimes\ldots\tau_{\nu_n}\right)U_S^\dagger \hat{D}_{\mathbf{m}(\rho)}^\dagger\,,
    \label{eq:GaussianState}
\ee
where $\{\tau_{\nu_k}\}^n_{k=1}$ are thermal states defined in~\eqref{eq:termal}, and the mean photon numbers $\nu_1,\nu_2,\ldots,\nu_n$ are defined in terms of the symplectic eigenvalues of $V\!(\rho)$ as $\nu_i\coloneqq \frac{d_i-1}{2}$.

The above decomposition allows us to easily show the following useful lemma.
\begin{lemma}\label{finite_sqrt}
For any Gaussian state $\rho$, the quantity $\Tr\sqrt{\rho}$ is finite.
\end{lemma}
\begin{proof}
Thanks to the decomposition of normal decomposition in \eqref{eq:GaussianState}, any Gaussian state $\rho$ is unitarily equivalent to a tensor product of thermal states. Hence, the trace of the square root of a Gaussian state is equal to the product of traces of the square root of thermal states. Moreover, the trace of the square root of a thermal state $\tau_N$ is finite, as
\bb
     \sqrt{\tau_{{N}}}\coloneqq \frac{1}{\sqrt{{N}+1}}\sum_{n=0}^\infty \left(\frac{{N}}{{N}+1}\right)^{n/2}\ketbra{n}\,.
\ee
and hence
\bb
    \Tr[\sqrt{\tau_N}]=\frac{1}{N+1}\sum_{n=0}^\infty\left(\frac{N}{N+1}\right)^{n/2}<\infty\,.
\ee
This concludes the proof.
\end{proof}

\subsubsection{Gaussian channels}
\begin{Def}[(Gaussian channel)] 
A Gaussian channel is a quantum channel that maps the set of Gaussian states into itself.
\end{Def}
An \emph{$n$-mode Gaussian channel} is a Gaussian channel from the set of $n$-mode states into itself. Similarly, a \emph{single-mode Gaussian channel} is a Gaussian channel from the set of single-mode states into itself. It is worth to mention that any Gaussian channel can be written in Stinespring representation in terms of a Gaussian unitary and an environmental vacuum state~\cite{BUCCO}.

An important example of single-mode Gaussian channel is the so-called \emph{pure loss channel}, which models the phenomenon of photon loss in optical fibres and free-space links. The pure loss channel is mathematically defined in the forthcoming Definition~\ref{def_pure_loss}, with a pictorial representation provided in Fig.~\ref{fig_pure_loss_sm}.
\begin{Def}[(Pure loss channel)]\label{def_pure_loss}
Let $S$ and $E$ be two single-mode systems.  The pure loss channel of transmissivity $\lambda$ is a quantum channel $\pazocal{E}_{\lambda}:S\to S$ defined as follows:
\bb
        \pazocal{E}_{\lambda}(\Theta)&\coloneqq\Tr_E\left[U_\lambda^{SE} \big(\Theta_S \otimes\ketbra{0}_E\big) (U_\lambda^{SE})^\dagger\right] 
\ee
for any linear operator $\Theta$ on $S$, where $U_\lambda^{SE}$ denotes the beam splitter unitary of transmissivity $\lambda$ and $\ketbra{0}_E$ denotes the vacuum state of $E$.
\end{Def}
\begin{figure}[h!]
	\centering
	\includegraphics[width=0.3\linewidth]{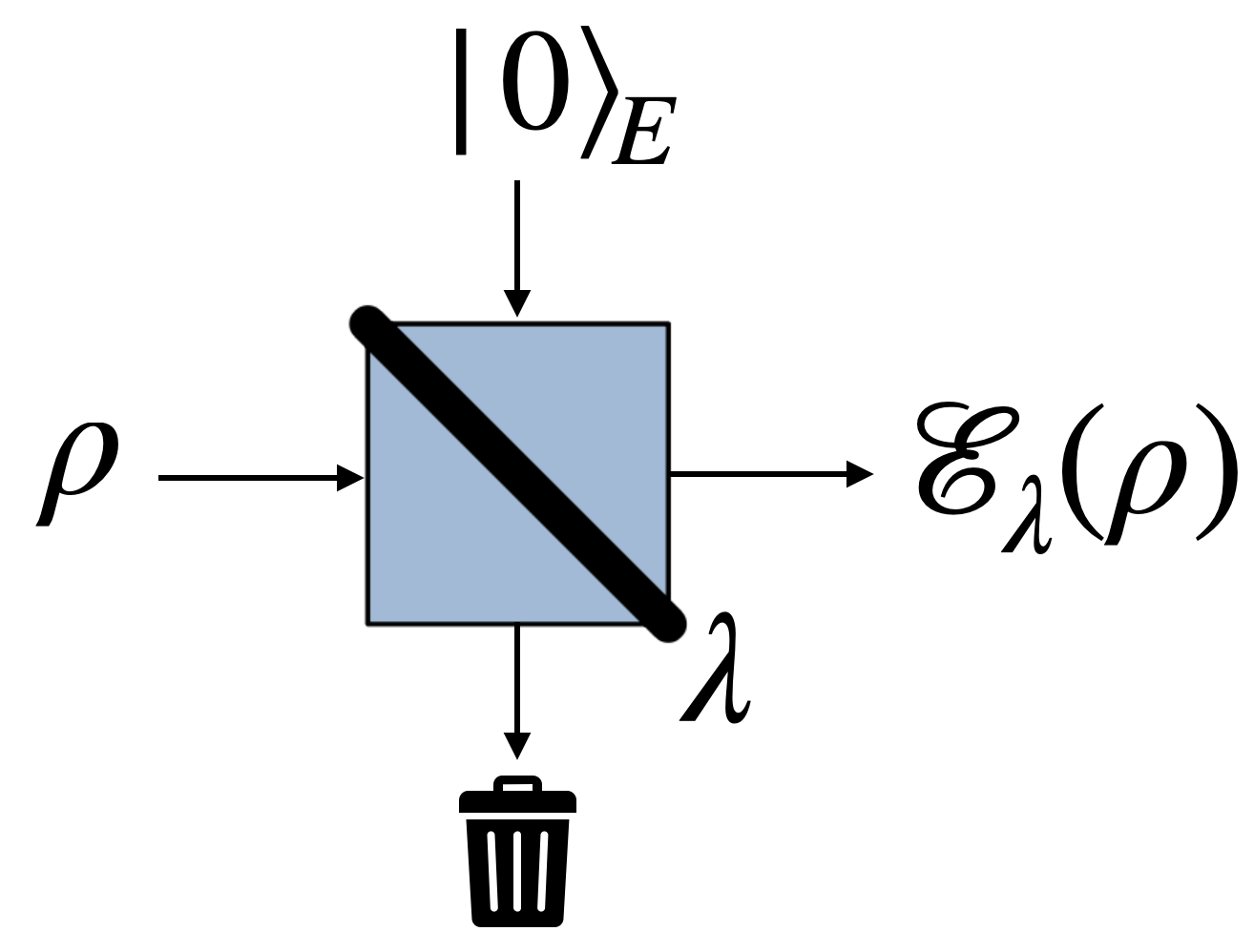}
	\caption{Pictorial representation of the pure loss channel. The pure loss channel $\mathcal{E}_\lambda$ of transmissivity $\lambda$ acts on the input state $\rho$ by mixing it in a beam splitter (i.e.~a semi-reflective mirror) of transmissivity $\lambda$ with an environment initialised in the vacuum $\ket{0}_E$. The pure loss channel is noiseless for $\lambda=1$ (it equals the identity channel), while it is completely noisy for $\lambda=0$ (it maps any input state into the vacuum).}
    \label{fig_pure_loss_sm}
\end{figure}
Another notable example of a single-mode Gaussian channel is the \emph{pure amplifier channel}, which serves as a model for phenomena such as spontaneous parametric down-conversion~\cite{Clerk_2010}, the dynamical Casimir effect in superconducting circuits~\cite{Moore1970}, the Unruh effect~\cite{Unruh1976}, and Hawking radiation~\cite{Hawking1972}.
\begin{Def}[(Pure amplifier channel)]\label{def_pure_ampl}
Let $S$ and $E$ be two single-mode systems.  The pure amplifier channel of gain $g$ is a quantum channel $\Phi_{g}:S\to S$ defined as follows:
\bb
        \Phi_{g}(\Theta)&\coloneqq\Tr_E\left[U_g^{SE} \big(\Theta_S \otimes\ketbra{0}_E\big) (U_g^{SE})^\dagger\right] 
\ee
for any linear operator $\Theta$ on $S$, where $U_g^{SE}$ denotes the two-mode squeezing unitary of parameter $g$ and $\ketbra{0}_E$ denotes the vacuum state of $E$.
\end{Def}

\subsubsection{Capacities of the pure loss channel and pure amplifier channel}
For all $\lambda\in[0,1]$ the quantum capacity $Q(\mathcal{E}_{\lambda})$~\cite{holwer, Wolf2007, Wolf2006}, the two-way quantum capacity $Q_2(\mathcal{E}_{\lambda})$~\cite{PLOB}, and the secret-key capacity $K(\mathcal{E}_{\lambda})$~\cite{PLOB} of the pure loss channel $\mathcal{E}_{\lambda}$ are given by:
\bb\label{capacities_pure_loss}
    Q(\mathcal{E}_{\lambda})&=   
        \begin{cases}
        \log_2\!\left(\frac{\lambda}{1-\lambda}\right) &\text{if $\lambda\in(\frac{1}{2},1]$ ,} \\
        0 &\text{if $\lambda\in[0,\frac{1}{2}]$ ,}
    \end{cases}\\
    Q_2(\mathcal{E}_{\lambda})&=K(\mathcal{E}_{\lambda})= \log_2\!\left(\frac{1}{1-\lambda}\right)\,.
\ee
Hence, the quantum capacity vanishes for $\lambda\in[0,\frac{1}{2}]$, while the two-way quantum capacity and secret-key capacity are always larger than zero for all $\lambda\in(0,1]$.

Moreover, for all $\Phi_g\ge1$ the quantum capacity $Q(\Phi_g)$~\cite{holwer, Wolf2007, Wolf2006}, the two-way quantum capacity $Q_2(\Phi_g)$~\cite{PLOB}, and the secret-key capacity $K(\Phi_{g})$~\cite{PLOB} of the pure amplifier channel $\Phi_{g}$ are all equal to each other and are given by:
\bb\label{capacities_pure_ampl}
    Q(\Phi_{g})=Q_2(\mathcal{E}_{\lambda})&=K(\mathcal{E}_{\lambda})= \log_2\!\left(\frac{g}{g-1}\right)\,.
\ee  
The above capacities can be calculated via coherent information quantities. Namely, for $\lambda\in[\frac{1}{2},1]$ the quantum capacity $Q(\mathcal{E}_{\lambda})$ of the pure loss channel is given by~\cite{holwer, Wolf2007, Wolf2006}:
\bb\label{coh_1}
Q(\mathcal{E}_{\lambda})=\lim\limits_{N_s\rightarrow\infty}I_c(A\,\rangle\, B)_{\mathbb{1}_{A}\otimes\mathcal{E}_{\lambda}(\ketbra{\Psi_{N_s}}_{AB})} =\log_2\!\left(\frac{\lambda}{1-\lambda}\right)\,,
\ee
where $\ketbra{\Psi_{N_s}}_{AB}$ denotes the two-mode squeezed vacuum state defined in Eq.~\eqref{eq:2mode_squeezed} and $I_c(A\,\rangle\, B)$ denotes the coherent information defined in Definition~\ref{def_coh_info}. Moreover,  the two-way quantum capacity and the secret-key capacity are given by~\cite{PLOB}:
\bb\label{coh_2}
    Q_2(\mathcal{E}_{\lambda})=K(\mathcal{E}_{\lambda})= \lim\limits_{N_s\rightarrow\infty}I_c(B\,\rangle\, A)_{\mathbb{1}_{A}\otimes\mathcal{E}_{\lambda}(\ketbra{\Psi_{N_s}}_{AB})}=\log_2\!\left(\frac{1}{1-\lambda}\right)\,.
\ee
Similarly, the capacities of the pure amplifier channel can be calculated via coherent information quantities as
\bb\label{coh_ampl}         Q_2(\Phi_{g})=Q_2(\Phi_{g})=K(\Phi_{g})= \lim\limits_{N_s\rightarrow\infty}I_c(A\,\rangle\, B)_{\mathbb{1}_{A}\otimes\Phi_{g}(\ketbra{\Psi_{N_s}}_{AB})}=\log_2\!\left(\frac{g}{g-1}\right)\,.
\ee
An upper bound on the $n$-shot capacities of the pure loss channel and of the pure amplifier channel is given by~\cite{MMMM}
\bb
    Q_2^{(\varepsilon,n)}(\mathcal{E}_\lambda)&\le K^{(\varepsilon,n)}(\mathcal{E}_\lambda)\le n\log_2\!\left(\frac{1}{1-\lambda}\right)+\log_26+2\log_2\!\left(\frac{1+\varepsilon}{1-\varepsilon}\right)\,,\\
    Q_2^{(\varepsilon,n)}(\Phi_g)&\le K^{(\varepsilon,n)}(\Phi_g)\le n\log_2\!\left(\frac{g}{g-1}\right)+\log_26+2\log_2\!\left(\frac{1+\varepsilon}{1-\varepsilon}\right)\,,
\ee
while a lower bound is proved in the present paper; specifically in Section~\ref{sec_lower_bound_pure_loss_channel} for the pure loss channel and in Section~\ref{sec_lower_bound_pure_ampl_channel} for the pure amplifier channel.
\subsubsection{Energy-constrained capacities}
In practise, a quantum communication protocol can not exploit \emph{infinite} energy. Instead, all optical input signals used in a communication protocol have a bounded energy (bounded by e.g.~the energy of the Sun, or the energy budget available in lab). It is therefore important to consider only those quantum communication protocols that exploit input states satisfying a suitable \emph{energy constraint}. Specifically, in the continuous-variable setting, it is common to consider the so-called \emph{energy-constrained capacities}~\cite{Davis2018,Mark-energy-constrained}. Given a positive number $N_s$ (representing the energy budget per input signal), the energy-constrained capacities are defined in the same way as the unconstrained capacities, apart from the fact that the optimisation is performed over those protocols such that the average expected value of the photon number operator $a^\dagger a$ on all input signals is required to be at most $N_s$. The energy-constrained quantum capacities of a quantum channel $\Phi$ are denoted as $Q(\Phi,N_s)$, $Q_2(\Phi,N_s)$, and $K(\Phi,N_s)$, while the $n$-shot energy-constrained capacities are denoted as $Q^{(\varepsilon,n)}(\Phi,N_s)$, $Q_2^{(\varepsilon,n)}(\Phi,N_s)$, and $K^{(\varepsilon,n)}(\Phi,N_s)$.

The energy-constrained quantum capacity of the pure loss channel $\mathcal{E}_\lambda$ and of the pure amplifier channel $\Phi_g$ have been determined exactly~\cite{holwer, Caruso2006, Wolf2007, Mark2012, Mark-energy-constrained, Noh2019} and they are given by the coherent information of the bipartite state obtained by sending an half of the two-mode squeezed vacuum state $\ketbra{\Psi_{N_s}}$ with local energy $N_s$ into the channel:
\bb\label{q_ec_pure_loss}
    Q\left(\mathcal{E}_\lambda,N_s\right)&=\max\{0, I_c(A\,\rangle\, B)_{\mathbb{1}_{A}\otimes\mathcal{E}_{\lambda}(\ketbra{\Psi_{N_s}}_{AB})} \}&=   
        \begin{cases}
        h(\lambda N_s)-h((1-\lambda)N_s) &\text{if $\lambda\in(\frac{1}{2},1]$ ,} \\
        0 &\text{if $\lambda\in[0,\frac{1}{2}]$ ,}
    \end{cases} 
\ee
\bb\label{q_ec_pure_ampl}
    Q\left(\Phi_g,N_s\right)&= I_c(A\,\rangle\, B)_{\mathbb{1}_{A}\otimes\Phi_{g}(\ketbra{\Psi_{N_s}}_{AB})}=h\!\left( gN_s + g-1\right)-h\!\left( (g-1)(N_s+1) \right)\,,
\ee
where $h(x)\coloneqq (x+1)\log_2(x+1)-x\log_2 x$.

The energy-constrained two-way quantum and secret-key capacities of the pure loss channel and pure amplifier channel have not yet been determined. However, a lower bound is given by~\cite{Pirandola2009}:
\bb\label{q2_ec_pure_loss}
    K\left(\mathcal{E}_\lambda,N_s\right)&\ge Q_2\left(\mathcal{E}_\lambda,N_s\right)\ge I_c(B\,\rangle\, A)_{\mathbb{1}_{A}\otimes\mathcal{E}_{\lambda}(\ketbra{\Psi_{N_s}}_{AB})}=h(N_s)-h\!\left((1-\lambda)N_s\right)\,,\\
    K\left(\Phi_g,N_s\right)&\ge Q_2\left(\Phi_g,N_s\right)\ge Q\left(\Phi_g,N_s\right)=I_c(A\,\rangle\, B)_{\mathbb{1}_{A}\otimes\Phi_{g}(\ketbra{\Psi_{N_s}}_{AB})}=h\!\left( gN_s + g-1\right)-h\!\left( (g-1)(N_s+1) \right)\,.
\ee 
In this work, we will establish a lower bound on the energy-constrained $n$-shot capacities of the pure loss channel and pure amplifier channel.

\newpage 
\section{Tail bounds on Gaussian states}
In this section, we solve the following simple-looking problems:
\begin{itemize}
    \item What is the probability that a fixed Gaussian state has more than $M$ photons? In other words, given a Gaussian state $\rho$, how to estimate the probability that the outcome of the photon number measurement on $\rho$ is larger than $M$? We answer this question in subsection~\eqref{tail_1}. In particular, we show that such a probability converges to zero exponentially fast in $M$, with exponential lifetime given by the mean photon number of the state.
    \item What is the trace-distance error in approximating a Gaussian state with its finite-dimensional approximation consisting of at most $M$ photons? Equivalently, given a Gaussian state $\rho$, how to estimate the trace distance between $\rho$ and its projection $\rho_M$ onto the finite-dimensional subspace with fixed maximum number of photons $M$? We answer this question in subsection~\eqref{tail_2}. Specifically, we prove that such an error converges to zero exponentially fast in $M$, with exponential lifetime given by the mean photon number of the state. 
\end{itemize}
To solve both problems, we determine explicit bounds which can be written only in terms of the first moment and covariance matrix of the state. We believe that these bounds provide technical tools of independent interest, which can be applied in all areas of quantum information with continuous variable systems. Indeed, by exploiting these tools, we design an efficient algorithm to compute the trace distance between two Gaussian states up to precision $\varepsilon$ in trace distance, which is also of independent interest. Finally, these new tools will turn out to be crucial in our analysis of one-shot bosonic quantum communication in Section~\ref{section_AEP0}.

\subsection{Probability that a fixed Gaussian state has more than $M$ photons}\label{tail_1}
In this subsection, we solve the following problem. Given a Gaussian state $\rho$ specified by its first moment $\textbf{m}$ and covariance matrix $V$, how to estimate the probability $P_{>M}$ that the outcome of the photon number measurement on $\rho$ is larger than $M$, as a function of $\textbf{m}$, $V$, and $M$? Estimating such a probability $P_{> M}$ is important, as the photon number measurement is a key element in the design of quantum technologies based on continuous-variable systems~\cite{BUCCO}.

By Born's rule, such a probability $P_{> M}$ is given by \bb
    P_{> M}\coloneqq\Tr[(\mathbb{1}-\Pi_M)\rho]\,,
\ee
where $\Pi_M$ is the projector onto the Hilbert space spanned by all the $n$-mode Fock states with photon number less or equal to $M$:
\bb
    \Pi_M\coloneqq\sum_{\textbf{k}\in\mathbb{N}^n:\,\sum_{i=1}^n k_i\le M}\ketbra{\textbf{k}}\,,
\ee
with $\ket{\textbf{k}}=\ket{k_1}\otimes\ket{k_2}\otimes\ldots\otimes\ket{k_n}$ being the $n$-mode Fock state. The following theorem allows us to upper bound the probability $P_{>M}$, establishing that it converges to zero exponentially fast in $M$.
\begin{thm}[(Probability that a fixed Gaussian state has more than $M$ photons)]\label{the:expdecay00}
Let $\rho$ be an $n$-mode Gaussian state with first moment $\textbf{m}$ and covariance matrix $V$. The probability $P_{>M}$ that the outcome of the photon number measurement on $\rho$ is larger than $M$ converges to zero exponentially fast in $M$. Specifically, it can be upper bounded as
\bb
    P_{>M}\le \alpha\,e^{-M /(4N+2) }\,,
\ee
where $N$ denotes the mean photon number given by $N=\frac{\Tr[V-\mathbb{1}]}{4}+\frac{\|\textbf{m}\|_2^2}{2}$, while $\alpha$ denotes a finite quantity defined in terms of the first moment $\textbf{m}$ and covariance matrix $V$ as
\bb
    \alpha&\coloneqq \frac{e^{\textbf{m}^\intercal((8N+4) \mathbb{1}-V)^{-1}\textbf{m} }}{\sqrt{\det\!\left[\frac{(8N+4) \mathbb{1}-V}{8N+3}\right]}}\leq 2^n e^{1/2}\,.
\ee 
A strictly tighter bound can be obtained in terms of the following one-parameter optimisation problem:
\bb\label{eq_87}
    P_{>M}\le \inf_{x>\|V\|_\infty}\frac{e^{\textbf{m}^\intercal(x \mathbb{1}-V)^{-1}\textbf{m} }}{\sqrt{\det\!\left[\frac{x \mathbb{1}-V}{x -1}\right]}}\, e^{-2\,\mathrm{arccoth}(x)\,M }\,,
\ee
where $\|V\|_\infty$ denotes the operator norm of $V$. 
\end{thm}
Before proving Theorem~\ref{the:expdecay00}, we need to establish some preliminary lemmas.
Let us start with the following.
\begin{lemma}\label{lemma_op_norm}
    The operator norm of the covariance matrix of an $n$-mode quantum state $\rho$ can be upper bounded in terms of its mean photon number $N\coloneqq \Tr[\hat{N}_n\rho]$ as
    \bb
        \|V(\rho)\|_\infty\le 1+2N+2\sqrt{N^2+N}\,,
    \ee
    where $\hat{N}_n$ denotes the total photon number operator. 
\end{lemma} 
\begin{proof}
    Note that the spectrum of any covariance matrix $V(\rho)$ is of the form $\{w_1,\frac{c_1}{w_1},\ldots, w_n,\frac{c_n}{w_n}\}$, where $w_1,w_2,\ldots, w_n\ge1$ and $c_1,\ldots, c_n\ge1$~\cite{Oh_2024}. To prove this simple claim, note that Williamson's decomposition~\cite{BUCCO} establishes that there exists a symplectic matrix $S$ and a diagonal matrix $D\ge\mathbb{1}$ such that $V(\rho)=SDS^\intercal$. In particular, it holds that $V(\rho)\ge SS^\intercal$. By Euler decomposition~\cite{BUCCO}, the spectrum of $SS^\intercal$ is of the form $\{z_1,\frac{1}{z_1},\ldots, z_n,\frac{1}{z_n}\}$, where $z_1,\ldots, z_n\ge1$. Since $V(\rho)\ge SS^\intercal$, Weyl's monotonicity theorem~\cite{BHATIA-MATRIX} implies the claim.

    With this claim at hand, note that
    \bb
        \Tr V(\rho)=\sum_{i=1}^n \left(w_i+\frac{c_i}{w_i}\right)\ge \sum_{i=1}^n \left(w_i+\frac{1}{w_i}\right)= \|V(\rho)\|_\infty+\frac{1}{\|V(\rho)\|_\infty}+\sum_{i=2}^n \left(w_i+\frac{1}{w_i}\right)\ge \frac{1}{\|V(\rho)\|_\infty}+\|V(\rho)\|_\infty+2n-2\,\,,
    \ee
    where the second equality follows from the fact that $w_1$ is the largest eigenvalue of the positive matrix $V(\rho)$, and thus it equals the operator norm of $V(\rho)$, while in the last inequality we exploited that $\min_{x\ge 1}( x+\frac{1}{x})=2$. 
    In particular, by using that
    \bb
    \Tr V= 2\sum_{i=1}^{2n}\Tr[\hat{R}_i^2\rho]-2\sum_{i=1}^{2n}[m_i(\rho)]^2\le2\sum_{i=1}^{2n}\Tr[\hat{R}_i^2\rho]= 4\Tr[\hat{N}_n\rho]+2n=4N+2n\,,
    \ee
    we obtain that 
    \bb
        \frac{1}{\|V(\rho)\|_\infty}+\|V(\rho)\|_\infty\le 4N+2\,.
    \ee
    Since $\|V(\rho)\|_\infty\ge 1$, by solving the above inequality with respect to $\|V(\rho)\|_\infty$, we conclude that
    \bb
        \|V(\rho)\|_\infty\le 1+2N+2\sqrt{N^2+N}\,.
    \ee
\end{proof}
For the rest of this subsection, we will use the following notation. First, let us recall that the Gaussian state $\rho_{\textbf{m},V}$ with first moment $\textbf{m}$ and covariance matrix $V$ can be expressed as~\cite{BUCCO}
    \bb\label{eq_gauss_state_h}
        \rho_{\textbf{m},V}=D_{\textbf{m}}\frac{e^{-\frac12 \hat{\textbf{R}}^\intercal H_V \textbf{R}  }}{\sqrt{\det\!\left[\frac{V+i\Omega}{2}\right]}}D_{\textbf{m}}^\dagger\,,
    \ee
    where $D_{\textbf{m}}$ denotes the displacement operator and
    \bb
        H_V\coloneqq 2i\Omega\mathrm{arccoth}(Vi\Omega)\,.
    \ee
    For any $\alpha>0$, let us denote as $\rho_{\textbf{m},V}^{(\alpha)}$ the Gaussian state proportional to $(\rho_{\textbf{m},V})^\alpha$, i.e.
    \bb\label{eq_def_rho_alpha}
        \rho_{\textbf{m},V}^{(\alpha)}\coloneqq \frac{(\rho_{\textbf{m},V})^\alpha}{\Tr [(\rho_{\textbf{m},V})^\alpha]}=D_{\textbf{m}}\frac{e^{-\frac12 \hat{\textbf{R}}^\intercal (\alpha H_V) \textbf{R}  }}{\Tr\left[e^{-\frac12 \hat{\textbf{R}}^\intercal (\alpha H_V) \textbf{R}  }\right]}D_{\textbf{m}}^\dagger\,.
    \ee 
    Moreover, given a covariance matrix $V$, let $V^{(\alpha)}$ be the covariance matrix of the Gaussian state $\rho_{\textbf{m},V}^{(\alpha)}$. In other words, we can write 
    \bb\label{power_cov}
        \rho_{\textbf{m},V}^{(\alpha)}=\rho_{\textbf{m},V^{(\alpha)}}\,.
    \ee
We are now ready to state the following lemma~\cite[Corollary~20]{seshadreesan2018}.
\begin{lemma}[(Overlap between the square of a Gaussian state and the inverse of a Gaussian state~\cite{seshadreesan2018})]\label{lemma_cor20}
    Let $\rho_{\textbf{m},V}$ and $\rho_{\textbf{t},W}$ be Gaussian states with first moments $\textbf{m},\textbf{t}$, respectively, and covariance matrices $V,W$, respectively. Let $V^{(2)}$ as defined in \eqref{power_cov}. If $W-V^{(2)}>0$, then it holds that
    \bb
        \Tr[\rho_{\textbf{m},V}^2\rho_{\textbf{t},W}^{-1}]=\frac{\det\!\left[\frac{W+i\Omega}{2}\right]}{\det\!\left[\frac{V+i\Omega}{2}\right]}\frac{  \sqrt{\det\!\left[\frac{V^{(2)}+i\Omega}{2}\right]} }{\sqrt{\det\!\left[\frac{W-V^{(2)}}{2}\right]}}e^{(\textbf{m}-\textbf{t})^\intercal(W-V^{(2)})^{-1}(\textbf{m}-\textbf{t})   }\,.
    \ee
\end{lemma}
With this lemma at hand, we can prove the following result.
\begin{lemma}[(Overlap between a Gaussian state and the inverse of a Gaussian state)]\label{lemma_overlap_ginv}
    Let $\rho_{\textbf{m},V}$ and $\rho_{\textbf{t},W}$ be Gaussian states with first moments $\textbf{m},\textbf{t}$, respectively, and covariance matrices $V,W$, respectively. If $W-V>0$, then it holds that
    \bb
        \Tr\!\left[\rho_{\textbf{m},V}(\rho_{\textbf{t},W})^{-1}\right]=\frac{\det\!\left[\frac{W+i\Omega}{2}\right]}{\sqrt{\det\!\left[\frac{W-V}{2}\right]}}e^{(\textbf{m}-\textbf{t})^\intercal(W-V)^{-1}(\textbf{m}-\textbf{t})   }\,.
    \ee
\end{lemma}
\begin{proof}
The crux of the proof is to exploit the known formula for the overlap between the square of a Gaussian state and the inverse of a Gaussian state reported in Lemma~\ref{lemma_cor20}. First, let us note that
\bb
        \rho_{\textbf{m},V}&\eqt{(i)}D_{\textbf{m}}\frac{e^{-\frac12 \hat{\textbf{R}}^\intercal H_V \textbf{R}  }}{\sqrt{\det\!\left[\frac{V+i\Omega}{2}\right]}}D_{\textbf{m}}^\dagger\\
        &=\frac{\det\!\left[\frac{V^{(1/2)}+i\Omega}{2}\right]}{\sqrt{\det\!\left[\frac{V+i\Omega}{2}\right]}}\left(D_{\textbf{m}}\frac{e^{-\frac12 \hat{\textbf{R}}^\intercal (H_V/2) \textbf{R}  }}{\sqrt{\det\!\left[\frac{V^{(1/2)}+i\Omega}{2}\right]}}D_{\textbf{m}}^\dagger\right)^2\\
        &\eqt{(ii)}\frac{\det\!\left[\frac{V^{(1/2)}+i\Omega}{2}\right]}{\sqrt{\det\!\left[\frac{V+i\Omega}{2}\right]}}\left(D_{\textbf{m}}\frac{e^{-\frac12 \hat{\textbf{R}}^\intercal (H_V/2) \textbf{R}  }}{  \Tr\left[e^{-\frac12 \hat{\textbf{R}}^\intercal (H_V/2) \textbf{R}  }\right]}D_{\textbf{m}}^\dagger\right)^2\\
        &\eqt{(iii)}\frac{\det\!\left[\frac{V^{(1/2)}+i\Omega}{2}\right]}{\sqrt{\det\!\left[\frac{V+i\Omega}{2}\right]}}(\rho_{\textbf{m},V}^{(1/2)})^2\,.
    \ee
Here, in (i), we exploited \eqref{eq_gauss_state_h}. In (ii), we exploited that $\Tr[e^{-\frac12 \hat{\textbf{R}}^\intercal (H_V/2) \textbf{R}  }]=\sqrt{\det\!\left[\frac{V^{(1/2)}+i\Omega}{2}\right]}$, which follows by using the definition of $\rho_{\textbf{m},V}^{(1/2)}$ in \eqref{eq_def_rho_alpha} and using that $\Tr \rho_{\textbf{m},V}^{(1/2)} =1$. Finally, in (iii) we used the definition of $\rho_{\textbf{m},V}^{(1/2)}$ in \eqref{eq_def_rho_alpha}.

Consequently, we have that
    \bb
        \Tr[\rho_{\textbf{m},V}\rho_{\textbf{t},W}^{-1}]&=\frac{\det\!\left[\frac{V^{(1/2)}+i\Omega}{2}\right]}{\sqrt{\det\!\left[\frac{V+i\Omega}{2}\right]}}\Tr[(\rho_{\textbf{m},V}^{(1/2)})^2\rho_{\textbf{t},W}^{-1}]\\
        &\eqt{(iv)}\frac{\det\!\left[\frac{V^{(1/2)}+i\Omega}{2}\right]}{\sqrt{\det\!\left[\frac{V+i\Omega}{2}\right]}} \frac{\det\!\left[\frac{W+i\Omega}{2}\right]}{\det\!\left[\frac{V^{(1/2)}+i\Omega}{2}\right]}\frac{  \sqrt{\det\!\left[\frac{V+i\Omega}{2}\right]} }{\sqrt{\det\!\left[\frac{W-V}{2}\right]}}e^{(\textbf{m}-\textbf{t})^\intercal(W-V)^{-1}(\textbf{m}-\textbf{t})   }\\
        &=\frac{\det\!\left[\frac{W+i\Omega}{2}\right]}{\sqrt{\det\!\left[\frac{W-V}{2}\right]}}e^{(\textbf{m}-\textbf{t})^\intercal(W-V)^{-1}(\textbf{m}-\textbf{t})   }\,,
    \ee
    where in (iv) we exploited Lemma~\ref{lemma_cor20}.
\end{proof}
We are now ready to prove Theorem~\ref{the:expdecay00}.
\begin{proof}[Proof of Theorem~\ref{the:expdecay00}]
Note that for all $t\ge0$ it holds that
\bb
    P_{>M}\coloneqq \Tr[(\mathbb{1}-\Pi_M)\rho]\le \Tr[e^{t\hat{N}_n}\rho] e^{-t M} \,,
\ee
where we exploited the operator inequality
\bb
    \mathbb{1}-\Pi_M\le e^{t\hat{N}_n}e^{-t M}\,,
\ee
which can be readily verified in Fock basis, where $\hat{N}_n$ denotes the photon number operator. We can minimise with respect to $t>0$ the right-hand side of \eqref{eq_1} in order to optimise the bound:
\bb\label{eq_0_tail}
    P_{>M}\le\inf_{t> 0}\Tr[e^{t\hat{N}_n}\rho] e^{-t M}\,.
\ee
However, note that the term $\Tr[e^{t\hat{N}_n}\rho]$ might diverge if $t$ is sufficiently large (e.g.~take $\rho$ to be a thermal state with strictly positive temperature). Hence, if there exists a sufficiently small $\bar{t}>0$ such that $\Tr[e^{\bar{t}\hat{N}_n}\rho]$ is finite, then $P_{>M}$ converges to zero \emph{exponentially} fast in $M$. We now show that, if $\rho$ is a Gaussian state, such a $\bar{t}$ actually exists. To this end, let us note that
\bb
    \Tr[e^{t\hat{N}_n}\rho]&=(1-e^{-t})^n\Tr\left[ \left(\tau^{\otimes n}_{\frac{1}{e^t-1}}\right)^{-1} \rho \right]\,,
\ee
where we introduced the thermal state $\tau_{\frac{1}{e^t-1}}$ with mean photon number equal to $\frac{1}{e^t-1}$, which can be expressed as
\bb
\left(\tau_{\frac{1}{e^t-1}}\right)^{\otimes n}=(1-e^{-t})^n \,e^{-t \hat{N}_n}\,.
\ee
Now, we want to exploit the formula for the overlap between a Gaussian state and the inverse of a Gaussian state reported in Lemma~\ref{lemma_overlap_ginv}. Since the covariance matrix of $\tau^{\otimes n}_{\frac{1}{e^t-1}}$ is given by $\coth (t/2) \mathbb{1}$, the hypothesis of Lemma~\ref{lemma_overlap_ginv} is satisfied if $t$ is such that $\coth(t/2)\mathbb{1}-V>0$, i.e.~$\coth (t/2)>\|V\|_\infty$. Hence, thanks to Lemma~\ref{lemma_overlap_ginv}, for any $t<2\mathrm{arccoth}(\|V\|_\infty)$ it holds that
\bb\label{chain_equal}
    \Tr[e^{t\hat{N}_n}\rho]    &=(1-e^{-t})^n \frac{\det\!\left[\frac{\coth( t/2) \mathbb{1}+i\Omega}{2}\right]}{\sqrt{\det\!\left[\frac{\coth( t/2) \mathbb{1}-V}{2}\right]}}e^{\textbf{m}^\intercal(\coth( t/2) \mathbb{1}-V)^{-1}\textbf{m}   }\\
    &=\left(\coth( t/2)-1\right)^{n} \frac{e^{\textbf{m}^\intercal(\coth( t/2) \mathbb{1}-V)^{-1}\textbf{m} }}{\sqrt{\det\!\left[\coth( t/2) \mathbb{1}-V\right]}} \\
    &=\frac{e^{\textbf{m}^\intercal(\coth( t/2) \mathbb{1}-V)^{-1}\textbf{m} }}{\sqrt{\det\!\left[\frac{\coth( t/2) \mathbb{1}-V}{\coth( t/2) -1}\right]}}
\ee
Hence, \eqref{eq_0_tail} implies that
\bb
    P_{>M}&\le  \inf_{0<t<2\mathrm{arccoth}(\|V\|_\infty)}\frac{e^{\textbf{m}^\intercal(\coth( t/2) \mathbb{1}-V)^{-1}\textbf{m} }}{\sqrt{\det\!\left[\frac{\coth( t/2) \mathbb{1}-V}{\coth( t/2) -1}\right]}}\, e^{-t M}\\
    &= \inf_{x>\|V\|_\infty}\frac{e^{\textbf{m}^\intercal(x \mathbb{1}-V)^{-1}\textbf{m} }}{\sqrt{\det\!\left[\frac{x \mathbb{1}-V}{x -1}\right]}}\, e^{-2\,\mathrm{arccoth}(x)\, M}\,,
\ee
which proves \eqref{eq_87}.

Moreover, note that Lemma~\ref{lemma_op_norm} guarantees that $1+2N+2\sqrt{N^2+N}\ge\|V\|_\infty$, where $N$ is the mean photon number of the state. Hence, since $8N+4>1+2N+2\sqrt{N^2+N}$ for all $N\ge0$, we can choose the value of $x$ in the above minimisation problem to be
\bb
    x&\coloneqq  8N+4 \,,
\ee
so that we can conclude that 
\bb
    P_{>M}\le\frac{e^{\textbf{m}^\intercal\left[(8N+4) \mathbb{1}-V\right]^{-1}\textbf{m} }}{\sqrt{\det\!\left[\frac{(8N+4) \mathbb{1}-V}{8N+3}\right]}}\, e^{-2\,\mathrm{arccoth}(8N+4)\,M}\,.
\ee
Moreover, by exploiting the elementary inequality $\mathrm{arccoth}(x)\ge\frac1x$, which is valid for any $x\ge 1$, we obtain that
\bb
    P_{>M}\le\frac{e^{\textbf{m}^\intercal\left[(8N+4) \mathbb{1}-V\right]^{-1}\textbf{m} }}{\sqrt{\det\!\left[\frac{(8N+4) \mathbb{1}-V}{8N+3}\right]}}\, e^{-M/(4N+2)}\,.
\ee
Moreover, note that $N$ can be expressed in terms of the first moment and covariance matrix as $N=\frac{\Tr[V-\mathbb{1}]}{4}+\frac{\|m\|_2^2}{2}$, as it can be proved by substituting the definitions of $V$, $m$, and $\hat{N}_n$. By using 
\bb (8N+4)\mathbb{1}\ge(4N+2 )\mathbb{1}\ge\left(1+2N+2\sqrt{N^2+N}\right)\mathbb{1}\geq V\,,
\ee
as guaranteed by Lemma~\ref{lemma_op_norm}, it holds that
\bb
\textbf{m}^\intercal\left[(8N+4) \mathbb{1}-V\right]^{-1}\textbf{m} &\leq \textbf{m}^\intercal\left[(8N+4) \mathbb{1}-V\right]^{-1}\textbf{m}\\
&\leq \textbf{m}^\intercal\textbf{m}(4N+2)^{-1}\\
&\leq \|m\|_2^2 (2\|m\|_2^2+2)^{-1}\\
&\leq 1/2\,,
\ee
and
\bb
\det\!\left[\frac{(8N+4) \mathbb{1}-V}{8N+3}\right]\geq \det\!\left[\frac{4N+2}{8N+3}\mathbb{1}\right]\geq 2^{-2n}\,.
\ee
Consequently, we conclude that
\bb
P_{>M}\le\frac{e^{\textbf{m}^\intercal\left[(8N+4) \mathbb{1}-V\right]^{-1}\textbf{m} }}{\sqrt{\det\!\left[\frac{(8N+4) \mathbb{1}-V}{8N+3}\right]}}\, e^{-M/(4N+2)}\le 2^n e^{-M/(4N+2)+1/2} \,.
\ee

\end{proof}
\subsection{Trace-distance tail bound on Gaussian states}\label{tail_2}
In this subsection, we solve the following problem. What is the trace-distance error in approximating a Gaussian state with its finite-dimensional approximation consisting of at most $M$ photons? More precisely, given a Gaussian state specified by its first moment $\textbf{m}$ and covariance matrix $V$, how to estimate the trace distance between itself and its projection onto the finite-dimensional subspace spanned by all the Fock states with at most $M$ photons, as a function of $\textbf{m}$, $V$, and $M$?

Approximating a Gaussian state with a finite-dimensional state can be useful not only from an analytical perspective --- as we will see in Section~\ref{section_AEP0} --- but also from a numerical one. Indeed, to numerically deal with an infinite-dimensional state, one usually truncates the Fock basis expansion of the state up to a certain photon cut-off. Due to the Holevo--Helstrom theorem~\cite{HELSTROM, Holevo1976}, the most meaningful way to measure the error incurred in such an approximation is given by the \emph{trace-distance error}, defined as the trace distance between the original state and its truncated approximation. Hence, it is important to estimate such a trace-distance error in terms of the photon cut-off. Notably, the forthcoming Theorem~\ref{the:expdecay} establishes that, for Gaussian states, such a trace-distance error decreases exponentially with the photon cut-off. This result has both analytical and numerical significance: analytically, it provides a rigorous and efficiently computable bound on the trace-distance error; numerically, it aids in estimating the trace-distance error when storing the truncated density matrix of a Gaussian state in a computer. Note that, in general, if one wants to numerically analyse properties of states prepared by applying general non-Gaussian operations to a Gaussian state, one necessarily need to make finite-dimensional approximations of the state.

Let us formalise the problem mathematically. Let us define the finite-dimensional Hilbert space $\HH_M$ spanned by all the $n$-mode Fock states with total number of photons less than $M$:
\bb
    \HH_M &\coloneqq  \mathrm{span}\!\left\{\ket{\textbf{k}}:\,\,\textbf{k}\in\N^n \,\,,\sum_{i=1}^n k_i\leq M\right\} \label{hmdef}
\ee
where $\ket{\textbf{k}}$ is the $n$-mode Fock state defined by $\ket{\textbf{k}}=\ket{k_1}\otimes\ket{k_2}\otimes\ldots\otimes\ket{k_n}$. We dub the integer $M$ as \emph{photon cut-off}. Moreover, we dub the projection of $\rho$ onto $\HH_M$, defined as
\bb
    \rho_M\coloneqq \frac{\Pi_M\rho\Pi_M}{\Tr[\Pi_M\rho\Pi_M]}\,,
\ee
where $\Pi_M$ is the projector onto $\HH_M$, as \emph{truncated approximation} of $\rho$. Given a Gaussian state $\rho$ with first moment $\textbf{m}$ and covariance matrix $V$, our goal is to find an upper bound on the trace distance $\frac{1}{2}\|\rho-\rho_M\|_1$ between $\rho$ and its truncated approximation $\rho_M$ in terms of the first moment $\textbf{m}$, $V$, and $M$. In the following theorem, we establish such an upper bound.
\begin{thm}[(Trace-distance tail bound on Gaussian states)]\label{the:expdecay}
Let $\rho$ be an $n$-mode Gaussian state with first moment $\textbf{m}$ and covariance matrix $V$. Let $\rho_M\coloneqq\frac{\Pi_M\rho\Pi_M}{\Tr[\Pi_M\rho\Pi_M]}$ be the projection of $\rho$ onto the Hilbert space spanned by all the $n$-mode Fock states with photon number less than $M$, where $\Pi_M$ is the projector on such a Hilbert space:
\bb
    \Pi_M\coloneqq\sum_{\textbf{k}\in\mathbb{N}^n:\,\sum_{i=1}^n k_i\le M}\ketbra{\textbf{k}}\,,
\ee
with $\ket{\textbf{k}}=\ket{k_1}\otimes\ket{k_2}\otimes\ldots\otimes\ket{k_n}$ being the $n$-mode Fock state. Then, the trace distance between $\rho$ and its truncated approximation $\rho_M$ converges exponentially fast to zero in the photon cut-off $M$, with exponential lifetime given by the mean photon number $N$. Specifically, such a trace distance can be upper bounded as
\bb
    \frac{1}{2}\left\|\rho-\rho_M\right\|_1\le c\,e^{-M /(8N+4) }\,,
\ee
where $c$ denotes a finite quantity defined in terms of the first moment $\textbf{m}$ and covariance matrix $V$ as
\bb
    c&\coloneqq \frac{e^{\frac{1}{2}\textbf{m}^\intercal((8N+4) \mathbb{1}-V)^{-1}\textbf{m} }}{\left(\det\!\left[\frac{(8N+4) \mathbb{1}-V}{8N+3}\right]\right)^{1/4}}\,\leq 2^{n/2}e^{1/4},
\ee 
and $N$ denotes the mean photon number given by $N=\frac{\Tr[V-\mathbb{1}]}{4}+\frac{\|\textbf{m}\|_2^2}{2}$.

A strictly tighter bound can be obtained in terms of the following one-parameter optimisation problem:
\bb
    \frac{1}{2}\left\|\rho-\rho_M\right\|_1\le \inf_{x>\|V\|_\infty} \frac{e^{\frac{1}{2}\textbf{m}^\intercal(x \mathbb{1}-V)^{-1}\textbf{m} }}{\left(\det\!\left[\frac{x \mathbb{1}-V}{x-1}\right]\right)^{1/4}} e^{-M \mathrm{arccoth}(x)}\,,
\ee
where $\|V\|_\infty$ denotes the operator norm of $V$. 
\end{thm}
\begin{proof}
    Gentle Measurement Lemma \cite[Lemma 6.15]{Sumeet_book}  establishes that
    \bb\label{eq_1}
    \frac{1}{2}\left\|\rho-\frac{\Pi_M\rho\Pi_M}{\Tr[\Pi_M\rho\Pi_M]}\right\|_1\le \sqrt{\Tr[(\mathbb{1}-\Pi_M)\rho]}\,,
\ee
i.e.~it holds that $\frac{1}{2}\|\rho-\rho_M\|_1\le \sqrt{P_{>M}}$, where we introduced the probability $P_{>M}$ that $\rho$ has more than $M$ photons. Hence, the statement directly follows by applying Theorem~\ref{the:expdecay00}.
\end{proof}
As an example of application of the trace-distance tail bound on Gaussian states presented in Theorem~\ref{the:expdecay}, in the next subsection we design an efficient algorithm to estimate the trace distance between two Gaussian states.

\subsection{Algorithm to estimate the trace distance between two Gaussian states}
Calculating the trace distance between two Gaussian states is important, as the trace distance is the most meaningful notion of distance between quantum states~\cite{HELSTROM, Holevo1976} and Gaussian states are omnipresent in nature. However, calculating exactly such a trace distance is impossible in general, as the difference of two Gaussian states is a non-Gaussian operator, and there is no way to compute in general the trace norm of a non-Gaussian (infinite-dimensional) operator. Therefore, it is necessary to come up with methods to \emph{approximate}, up to a fixed precision, the trace distance between two Gaussian states. While recent works have established both lower~\cite{mele2024learningquantumstatescontinuous} and upper~\cite{bittel2024optimalestimatestracedistance,fanizza2024efficienthamiltonianstructuretrace,mele2024learningquantumstatescontinuous,holevo2024estimatestracenormdistancequantum,Banchi_2015} bounds on the trace distance between two Gaussian states, there remains no known algorithm capable of computing such a trace distance up to a fixed precision. Here, we fill this gap, by providing the first algorithm to compute the trace distance between Gaussian states up to a fixed precision.

Our algorithm leverages the trace-distance tail bound on Gaussian states introduced in Theorem~\ref{the:expdecay}. Specifically, by denoting as $\rho_{\textbf{m},V}$ the Gaussian state with first moment $\textbf{m}$ and covariance matrix $\textbf{t}$, the algorithm exploits the fact that 
\bb\label{eq_algo_0}
    \frac{1}{2}\left\|\rho_{\textbf{m},V}-\frac{\Pi_M\rho_{\textbf{m},V}\Pi_M}{\Tr[\Pi_M\rho_{\textbf{m},V}\Pi_M]}\right\|_1\le 2^{n/2}e^{1/4}e^{-M /(8N+4) }
\ee
where $N$ denotes the mean photon number given by 
\bb\label{mean_ph_f}
    N=\frac{\Tr[V-\mathbb{1}]}{4}+\frac{\|\textbf{m}\|_2^2}{2}\,.
\ee
The algorithm is given below and its correctness is proved in Theorem~\ref{thm_algo}.
\begin{algorithm}[H]\label{algo1}
\caption{(Computation of the trace distance $\frac12\|\rho_{\textbf{m},V}-\rho_{\textbf{t},W}\|_1$ between two Gaussian states $\rho_{\textbf{m},V}$ and $\rho_{\textbf{t},W}$, given their first moments $\textbf{m},\textbf{t}$ and covariance matrices $V,W$.}
\begin{algorithmic}[1]
\Require First moments $\textbf{m},\textbf{t}$; covariance matrices $V,W$; error tolerance $\epsilon > 0$.
 \Ensure Approximation $\tilde{d}$ such that $ \left|\frac12\|\rho_{\textbf{m},V}-\rho_{\textbf{t},W}\|_1- \tilde{d}\right| \leq \epsilon$
  \State Use Eq.~\eqref{mean_ph_f} to find the maximum $N$ between the mean photon number of $\rho_{\textbf{m},V}$ and the one of $\rho_{\textbf{t},W}$.
 \State Find $M$ large enough so that the right-hand side of Eq.~\eqref{eq_algo_0} is smaller than $\epsilon/3$. A value of $M = \pazocal{O}\!\left((n+1)(N+1)\log \!\left(\frac{1}{\epsilon}\right)\right)$ is sufficient. 
 \State Calculate the matrix elements of $\rho_{\textbf{m},V}$ with respect to those Fock states which have number of photons not larger than $M$ (e.g.~via the methods introduced in \cite{Quesada_2019}). With such matrix elements, construct a finite-dimensional matrix $A_1$ and set $\rho_1\coloneqq A_1/\Tr A_1$.
\State Analogously, calculate the matrix elements of $\rho_{\textbf{t},W}$ with respect to the same basis as above, and construct the corresponding finite-dimensional matrix $A_2$. Finally, set $\rho_2\coloneqq A_2/\Tr A_2$.
 \State Compute the trace distance between the finite dimensional matrices $\rho_1$ and $\rho_2$ up to an error $\epsilon/3$.
 \State \Return the estimate of the above trace distance.
 \end{algorithmic}
 \end{algorithm}
\begin{thm}[(Algorithm to approximate the trace distance between Gaussian states)]\label{thm_algo}
    Let $\rho_{\textbf{m},V}$ and $\rho_{\textbf{t},W}$ be $n$-mode Gaussian states with first moments $\textbf{m},\textbf{t}$ and covariance matrices $V,W$. Given in input the first moments $\textbf{m},\textbf{t}$, covariance matrices $V,W$, and error tolerance $\epsilon > 0$, the above algorithm outputs $\tilde{d}$ such that it approximates the trace distance between the two states up to an error $\epsilon$:
    \bb\label{eq_to_prove_algo1}
        \left|\frac12\|\rho_{\textbf{m},V}-\rho_{\textbf{t},W}\|_1- \tilde{d}\right| \leq \epsilon\,.
    \ee
    Moreover, the runtime of the algorithm scales exponentially in the number of modes $n$; specifically as
    \bb
        \pazocal{O}\!\left(2^{6}(n+1)E\log (2/\epsilon)\right)^{3n}+t_{\mathrm{Fock}}\,,
    \ee
    where $E$ is the maximum mean energy per mode of the two states, i.e.
    \bb
        E\coloneqq \frac1n\max\left( \frac{\Tr V}{4}+\frac{\|\textbf{m}\|_2^2}{2},\frac{\Tr W}{4}+\frac{\|\textbf{t}\|_2^2}{2}\right)\,,
    \ee
    while $t_{\mathrm{Fock}}$ is the time required to calculate the Fock basis matrix elements in Step~3 and Step~4. Additionally, the notation $\pazocal{O}$ hides logarithmic factors in $E$.
\end{thm}
\begin{proof}
    In this proof we will use the notation introduced in the algorithm in Table~\ref{algo1}, Eq.~\eqref{eq_to_prove_algo1} trivially follows by multiple applications of triangle inequality and by observing that
    \bb
        \|\rho_2-\rho_1\|_1=\left\|\frac{\Pi_M\rho_{\textbf{t},W}\Pi_M}{\Tr[\Pi_M\rho_{\textbf{t},W}\Pi_M]}-\frac{\Pi_M\rho_{\textbf{m},V}\Pi_M}{\Tr[\Pi_M\rho_{\textbf{m},V}\Pi_M]}\right\|_1\,.
    \ee
    Let us estimate the runtime of the algorithm. Let $N_1=\frac{\Tr[V-\mathbb{1}]}{4}+\frac{\|m\|_2^2}{2}$, $N-2=\frac{\Tr[W-\mathbb{1}]}{4}+\frac{\|t\|_2^2}{2}$. At Step~2, $M$ is chosen such that
 
    \bb
        2^ne^{-M /(8N_1+4)+1/4 }&\le \epsilon/3\,,\\
         2^ne^{-M /(8N_2+4)+1/4 }&\le \epsilon/3\,,
    \ee
    i.e.~such that
    \bb
        M\ge \max(8N_1+4,8N_2+4)\log\frac{3e^{1/4}\cdot 2^n}{\epsilon}\,.
    \ee
    Hence, since $nE=N+n/2$, it suffices to take 
    \bb
        M\coloneqq \lceil 8nE\log \frac{3e^{1/4}\cdot 2^n}{\epsilon}\rceil\,.
    \ee
    Moreover, Step~5 requires the diagonalisation of the finite dimensional matrix $\rho_1-\rho_2$, a task well-known to be achieved in a time scaling as $\pazocal{O}(d^3)$, where $d$ denotes the dimension of the matrix $\rho_1-\rho_2$. By definition, such a dimension $d$ is given by the total number of $n$-mode Fock states with number of photons not larger than $M$. An elementary counting argument shows that
    \bb
        d=\binom{n+M}{n}\,.
    \ee
    In particular, by exploiting the elementary inequality $\binom{a}{b}\le (e\frac{a}{b})^b$ valid for any $a,b\in\mathbb{N}$, it holds that
    \bb
        d\le \left( e+ e\frac{\lceil 4nE\log \frac{3\cdot 2^n}{e\epsilon}\rceil}{n}\right)^n\le \left( e+\frac{e}{n}+4eE\log \frac{3e^{1/4}\cdot 2^n}{\epsilon}\right)^n = \pazocal{O}\left(\!\left(2^{6}(n+1)E\log (2/\epsilon)\right)^n\right)\,.
    \ee
    Consequently, the runtime of Step~5 is $\pazocal{O}\left(\!\left(2^{6}(n+1)E\log (2/\epsilon)\right)^{3n}\right)$. This concludes the proof.
\end{proof}
Notably, the runtime $\pazocal{O}\!\left(\!\left(2^{6}(n+1)E\log (2/\epsilon)\right)^{3n}\right)$ of the above algorithm, with its logarithmic dependence on the trace-distance error $\epsilon$, is a distinctive feature of Gaussian states. In contrast, as proved in Lemma~\ref{lemma_alt_alg}, a similar algorithm for computing the trace distance between two $n$-mode non-Gaussian states with precision error $\epsilon$ has a significantly worse upper bound on the runtime of $\pazocal{O}\!\left(\!\left(\frac{16E}{\epsilon^2}\right)^{3n}\right)$.
\begin{lemma}\label{lemma_alt_alg}
    The trace distance between two $n$-mode (non-Gaussian) states can be computed up to a trace-distance error $\epsilon$ in time $\pazocal{O}\!\left(\!\left(\frac{16E}{\epsilon^2}\right)^{3n}\right)$, where $E$ denotes the maximum between the mean energy per mode of the two states.
\end{lemma}
\begin{proof}
    By denoting as $\Pi_M$ the projector onto the Hilbert space spanned by all the $n$-mode Fock states with at most $M$ photons, it holds that
    \bb
        \frac12\left\|\rho-\frac{\Pi_M\rho\Pi_M}{\Tr[\Pi_M\rho\Pi_M]}\right\|_1\le \sqrt{\Tr[(\mathbb{1}-\Pi_M)\rho]}\le \sqrt{\frac{\Tr[\hat{N}_n\rho]}{M}}\,,
    \ee
    where we used Gentle Measurement Lemma~\cite[Lemma 6.15]{Sumeet_book} and the operator inequality $\mathbb{1}-\Pi_M\le \frac{\hat{N}_n}{M}$.
    Consequently, we have that
    \bb
        \frac12\left\|\rho-\frac{\Pi_M\rho\Pi_M}{\Tr[\Pi_M\rho\Pi_M]}\right\|_1\le \epsilon/3\,,
    \ee
    if the photon cut-off $M$ satisfies
    \bb\label{eq_m_en}
        M\ge \frac{9\Tr[\hat{N}_n\rho]}{\epsilon^2}\,.
    \ee
    Let us now follow the same reasoning as the proof of Theorem~\ref{thm_algo}. Thanks to Eq.~\eqref{eq_m_en}, in order to have both the two states approximated with trace-distance error $\varepsilon/3$ via their projection onto the Hilbert space spanned by all the Fock states with at most $M$ photons, it suffices to choose
    \bb
        M\coloneqq \left\lceil\frac{9nE}{\epsilon^2}\right\rceil\,,
    \ee
    where $E$ is the maximum mean energy per mode of the two states. The dimension $d$ of such Hilbert space is
    \bb
        d=\binom{M+n}{n}\le \left(e\frac{M+n}{n}\right)^n\le  \left(e+\frac{e}{n}+\frac{9eE}{\epsilon^2}\right)^n=\pazocal{O}\!\left(\!\left(\frac{16 E}{\epsilon^2}\right)^n\right)\,.
    \ee
    Hence, it suffices to diagonalise a matrix of dimension $d$, which has a runtime of $O(d^3)$.
\end{proof}

\newpage

\section{Non-asymptotic analysis of quantum communication across Gaussian channels}\label{section_AEP0}
This section aims to establish an easily computable lower bound on the $n$-shot quantum capacity, $n$-shot two-way quantum capacity, and $n$-shot secret-key capacity of an arbitrary Gaussian channel.

As we will see below, to obtain an easily computable lower bound on the $n$-shot capacities, it suffices to lower bound the smooth max relative entropy $D_{\max}^\varepsilon$ defined in Definition~\ref{def_max_smooth} in terms of quantities that do not require optimisations. This can be achieved by applying the \emph{asymptotic equipartition property} proved in \cite[Theorem 9]{Tomamichel_2009}, which holds in finite dimensions. This property was recently generalised to infinite dimensions in~\cite[Theorem 6.2]{fawzi2023asymptotic}, while an earlier, non-tight infinite-dimensional generalisation was presented in~\cite[Proposition 8]{Furrer_2011}.

We will now state the infinite-dimensional generalisation of the asymptotic equipartition property from~\cite[Theorem 6.2]{fawzi2023asymptotic}, which is formulated within the framework of \emph{von Neumann algebras}. Although we do not need to introduce this framework in our paper, it is relevant to mention that in the context of von Neumann algebras, any operator $\sigma$ with a bounded trace is also \emph{normal}~\cite{fawzi2023asymptotic}. Therefore, we can apply~\cite[Theorem 6.2]{fawzi2023asymptotic} under the assumption that the operator $\sigma$ has a bounded trace.

\begin{lemma}[(Infinite-dimensional asymptotic equipartition property~\cite{fawzi2023asymptotic} )]\label{lemma_fawzi}
    Let $\rho$ be a state and let $\sigma$ be an arbitrary positive semi-definite operator with bounded trace. For any $\varepsilon\in(0,1)$ with $n\ge2\log_2\!\left(\frac{2}{\varepsilon^2}\right)$, we have
    \bb\label{generalised_AEP}
    D_{\max}^\varepsilon(\rho^{\otimes n}\|\sigma^{\otimes n})\le n D(\rho\|\sigma)+4\sqrt{n}\log_2({\bar{\eta}})\sqrt{\log_2\!\left(\frac{2}{\varepsilon^2}\right)}\,,
    \ee
    where
    \bb
        \bar{\eta}\coloneqq \sqrt{2^{D_{3/2}(\rho\|\sigma)}}+\sqrt{2^{-D_{1/2}(\rho\|\sigma)}}+1\,,
    \ee
    where $D_{1/2}$ and $D_{3/2}$ denote the Petz Rényi relative entropy defined in Definition~\ref{petz_renyi}.
\end{lemma}

\subsection{Quantum capacity}
In this paragraph, we prove a lower bound on the $n$-shot quantum capacity of Gaussian channels, presented in the forthcoming Theorem~\ref{thm_lower_bound}.

We start by stating the following result regarding finite-dimensional channels, which establishes a lower bound on the one-shot quantum capacity in terms of the smooth conditional max-entropy.
\begin{lemma}[(Lower bound on the one-shot quantum capacity of finite-dimensional channels~\cite{Sumeet_book})]\label{Lemma_Sumeet_book}
    Let $\NN $ be a finite-dimensional quantum channel. For all $\varepsilon\in(0,1)$, $\delta\in(0,1)$, and $\eta\in[0,\frac{\sqrt{\delta}}{4})$, the one-shot quantum capacity of $\NN$ can be lower bounded as:
    \bb
        Q^{(\varepsilon)}(\NN)\ge -H_{\max}^{\frac{\varepsilon\sqrt{\delta}}{4}-\eta} (A|B)_{\NN_{A'\to B}(\psi_{AA'})}+\log_2\!\left(\eta^4(1-\delta)\right)
    \ee 
    for any pure state $\psi_{AA'}$, where $H_{\max}^{\varepsilon}$ denotes the smooth conditional max-entropy defined in Definition~\ref{smooth_max_conditional}. 
\end{lemma}
We are now ready to present our lower bound.
\begin{thm}[(Lower bound on the $n$-shot quantum capacity of Gaussian channels)]\label{thm_lower_bound}
Let $\NN $ be a Gaussian channel. Let $\varepsilon\in(0,1)$, $\delta\in(0,1)$, $\eta\in[0,\frac{\varepsilon\sqrt{\delta}}{4})$, and $n\in\N$ such that $n\ge2\log_2\!\left(\frac{2}{\varepsilon^2}\right)$. For any pure Gaussian state $\Phi_{AA'}$, the $n$-shot quantum capacity of $\NN$ satisfies the following lower bound:
    \bb\label{lower_bound_CV}
        Q^{(\varepsilon,n)}(\NN)\ge n I_c(A\,\rangle\, B)_{\Psi}-\sqrt{n}4\log_2({\bar{\eta}})\sqrt{\log_2\!\left(\frac{2}{\left(\frac{\varepsilon\sqrt{\delta}}{4}-\eta \right)^2}\right)}+\log_2\!\left(\eta^4(1-\delta)\right)\,,
    \ee
    where
    \bb\label{eq_bar_eta}
        \bar{\eta}&\coloneqq \sqrt{2^{H_{1/2}(A|B)_{\Psi}}}+\sqrt{2^{H_{1/2}(A|E)_{\Psi}}}+1\,,\\
        \Psi_{ABE}&\coloneqq V_{A'\to BE}\Phi_{AA'}(V_{A'\to BE})^{\dagger}\,,
    \ee
    and $V_{A'\to BE}$ is a Stinespring dilation isometry associated with $\NN_{A'\to B}$. Moreover, 
    \bb
        I_c(A\,\rangle\, B)_{\Psi}&\coloneqq S(B)_\Psi-S(AB)_{\Psi}\,,\\
        H_{1/2}(A|B)_{\Psi}&\coloneqq 2\log_2\!\left(\Tr\sqrt{\Psi_{AB}}\,\mathbb{1}_A\otimes\sqrt{\Psi_B}\right)\,,\\
        H_{1/2}(A|E)_{\Psi}&\coloneqq 2\log_2\!\left(\Tr\sqrt{\Psi_{AE}}\,\mathbb{1}_A\otimes\sqrt{\Psi_E}\right)\,
    \ee
    are finite quantities denoting coherent information and conditional Petz Rényi entropies. In particular, by choosing $\delta=\frac{1}{4}$ and $\eta=\frac{\varepsilon}{16}$, we get
    \bb\label{lower_bound_CV2}
        Q^{(\varepsilon,n)}(\NN)\ge n I_c(A\,\rangle\, B)_{\Psi}-\sqrt{n}4\log_2({\bar{\eta}})\sqrt{\log_2\!\left(\frac{2^9}{\varepsilon^2}\right)}-\log_2\!\left(\frac{2^{18}}{3\varepsilon^4}\right)\,.
    \ee
In particular, if the input state $\Phi_{AA'}$ satisfies the energy constraint $\Tr[\hat{N}_{A'}\Phi_{AA'}]\le N_s$, where $\hat{N}_{A'}$ is the photon number operator on $A'$, then the right-hand side of \eqref{lower_bound_CV2} constitutes also a lower bound on the energy-constrained $n$-shot quantum capacity $Q^{(\varepsilon,n)}(\NN,N_s)$.
\end{thm}
\begin{proof}
    Currently available lower bounds on the one-shot quantum capacity are formulated in the finite-dimensional setting, e.g.~the one provided above in Lemma~\ref{Lemma_Sumeet_book}. Therefore, we will begin by constructing a sequence of finite-dimensional approximations of the infinite-dimensional channel $\NN$ and then apply the finite-dimensional lower bound reported in Lemma~\ref{Lemma_Sumeet_book}.

Let us introduce some relevant notations. Given a continuous-variable system $S$, let us denote as $m_S$ the number of modes of $S$. Moreover, for any $k\in\mathbb{N}$, let us define the finite-dimensional Hilbert space $\HH^{S}_k$ spanned by all the $m_S$-mode Fock states on $S$ with total number of photons less than $k$:
\bb
    \HH^S_k &\coloneqq  \mathrm{span}\!\left\{\ket{\textbf{i}}_S:\,\,\textbf{i}\in\N^{m_S} \,\,,\sum_{j=1}^{m_S} i_j\leq k\right\}\,, 
\ee
where $\ket{\textbf{i}}$ is the ${m_S}$-mode Fock state defined by
\bb
    \ket{\textbf{i}}_S=\ket{i_1}\otimes\ket{i_2}\otimes\ldots\otimes\ket{i_{m_S}}\,.
\ee
Let us define as $P^S_k$ the projector on $\HH^S_k$, and let us note that
\bb\label{dim_hms}
    \Tr P^S_k=\dim\mathcal{H}_k^S= \left|\left\{i_1,i_2,\ldots,i_{m_S}\in\N\,,\quad\sum_{j=1}^{m_S} i_j\leq k\right\}\right| =\binom{k+m_S}{m_S} \,,
\ee
where in the last equality we used a standard combinatorial argument. In addition, let us consider a qudit system $S_k \cong\mathbb{C}^{\binom{k+m_S}{m_S}}$  of dimension $\binom{k+m_S}{m_S}$ and a co-isometry $V_{S\to S_k }$ that maps the subspace $\mathcal{H}_k^S$ into $S_k$. Moreover, let $\mathcal{V}_{S\to S_k }(\cdot)\coloneqq V_{S\to S_k }(\cdot)V_{S\to S_k }^\dagger $ be the corresponding superoperator. Moreover, let us consider the projection of the state $\Phi_{AA'}$ induced by the projector $P_k^{A}$:
    \bb
        \Phi^{(k)}_{A_k A'}\coloneqq  \mathcal{V}_{A\to A_k}\left(\frac{P_k^{A}\otimes\mathbb{1}_{A'}\Phi_{AA'}P_k^{A}\otimes\mathbb{1}_{A'}}{\Tr[P_k^{A}\Phi_{A}]}\right)
    \ee
    Since $\Phi^{(k)}_{A_k A'}$ is a pure state and the system $A_k$ has dimension $\binom{m_A+k}{m_A}$, the support of $\Tr_{A}\Phi^{(k)}_{A_k A'}$ is contained in a $\binom{m_A+k}{m_A}$-dimensional subspace of $A'$. Let $V'_{A'_k\to A'}$ be the isometric embedding that maps $A'_k\cong \mathbb{C}^{\binom{m_A+k}{m_A}}$ into such a subspace, let $\mathcal{V}'_{A'_k\to A'}(\cdot)\coloneqq V'_{A'_k\to A'}(\cdot)(V'_{A'_k\to A'})^\dagger$ be its associated channel, and let us introduce the following pure state on $A_k A'_k$:
    \bb
        \Phi^{(k)}_{A_k A'_k}\coloneqq \mathcal{V}'^\dagger_{A'_k\to A'} (\Phi^{(k)}_{A_k A'})\,,
    \ee
    where we denoted $\mathcal{V}'^\dagger_{A'_k\to A'}(\cdot)\coloneqq V'^\dagger_{A'_k\to A'}(\cdot)(V'_{A'_k\to A'})$.     Moreover, let us introduce the following \emph{finite}-dimensional channel $\NN^{(k)}:A'_k\to B_k$, which constitutes a finite-dimensional approximation of the Gaussian channel $\NN$:
    \bb
        \NN^{(k)}_{A'_k\to B_k}\coloneqq \mathcal{V}_{B\to B_k } \circ \mathcal{P}^{(k)}_{B\to B}\circ \NN_{A'\to B}\circ \mathcal{V}'_{A'_k\to A'}\,,
    \ee
    where
    \bb
        \mathcal{P}^{(k)}_{B\to B}(\cdot) \coloneqq P_{k}^B(\cdot)P^{B}_k +\Tr\!\left[(\mathbb{1}_{B}-P_{k}^B )(\cdot)\right]\ketbra{0}_{B}\,.
    \ee
    We can thus apply Lemma~\ref{Lemma_Sumeet_book} to such a finite-dimensional channel to obtain
    \bb
        Q^{(\varepsilon,n)}(\NN)&\geqt{(i)} -H_{\max}^{\frac{\varepsilon\sqrt{\delta}}{4}-\eta} (A'^{n}_k|B_k^{n})_{\omega_{A_kB_k}^{\otimes n}}+\log_2\!\left(\eta^4(1-\delta)\right)\,,
    \ee
    where $\omega_{A_kB_k}\coloneqq \mathbb{1}_{A_k}\otimes \NN^{(k)}_{A'_k\to B_k}(\Phi^{(k)}_{A_kA_k'})$. Now, we want to lower bound $Q^{(\varepsilon)}$ using the asymptotic equipartition property reported in Lemma~\ref{lemma_fawzi}. To do so, we need to write the above lower bound on $Q^{(\varepsilon)}$ in terms of the smooth conditional max-entropy $D_{\max}^\varepsilon$. To this end, let us write the Stinespring dilation of the above channels. Note that the isometry
    \bb\label{def_w_bbe}
        W^{(k)}_{B\to BE'}\coloneqq P_k^B\otimes\ket{0}_{E'}+\sum_{\substack{\textbf{i}\in\mathbb{N}^{m_B}\\\sum_{j=1}^{m_B}i_j \not\le k}} \ketbraa{0}{\textbf{i}}_B\otimes\ket{\textbf{i}}_{E'}\,
    \ee 
    constitutes a Stinespring dilation isometry for $\mathcal{P}^{(k)}_{B\to B}$. Consequently, we can construct a Stinespring dilation isometry for $\NN^{(k)}_{A'_k\to B_k}$ as follows:
    \bb
        V_{A'_k\to B'_k E E' }\coloneqq V_{B\to B_k } W^{(k)}_{B\to BE'} V_{A'\to BE} V'_{A'_k\to A'}\,.
    \ee 
    By setting
    \bb
        \Psi^{(k)}_{A_kB_kEE'}\coloneqq V_{A'_k\to B'_k E E' }\Phi^{(k)}_{A_kA'_k}V_{A'_k\to B'_k E E' }^\dagger\,,
    \ee
    we have that
    \bb\label{ineq_proof_thm_lower_bound}
        Q^{(\varepsilon,n)}(\NN)&\ge -H_{\max}^{\frac{\varepsilon\sqrt{\delta}}{4}-\eta} (A'^{n}_k|B_k^{n})_{\omega_{A_kB_k}^{\otimes n}}+\log_2\!\left(\eta^4(1-\delta)\right)\\
        &\eqt{(i)} -H_{\max}^{\frac{\varepsilon\sqrt{\delta}}{4}-\eta} (A'^{n}_k|B_k^{n})_{(\Psi^{(k)}_{A_kB_kEE'})^{\otimes n}}+\log_2\!\left(\eta^4(1-\delta)\right)\\
        &\eqt{(ii)} H_{\min}^{\frac{\varepsilon\sqrt{\delta}}{4}-\eta} (A'^{n}_k|E^{n}E'^{n})_{(\Psi^{(k)}_{A_kB_kEE'})^{\otimes n}}+\log_2\!\left(\eta^4(1-\delta)\right)\\
        &\geqt{(iii)} -D_{\max}^{\frac{\varepsilon\sqrt{\delta}}{4}-\eta}\left((\Psi^{(k)}_{A_kEE'})^{\otimes n}\|\mathbb{1}_{A_k^n}\otimes (\Psi^{(k)}_{EE'})^{\otimes n}\right)+\log_2\!\left(\eta^4(1-\delta)\right)\,.
    \ee
    Here, in (i) we just used the fact that
    $\omega_{A_kB_k}=\Tr_{EE'}\Psi^{(k)}_{A_kB_kEE'}$ together with our notation of omitting partial traces in the subscript of entropies; in (ii) we used the duality relation between the smooth conditional max-entropy and the smooth conditional min-entropy reported in Lemma~\ref{lemma_duality}; and in (iii) we leveraged the definition of smooth conditional min-entropy in Definition~\ref{def_smooth_min_ent}.

    We are now ready to apply the asymptotic equipartition property reported in Lemma~\ref{lemma_fawzi}. Note that the hypothesis of Lemma~\ref{lemma_fawzi} is satisfied because the operator $\sigma\coloneqq\mathbb{1}_{A_k}\otimes \Psi^{(k)}_{EE'}$ has bounded trace (specifically, its trace equals $\binom{m_A+k}{m_A}$). Consequently, we obtain that
    \bb\label{ineq_proof_lower}
        Q^{(\varepsilon,n)}(\NN)&\ge -n D\left(\Psi^{(k)}_{A_kEE'}\|\mathbb{1}_{A_k}\otimes \Psi^{(k)}_{EE'}\right)\\
        &\quad-{\sqrt{n}}4\log_2\!\left( \sqrt{2^{D_{3/2}(\Psi_{A_k EE'}^{(k)}\|\mathbb{1}_{A_k}\otimes\Psi_{EE'}^{(k)})}}+\sqrt{2^{-D_{1/2}(\Psi_{A_kEE'}^{(k)}\|\mathbb{1}_{A_k}\otimes\Psi_{EE'}^{(k)})}}+1 \right)\sqrt{\log_2\!\left(\frac{2}{\left(  \frac{\varepsilon\sqrt{\delta}}{4}-\eta  \right)^2}\right)}+\log_2\!\left(\eta^4(1-\delta)\right)\,.
    \ee
    The above inequality holds for any $k$. Now, we would like to take the limit $k\rightarrow\infty$ in the above inequality in order to make the following substitution 
    \bb
        \Psi_{A_kB_kEE'}^{(k)}\leftrightarrow  \Psi_{ABEE'} \coloneqq   \Psi_{ABE}^\dagger\otimes\ketbra{0}_{E'}\,,
    \ee
    where $\Psi_{ABE}\coloneqq V_{A'\to BE}\Phi_{AA'}V_{A'\to BE}$. This limit procedure is far from trivial and requires careful analysis of the entropic quantities involved in \eqref{ineq_proof_lower}. Let us start by showing that the limit $\Psi_{ABEE'}^{(k)}\xrightarrow{k\rightarrow\infty}  \Psi_{ABEE'}$ holds in trace norm. To this end, let us observe that the state $\Psi_{A_kB_kEE'}^{(k)}$ can be written as
    \bb\label{eq_def_psik_proof}
        \ket{\Psi_{A_kB_kEE'}^{(k)}}= V_{A\to A_k}\otimes V_{B\to B_k}\otimes \mathbb{1}_{EE'}\ket{\Psi_{ABEE'}^{(k)}}\,,
    \ee
    where we introduced 
    \bb
        \ket{\Psi_{ABEE'}^{(k)}}\coloneqq \frac{P_k^A\otimes W_{B\to BE'}^{(k)}\ket{\Psi_{ABE}}}{\sqrt{\Tr[P_k^A \Psi_{A}]}}\,,
    \ee
    and we used the vectorial notation $\psi=\ketbra{\psi}$ for any pure state $\psi$. Since it will be useful for the following, we are now going to prove not only that the convergence $\Psi_{ABEE'}^{(k)}\xrightarrow{k\rightarrow\infty}  \Psi_{ABEE'}$ holds in trace norm but also that such a convergence is exponentially fast in $k$:
    \bb\label{proof_exp_convergence}
        \|\Psi_{ABEE'}^{(k)}-\Psi_{ABEE'}\|_1&=2\sqrt{1-\Tr[\Psi_{ABEE'}^{(k)}\Psi_{ABEE'}]}\\
        &=2\sqrt{1-\frac{\Tr[\Psi_{AB}P_k^A\otimes P_k^B]}{\Tr[\Psi_{A}P_k^A]}}\\
        &\le 2\sqrt{1-\Tr[\Psi_{AB}P_k^A\otimes P_k^B]}\\
        &= 2\sqrt{\Tr[\Psi_{AB}(\mathbb{1}_A-P_k^A)\otimes P_k^B]+\Tr[\Psi_{AB}P_k^A\otimes (\mathbb{1}_B-P_k^B)]+ \Tr[\Psi_{AB}(\mathbb{1}_A-P_k^A)\otimes (\mathbb{1}_B-P_k^B)]} \\
        &\leqt{(i)}2\sqrt{\Tr[\Psi_{A}(\mathbb{1}_A-P_k^A)]+\Tr[\Psi_{B} (\mathbb{1}_B-P_k^B)]+ \sqrt{\Tr[\Psi_{A}(\mathbb{1}_A-P_k^A)]\Tr[\Psi_{A}(\mathbb{1}_A-P_k^A)]}}\\
        &\le 2\left(\sqrt{ \Tr[\Psi_{A}(\mathbb{1}_A-P_k^A)]}+ \sqrt{\Tr[\Psi_{B}(\mathbb{1}_B-P_k^B)] )}\right)\\
        &\eqt{(ii)} 2\left( c_{\Psi_A}e^{-b_{\Psi_A}k} + c_{\Psi_B}e^{-b_{\Psi_B}k} \right)\,.
    \ee
    Here, in (i), we exploited that $P_k\le \mathbb{1}$ and Cauchy-Schwarz inequality; in (ii), we employed Theorem~\ref{the:expdecay}, which implies that for any Gaussian state $\rho$ there exist constants $c_\rho$ and $b_\rho$ such that $\sqrt{\Tr[\rho(\mathbb{1}-P_k)]}\le c_\rho e^{-b_\rho k}$ for any $k\in\mathbb{N}$. Consequently, we conclude that
    \bb\label{lim_trace_norm_conv}
        \lim\limits_{k\rightarrow\infty} \|\Psi_{ABEE'}^{(k)}-\Psi_{ABEE'}\|_1=0\,,
    \ee
    and that such a convergence is exponentially fast in $k$.

    Note that can rewrite the entropic quantities involved in \eqref{ineq_proof_lower} in terms of $\Psi_{ABEE'}^{(k)}$ as follows:
    \bb
        D\left(\Psi^{(k)}_{A_kEE'}\|\mathbb{1}_{A_k}\otimes \Psi^{(k)}_{EE'}\right)&= D\left(\Psi^{(k)}_{AEE'}\|\mathbb{1}_{A}\otimes \Psi^{(k)}_{EE'}\right)\,,\\
        D_{1/2}(\Psi_{A_kEE'}^{(k)}\|\mathbb{1}_{A_k}\otimes\Psi_{EE'}^{(k)})&=D_{1/2}(\Psi_{AEE'}^{(k)}\|\mathbb{1}_{A}\otimes\Psi_{EE'}^{(k)})\,,\\
        D_{3/2}(\Psi_{A_k EE'}^{(k)}\|\mathbb{1}_{A_k}\otimes\Psi_{EE'}^{(k)})&=D_{3/2}(\Psi_{A EE'}^{(k)}\|\mathbb{1}_{A}\otimes\Psi_{EE'}^{(k)})\,.
    \ee
    Hence, in order to prove the theorem, it suffices to show that
    \bb\label{limit_to_prove}
        \lim\limits_{k\rightarrow\infty} D\left(\Psi^{(k)}_{AEE'}\|\mathbb{1}_{A}\otimes \Psi^{(k)}_{EE'}\right)&=-I_c(A\,\rangle\, B)_{\Psi}\,,\\
        \lim\limits_{k\rightarrow\infty} D_{1/2}(\Psi_{AEE'}^{(k)}\|\mathbb{1}_{A}\otimes\Psi_{EE'}^{(k)})&= -H_{1/2}(A|E)_{\Psi}\,, \\
        \lim\limits_{k\rightarrow\infty} D_{3/2}(\Psi_{A EE'}^{(k)}\|\mathbb{1}_{A}\otimes\Psi_{EE'}^{(k)})&= H_{1/2}(A|B)_{\Psi}\,.
    \ee

    Let us start with the proof of the first limit in \eqref{limit_to_prove}. The term $D\left(\Psi^{(k)}_{AEE'}\|\mathbb{1}_{A}\otimes \Psi^{(k)}_{EE'}\right)$ can be expressed in terms of the conditional entropy as:
    \bb
        D\left(\Psi^{(k)}_{AEE'}\|\mathbb{1}_{A}\otimes \Psi^{(k)}_{EE'}\right)&=-S(\Psi^{(k)}_{AEE'})+S(\Psi^{(k)}_{EE'})\\
        &=-S(\Psi^{(k)}_{B})+S(\Psi^{(k)}_{AB})\\
        &=S(A|B)_{ \Psi^{(k)}_{AB} }\,.
    \ee
    Moreover, by exploiting the monotonicity of the trace norm under partial traces~\cite{MARK} together with the fact that the limit $\Psi_{ABEE'}^{(k)}\xrightarrow{k\rightarrow\infty}  \Psi_{ABEE'}$ holds in trace norm (as proved above in \eqref{lim_trace_norm_conv}), we have that the limit 
    $\Psi_{AB}^{(k)}\xrightarrow{k\rightarrow\infty}  \Psi_{AB}$ holds in trace norm. In addition, since $\Psi_{AB}$ is a Gaussian state, it follows that both $\Psi_{AB}$ and $\Psi_{AB}^{(k)}$ have bounded mean photon number. Consequently, since the conditional entropy is continuous with respect the trace norm over the set of states having bounded mean photon number~\cite{tightuniform}, it follows that 
    \bb\label{lim_sab_k_inf}
        \lim\limits_{k\rightarrow\infty}S(A|B)_{ \Psi^{(k)}_{AB} }=S(A|B)_{ \Psi_{AB} }\,,
    \ee
    and we can thus conclude that 
    \bb
        \lim\limits_{k\rightarrow\infty} D\left(\Psi^{(k)}_{AEE'}\|\mathbb{1}_{A}\otimes \Psi^{(k)}_{EE'}\right)=S(A|B)_{ \Psi }=-I_c(A\,\rangle\, B)_{\Psi}\,.
    \ee
    Let us now prove the second limit in \eqref{limit_to_prove}:
    \bb
        \lim\limits_{k\rightarrow\infty} D_{1/2}(\Psi_{AEE'}^{(k)}\|\mathbb{1}_{A}\otimes\Psi_{EE'}^{(k)})&= -H_{1/2}(A|E)_{\Psi}\,.
    \ee
    Thanks to the definition of conditional Petz-Rényi entropy and to the fact that $\Psi_{AEE'}=\Psi_{AE}\otimes\ketbra{0}_{E'}$, it suffices to prove that
    \bb\label{eq_trace}    \lim\limits_{k\rightarrow\infty}\Tr\left[\sqrt{\Psi_{AEE'}^{(k)}}\, \mathbb{1}_A\otimes\sqrt{\Psi_{EE'}^{(k)}} \right]=\Tr\left[\sqrt{\Psi_{AEE'}} \,\mathbb{1}_A\otimes\sqrt{\Psi_{EE'}} \right]\,.
    \ee
    Note that $\Tr\left[\sqrt{\Psi_{AEE'}^{(k)}}\, \mathbb{1}_A\otimes\sqrt{\Psi_{EE'}^{(k)}} \right]=\Tr\left[\sqrt{\Psi_{AEE'}^{(k)}} P_k^A\otimes\sqrt{\Psi_{EE'}^{(k)}} \right]$ and
    
    \bb
        \Tr\left[\sqrt{\Psi_{AEE'}^{(k)}}\, P_k^A\otimes\sqrt{\Psi_{EE'}^{(k)}} \right]&=\Tr\left[\left(\sqrt{\Psi_{AEE'}^{(k)}}-\sqrt{\Psi_{AEE'}}\right)\, P_k^A\otimes\sqrt{\Psi_{EE'}^{(k)}} \right]
        \\&\quad+\Tr\left[\sqrt{\Psi_{AEE'}}\, \left(P_k^A-\mathbb{1}_A\right)\otimes\sqrt{\Psi_{EE'}} \right]
        \\&\quad +  \Tr\left[\sqrt{\Psi_{AEE'}}\, P_k^A\otimes\left(\sqrt{\Psi_{EE'}^{(k)}}-\sqrt{\Psi_{EE'}}\right) \right]
        \\&\quad + \Tr\left[\sqrt{\Psi_{AEE'}}\, \mathbb{1}_A\otimes\sqrt{\Psi_{EE'} } \right]\,.
    \ee
    Hence, in order to show \eqref{eq_trace}, it suffices to show that
    \bb\label{eq_a_petz12}
        \lim\limits_{k\rightarrow\infty}\left| \Tr\left[\left(\sqrt{\Psi_{AEE'}^{(k)}}-\sqrt{\Psi_{AEE'}}\right)\, P_k^A\otimes\sqrt{\Psi_{EE'}^{(k)}} \right] \right|&=0\,,
    \ee
    \bb\label{eq_b_petz12}
        \lim\limits_{k\rightarrow\infty}\left| \Tr\left[\sqrt{\Psi_{AEE'}}\, \left(P_k^A-\mathbb{1}_A\right)\otimes\sqrt{\Psi_{EE'}} \right]\right|&=0\,,
    \ee
    \bb\label{eq_c_petz12}
        \lim\limits_{k\rightarrow\infty}\left| \Tr\left[\sqrt{\Psi_{AEE'}}\, P_k^A\otimes\left(\sqrt{\Psi_{EE'}^{(k)}}-\sqrt{\Psi_{EE'}}\right) \right] \right|&=0\,.
    \ee
    Let us start by showing the validity of~\eqref{eq_a_petz12}:
    \bb\label{step_proof_eq_a}
        \left| \Tr\left[\left(\sqrt{\Psi_{AEE'}^{(k)}}-\sqrt{\Psi_{AEE'}}\right)\, P_k^A\otimes\sqrt{\Psi_{E}^{(k)}} \right] \right|&\leqt{(iii)}\left\|\sqrt{\Psi_{AEE'}^{(k)}}-\sqrt{\Psi_{AEE'}} \right\|_2 \left\| P_k^A\otimes\sqrt{\Psi_{EE'}^{(k)}} \right\|_2
        \\&\eqt{(iv)} \left\|\sqrt{\Psi_{AEE'}^{(k)}}-\sqrt{\Psi_{AEE'}} \right\|_2 \sqrt{\binom{k+m_A}{m_A}}
        \\&\leqt{(v)} \sqrt{\binom{k+m_A}{m_A}\| \Psi_{AEE'}^{(k)}- \Psi_{AEE'}   \|_1} \\
        &\leqt{(vi)} \sqrt{\binom{k+m_A}{m_A}\| \Psi_{ABEE'}^{(k)}- \Psi_{ABEE'}   \|_1}\\
        &\leqt{(vii)}\sqrt{\binom{k+m_A}{m_A}2\left( c_{\Psi_A}e^{-b_{\Psi_A}k} + c_{\Psi_B}e^{-b_{\Psi_B}k} \right)}\xrightarrow{k\rightarrow\infty} 0\,.
    \ee
    Here, in (iii) we used H\"older's inequality. In (iv), we exploited that $\left\|P_k^A\otimes \sqrt{\Psi_{EE'}^{(k)}}\right\|_2=\sqrt{\Tr[P_k^A]\Tr[\Psi_{EE'}^{(k)}]}=\sqrt{\binom{k+m_A}{m_A}}$, where in the last equality we employed \eqref{dim_hms}. In (v), we used that for any positive semi-definite operators $A,B\ge0 $ it holds that~\cite[Section X.6]{bhatia97}
    \bb\label{bathia_inequality}
        \|A^{1/t}-B^{1/t}\|_{tp}\le \|A-B\|^{1/t}_{p}\quad\text{for all $p,t\ge 1$\,.}
    \ee
    In (vi), we used the monotonicity of the trace norm under partial traces. Finally, in (vii), we exploited \eqref{proof_exp_convergence}.

    Now, let us show~\eqref{eq_b_petz12}. We have that
    \bb
        \left| \Tr\left[\sqrt{\Psi_{AEE'}}\, \left(P_k^A-\mathbb{1}_A\right)\otimes\sqrt{\Psi_{EE'}} \right]\right|&=\left| \Tr\left[\left(P_k^A-\mathbb{1}_A\right)\otimes\sqrt{\Psi_{EE'}} 
 \sqrt{\Psi_{AEE'}}\right]\right|
 \\&=\left|\sum_{\substack{\textbf{i}\in\mathbb{N}^{m_B}\\\sum_{j=1}^{m_B}i_j \not\le k}}\bra{\textbf{i}}_A \Tr_{EE'}\left[ \mathbb{1}_A\otimes\sqrt{\Psi_{EE'}}\, 
 \sqrt{\Psi_{AEE'}}\right]\ket{\textbf{i}}_A  \right|\xrightarrow{k\rightarrow\infty}0\,,
    \ee
where in the last line we exploited that the series
\bb
    \sum_{\substack{\textbf{i}\in\mathbb{N}^{m_B}\\\sum_{j=1}^{m_B}i_j \le k}}\bra{\textbf{i}}_A \Tr_{EE'}\left[ \mathbb{1}_A\otimes\sqrt{\Psi_{EE'}}\, 
 \sqrt{\Psi_{AEE'}}\right]\ket{\textbf{i}}_A
\ee
converges to $\Tr\left[\sqrt{\Psi_{AEE'}}\, \mathbb{1}_A\otimes\sqrt{\Psi_{EE'} } \right]$ for $k\rightarrow\infty$, which is a finite quantity because $\Psi_{AEE'}$ is a Gaussian state. Indeed, for any Gaussian state $\rho_{AB}$ it holds that
\bb
    \Tr[\sqrt{\rho_{AB}}\,\mathbb{1}_A\otimes \sqrt{\rho_B}]\le \Tr[\sqrt{\rho_{AB}}]<\infty\,,
\ee
where in the first inequality we exploited that $ \sqrt{\sigma}\le\mathbb{1}$ for any state $\sigma$, and in the second inequality we employed Lemma~\ref{finite_sqrt}, which states that $\Tr\sqrt{\sigma}<\infty$ for any Gaussian state $\sigma$.

Now, let us show~\eqref{eq_c_petz12}:
\bb
    \left|\Tr\left[\sqrt{\Psi_{AEE'}}\, P_k^A\otimes\left(\sqrt{\Psi_{EE'}^{(k)}}-\sqrt{\Psi_{EE'}}\right) \right] \right|&\leqt{(i)} \left\| \sqrt{\Psi_{AEE'}} \right\|_2\left\| P_k^A\otimes\left( \sqrt{\Psi_{EE'}^{(k)}}-\sqrt{\Psi_{EE'}} \right) \right\|_2
    \\&\eqt{(ii)} \sqrt{\binom{k+m_A}{m_A}}\left\| \sqrt{\Psi_{EE'}^{(k)}}-\sqrt{\Psi_{EE'}} \right\|_2
    \\&\leqt{(iii)} \sqrt{\binom{k+m_A}{m_A}\|\Psi_{EE'}^{(k)}-\Psi_{EE'}\|_1}
    \\&\leqt{(iv)} \sqrt{\binom{k+m_A}{m_A}\|\Psi_{ABEE'}^{(k)}-\Psi_{ABEE'}\|_1}\\
    &\leqt{(v)}\sqrt{\binom{k+m_A}{m_A}2\left( c_{\Psi_A}e^{-b_{\Psi_A}k} + c_{\Psi_B}e^{-b_{\Psi_B}k} \right)}\xrightarrow{k\rightarrow\infty} 0\,.
\ee
Here, in (i), we used H\"older's inequality. In (ii), we exploited the multiplicativity of the 2-norm under tensor product, Eq.~\eqref{dim_hms}, and the fact that $\|\sqrt{\Psi_{AEE'}}\|_2=1$.
In (iii), we used the inequality in \eqref{bathia_inequality}. In (iv), we leveraged the monotonicity of the trace norm under partial traces. Finally, in (v), we exploited \eqref{proof_exp_convergence}. This concludes the proof of the second limit in \eqref{limit_to_prove}.

Now, we only need to prove the third limit in \eqref{limit_to_prove}:
    \bb
        \lim\limits_{k\rightarrow\infty} D_{3/2}(\Psi_{A EE'}^{(k)}\|\mathbb{1}_{A}\otimes\Psi_{EE'}^{(k)})&= H_{1/2}(A|B)_{\Psi}\,.
    \ee
Since the definition of conditional Petz-Rényi entropy and the duality relation in Lemma~\ref{duality_conditional_entropies} imply that
    \bb
        D_{3/2}(\Psi_{A EE'}^{(k)}\|\mathbb{1}_{A}\otimes\Psi_{EE'}^{(k)})=-H_{3/2}(A|EE')_{\Psi^{(k)}}=H_{1/2}(A|B)_{\Psi^{(k)}}\,,
    \ee
we only need to show that
\bb
    \lim\limits_{k\rightarrow\infty} H_{1/2}(A|B)_{\Psi^{(k)}}=  H_{1/2}(A|B)_{\Psi}\,,
\ee
which is equivalent to prove that
\bb    
    \lim\limits_{k\rightarrow\infty}\Tr\left[\sqrt{\Psi_{AB}^{(k)}}\, \mathbb{1}_A\otimes\sqrt{\Psi_{B}^{(k)}} \right]=\Tr\left[\sqrt{\Psi_{AB}} \,\mathbb{1}_A\otimes\sqrt{\Psi_{B}} \right]\,.
\ee
The proof of this limit is exactly the same as the limit in \eqref{eq_trace}, i.e.~$\lim\limits_{k\rightarrow\infty}\Tr\left[\sqrt{\Psi_{AEE'}^{(k)}}\, \mathbb{1}_A\otimes\sqrt{\Psi_{EE'}^{(k)}} \right]=\Tr\left[\sqrt{\Psi_{AEE'}} \,\mathbb{1}_A\otimes\sqrt{\Psi_{EE'}} \right]$. Indeed, to prove the latter, we only exploited the fact that $\Psi_{AEE'}$ is Gaussian and that the convergence $\Psi_{ABEE'}^{(k)}\xrightarrow{k\rightarrow\infty}\Psi_{ABEE'}$ is exponentially fast in trace norm.
\end{proof}
\subsection{Two-way quantum capacity and secret-key capacity}
In this paragraph, we establish a lower bound on the $n$-shot two-way quantum capacity and one-shot secret-key capacity of Gaussian channels, presented in the forthcoming Theorem~\ref{thm_lower_bound}.

We start by stating the following result regarding the \emph{one-shot distillable entanglement} and the \emph{one-shot distillable key} of \emph{finite-dimensional} states~\cite{Sumeet_book}. Roughly speaking, the one-shot distillable entanglement of a bipartite quantum state $\rho_{AB}$ is the maximum number of ebits that can be distilled from $\rho_{AB}$ with an error not exceeding $\varepsilon$ by exploiting LOCCs. Analogously, the one-shot distillable key of a bipartite quantum state $\rho_{AB}$ is the maximum number of secret-key bits that  can be generated from $\rho_{AB}$ with an error not exceeding $\varepsilon$ by exploiting LOCCs. We refer to~\cite{Sumeet_book} for rigorous definitions of one-shot distillable entanglement and one-shot distillable key. 
\begin{lemma}[(Lower bound on the one-shot distillable entanglement and one-shot distillable key~\cite{Sumeet_book})]\label{Lemma_Sumeet_book_2}
    Let $\rho_{AB}$ be a bipartite quantum state on finite-dimensional systems $A$ and $B$. For all $\varepsilon\in(0,1)$ and $\eta\in[0,\sqrt{\varepsilon})$, the one-shot distillable key $E_K^{(\varepsilon)}$ and the one-shot distillable entanglement $E_D^{(\varepsilon)}$ of $\rho_{AB}$ satisfy the following inequality:
    \bb\label{rate_dist}
        E^{(\varepsilon)}_K(\rho) \ge E^{(\varepsilon)}_D(\rho)\ge -H_{\max}^{\sqrt{\varepsilon}-\eta} (A|B)_{\rho}+4\log_2\eta\,.
    \ee 
\end{lemma}
Thanks to the above lemma, we can easily derive a lower bound on the $n$-shot two-way quantum capacity and $n$-shot secret-key capacity of finite-dimensional quantum channels.
\begin{lemma}[(Lower bound on the $n$-shot two-way quantum capacity and one-shot secret-key capacity of finite-dimensional channels)]\label{lemma_lower_2_way}
    Let $\NN_{A'\to B} $ be a quantum channel, where $A'$ and $B$ denote finite-dimensional systems representing the input and the output of the channel, respectively. Let $\varepsilon\in(0,1)$, $\eta\in[0,\sqrt{\varepsilon})$, $n\in\mathbb{N}$, and let $\psi_{AA'}$ be a bipartite pure state on the input $A'$ and on an arbitrary ancillary system $A$. Then, the $n$-shot two-way quantum capacity and the $n$-shot secret-key capacity of $\NN$ satisfy the following inequalities:
    \bb\label{first_lower}
        K^{(\varepsilon,n)}(\NN)&\ge Q_2^{(\varepsilon,n)}(\NN)\ge -H_{\max}^{\sqrt{\varepsilon}-\eta} (A^n|B^n)_{\omega_{AB}^{\otimes n}}+4\log_2\eta\,,
    \ee
    \bb\label{second_lower}
        K^{(\varepsilon,n)}(\NN)&\ge Q_2^{(\varepsilon,n)}(\NN)\ge -H_{\max}^{\sqrt{\varepsilon}-\eta} (B^n|A^n)_{\omega_{AB}^{\otimes n}}+4\log_2\eta\,,
    \ee 
    where $\omega_{AB}\coloneqq\NN_{A'\to B}(\psi_{AA'})$, while $A^n$ and $B^n$ denote $n$ copies of systems $A$ and $B$, respectively.
\end{lemma}
\begin{proof}
    The inequality $K^{(\varepsilon,n)}(\NN)\ge Q_2^{(\varepsilon,n)}(\NN)$ always holds due to Eq.~\eqref{ineq_nshot_q_k}. To prove that Eq.~\eqref{first_lower}, consider the following $(n,|M|,\varepsilon)$ LOCC-assisted quantum communication protocol over $\NN_{A'\to B}$, where $\log_2|M|\ge -H_{\max}^{\sqrt{\varepsilon}-\eta} (A^n|B^n)_{\omega_{AB}^{\otimes n}}+4\log_2\eta$:
    \begin{itemize}
        \item Step (i): Alice sends the half $A'$
        of the state $\psi_{AA'}$  to Bob via the channel $\NN_{A'\to B}$ for $n$ times. This results in them sharing the state $\omega_{AB}^{\otimes n}$.
        \item Step (ii): Alice and Bob apply the one-shot entanglement distilation protocol on $\omega_{AB}^{\otimes n}$ that achieves the performance stated in Lemma~\ref{Lemma_Sumeet_book_2}. 
    \end{itemize}
Since $Q_2^{(\varepsilon,n)}$ is defined as the maximum $\log_2|M|$ achievable by any $(n,|M|,\varepsilon)$ LOCC-assisted quantum communication protocol over $\NN_{A'\to B}$, it follows that Eq.~\eqref{first_lower} holds.

To prove \eqref{second_lower}, it suffices to consider a similar $(n,|M|,\varepsilon)$ LOCC-assisted quantum communication protocol over $\NN_{A'\to B}$, where in step (ii), Alice and Bob exchange their roles. 
\end{proof}
We are now ready to present our lower bounds.
\begin{thm}[(Lower bound on the $n$-shot two-way quantum capacity and $n$-shot secret-key capacity of Gaussian channels)]\label{thm_lower_bound_2way}
Let $\NN$ be a single-mode Gaussian channel. Let $\varepsilon\in(0,1)$, $\eta\in[0,\sqrt{\varepsilon})$, and $n\in\N$ such that $n\ge2\log_2\!\left(\frac{2}{\varepsilon^2}\right)$. For any pure Gaussian state $\Phi_{AA'}$, the $n$-shot two-way quantum capacity and the $n$-shot secret-key capacity of $\NN$ satisfy the following lower bounds:
    \bb\label{lower_bound_CV_2way}
        K^{(\varepsilon,n)}(\NN)&\ge Q_2^{(\varepsilon,n)}(\NN)\ge n I_c(A\,\rangle\, B)_{\Psi}-\sqrt{n}4\log_2({\bar{\eta}_1})\sqrt{\log_2\!\left(\frac{2}{\left(\sqrt{\varepsilon}-\eta \right)^2}\right)}+4\log_2\eta\,,\\
        K^{(\varepsilon,n)}(\NN)&\ge Q_2^{(\varepsilon,n)}(\NN)\ge n I_c(B\,\rangle\, A)_{\Psi}-\sqrt{n}4\log_2({\bar{\eta}_2})\sqrt{\log_2\!\left(\frac{2}{\left(\sqrt{\varepsilon}-\eta \right)^2}\right)}+4\log_2\eta\,,
    \ee
    where
    \bb\label{eq_bar_eta_2way}
        \bar{\eta}_1&\coloneqq \sqrt{2^{H_{1/2}(A|B)_{\Psi}}}+\sqrt{2^{H_{1/2}(A|E)_{\Psi}}}+1\,,\\
        \bar{\eta}_2&\coloneqq \sqrt{2^{H_{1/2}(B|A)_{\Psi}}}+\sqrt{2^{H_{1/2}(B|E)_{\Psi}}}+1\,,\\
        \Psi_{ABE}&\coloneqq V_{A'\to BE}\Phi_{AA'}(V_{A'\to BE})^{\dagger}\,,
    \ee
    and $V_{A'\to BE}$ is a Stinespring dilation isometry associated with $\NN_{A'\to B}$. Moreover, 
    \bb
        I_c(A\,\rangle\, B)_{\Psi}&\coloneqq S(B)_\Psi-S(AB)_{\Psi}\,,\\
        I_c(B\,\rangle\, A)_{\Psi}&\coloneqq S(A)_\Psi-S(AB)_{\Psi}\,,\\ 
        H_{1/2}(A|B)_{\Psi}&\coloneqq 2\log_2\!\left(\Tr\sqrt{\Psi_{AB}}\,\mathbb{1}_A\otimes\sqrt{\Psi_B}\right)\,,\\
        H_{1/2}(B|A)_{\Psi}&\coloneqq 2\log_2\!\left(\Tr\sqrt{\Psi_{AB}}\,\sqrt{\Psi_A}\otimes \mathbb{1}_B\right)\,,\\
        H_{1/2}(A|E)_{\Psi}&\coloneqq 2\log_2\!\left(\Tr\sqrt{\Psi_{AE}}\,\mathbb{1}_A\otimes\sqrt{\Psi_E}\right)\,,\\
        H_{1/2}(B|E)_{\Psi}&\coloneqq 2\log_2\!\left(\Tr\sqrt{\Psi_{BE}}\,\mathbb{1}_B\otimes\sqrt{\Psi_E}\right)\,,
    \ee
    are finite quantities denoting coherent information, reverse coherent information, and conditional Petz Rényi entropies. In particular, by choosing $\eta=\frac{\sqrt{\varepsilon}}{2}$, we get
    \bb\label{b_lower_bound_CV_Q2_1}
        K^{(\varepsilon,n)}(\NN)&\ge Q_2^{(\varepsilon,n)}(\NN)\ge nI_c(A\,\rangle\, B)_{\Psi}-\sqrt{n}4\log_2({\bar{\eta}_1})\sqrt{\log_2\!\left(\frac{8}{\varepsilon}\right)}-\log_2\!\left(\frac{16}{\varepsilon^2}\right)\,,\\
        K^{(\varepsilon,n)}(\NN)&\ge Q_2^{(\varepsilon,n)}(\NN) \ge nI_c(B\,\rangle\, A)_{\Psi}-\sqrt{n}4\log_2({\bar{\eta}_2})\sqrt{\log_2\!\left(\frac{8}{\varepsilon}\right)}-\log_2\!\left(\frac{16}{\varepsilon^2}\right)\,.
    \ee
In particular, if the input state $\Phi_{AA'}$ satisfies the energy constraint $\Tr[\hat{N}_{A'}\Phi_{AA'}]\le N_s$, where $\hat{N}_{A'}$ is the photon number operator on $A'$, then the right-hand sides of \eqref{lower_bound_CV2} constitute also a lower bound on the energy-constrained $n$-shot two-way quantum capacity $Q_2^{(\varepsilon,n)}(\NN,N_s)$ and the energy-constrained $n$-shot secret-key capacity $K^{(\varepsilon,n)}(\NN,N_s)$.
\end{thm}
\begin{proof}
    The proof exploits Lemma~\ref{lemma_lower_2_way} and it is entirely analogous to the proof of Theorem~\ref{thm_lower_bound}.
\end{proof}

\subsection{Conditional Petz-Rényi entropy of Gaussian states}\label{section_petz}
The lower bound on the one-shot capacities proved in Theorem~\ref{thm_lower_bound} and Theorem~\ref{thm_lower_bound_2way} involves the calculation of the conditional Petz-Rényi entropy of Gaussian states. Therefore, in order to explicitly calculate such bounds for specific examples of Gaussian channels, we need to establish a closed formula for the conditional Petz-Rényi entropy of any Gaussian state in terms of its first moment and covariance matrix. We provide such a closed formula in the forthcoming lemma.
\begin{lemma}[(Conditional Petz-Rényi entropy with $\alpha=1/2$ of Gaussian states)]\label{thm_cond_petz_Gaussian}
    Let $A$ be an $m$-mode system and let $B$ be an $n$-mode system. Let $\rho_{AB}$ be a Gaussian state on $AB$. Then, its conditional Petz-Rényi entropy with $\alpha=1/2$ can be calculated as:
\bb 
    H_{1/2}(A|B)_{\rho_{AB}}=\log_2\left(\frac{\sqrt{\det\!\left(V_{\mathrm{sqrt}}(\rho_{AB})\right)\,\det\!\left( V_{\mathrm{sqrt}}(\rho_B)\right)}}{ \det\!\left(   
  \frac{ V_{\mathrm{sqrt}}(\rho_{AB})\big|_B + V_{\mathrm{sqrt}}(\rho_B)}{2}\right)}\right)\,,
\ee
where we denoted 
\bb
    V_{\mathrm{sqrt}}(\sigma)&\coloneqq \left[\mathbb{1}+\sqrt{\mathbb{1}- (iV(\sigma)\Omega)^{-2} }\right]V(\sigma)\qquad\forall\,\sigma\,,
    \ee
and $V_{\mathrm{sqrt}}(\rho_{AB})\big|_B$ is the $2n\times 2n$ bottom-right block of the matrix $V_{\mathrm{sqrt}}(\rho_{AB})$.
\end{lemma}
Before proving the above lemma, let us first understand how to explicitly calculate \( V_{\mathrm{sqrt}}(\sigma) \) given \( V(\sigma) \).
\begin{remark}[(Calculation of $V_{\mathrm{sqrt}}$)]
    The definition of $V_{\mathrm{sqrt}}(\sigma)$ includes the term $\sqrt{\mathbb{1}- (iV(\sigma)\Omega)^{-2} }$, which involves applying a function to a (non-Hermitian) matrix. To compute this term, it suffices to diagonalise the matrix $iV(\sigma)\Omega$:\bb
        \sqrt{\mathbb{1}- (iV(\sigma)\Omega)^{-2} }=X\sqrt{\mathbb{1}- D^{-2} }X^{-1}\,,
    \ee
    where $iV(\sigma)\Omega=XDX^{-1}$ is the eigendecomposition of $iV(\sigma)\Omega$. The latter can be obtained e.g.~using the Williamson's decomposition of the covariance matrix $V(\sigma)$, where $D$ is the diagonal matrix of symplectic eigenvalues satisfying $D\ge \mathbb{1}$. In particular, the latter condition implies that the diagonal matrix $\mathbb{1}- D^{-2}$ is real and positive. Consequently, the matrix square root $\sqrt{\mathbb{1}- D^{-2} }$ is well-defined.
\end{remark}
\begin{proof}[Proof of Lemma~\ref{thm_cond_petz_Gaussian}]
The conditional Petz-Rényi entropy with $\alpha=1/2$ of $\rho_{AB}$ can be expressed in terms of the overlap between the square roots of $\rho_{AB}$ and $\mathbb{1}_A\otimes\rho_B$ as
    \bb
        H_{1/2}(A|B)_{\rho_{AB}}\coloneqq -D_{1/2}(\rho_{AB}\|\mathbb{1}_A\otimes\rho_B)=2\log_2\!\left(\Tr[\sqrt{\rho_{AB}}\,\sqrt{\mathbb{1}_A\otimes\rho_B}] \right)\,.
    \ee
The overlap between two Gaussian states $\sigma_1$ and $\sigma_2$ is given by~\cite{BUCCO}
        \bb
            \Tr[\sigma_1\sigma_2]=\frac{1}{\sqrt{\det\!\left(\frac{V(\sigma_1)+V(\sigma_2)}{2}\right)}}e^{-\left(\textbf{m}(\sigma_1)-\textbf{m}(\sigma_2)\right)^\intercal \left(V(\sigma_1)+ V(\sigma_2)\right)^{-1}\left(\textbf{m}(\sigma_1)-\textbf{m}(\sigma_2)\right)}\,.
        \ee
        Moreover, the normalised square root $\frac{\sqrt{\sigma}}{\Tr\sqrt{\sigma}}$ of a Gaussian state $\sigma$ is a Gaussian state with the same first moment as $\sigma$ and covariance matrix given by~\cite[Lemma~2]{Holevo1972} (see also~\cite[Section~3]{Paraoanu_2000} and \cite{Banchi_2015}):
        \bb\label{eq_v_sqrt0000}
         V_{\mathrm{sqrt}}(\sigma)\coloneqq V\!\left(\frac{\sqrt{\sigma}}{\Tr\sqrt{\sigma}}\right)=\left[\mathbb{1}+\sqrt{\mathbb{1}- (iV(\sigma)\Omega)^{-2} }\right]V(\sigma)\,.
        \ee
        Consequently, we obtain that
        \bb
             1=\Tr\!\left[\frac{\sqrt{\sigma}}{\Tr\sqrt{\sigma}}\frac{\sqrt{\sigma}}{\Tr\sqrt{\sigma}}\right]\left(\Tr\sqrt{\sigma}\right)^2=\frac{\left(\Tr\sqrt{\sigma}\right)^2}{\sqrt{\det\!\left(V\left(\frac{\sqrt{\sigma}}{\Tr\sqrt{\sigma}}\right)\right)}}=\frac{\left(\Tr\sqrt{\sigma}\right)^2}{\sqrt{\det\!\left( V_{\mathrm{sqrt}}(\sigma) \right)}}\,.
\ee
In particular, this implies that
\bb\label{trace_sqrt_gaussian}
    \Tr\sqrt{\sigma}=\left[\det\!\left(V_{\mathrm{sqrt}}(\sigma)\right)\right]^{1/4}\,.
\ee
Hence, we obtain that
    \bb
    \Tr[\sqrt{\rho_{AB}}\sqrt{\mathbb{1}_A\otimes \rho_B} ]&=\Tr[\sqrt{\rho_{AB}}]\,\Tr[\sqrt{\rho_{B}}] \Tr_B\!\left[  \frac{\Tr_A\!\sqrt{\rho_{AB}}}{\Tr[\sqrt{\rho_{AB}}]}\frac{\sqrt{\rho_B}}{\Tr\sqrt{\rho_B}}\right]\\
    &= \left[\det\!\left(V_{\mathrm{sqrt}}(\rho_{AB})\right)  \,
 \det\!\left(V_{\mathrm{sqrt}}(\rho_{B})\right)\right]^{1/4} \Tr_B\!\left[  \frac{\Tr_A\!\sqrt{\rho_{AB}}}{\Tr[\sqrt{\rho_{AB}}]}\frac{\sqrt{\rho_B}}{\Tr\sqrt{\rho_B}}\right]\\
 &=\frac{\left[\det\!\left(V_{\mathrm{sqrt}}(\rho_{AB})\right)  \,
 \det\!\left(V_{\mathrm{sqrt}}(\rho_{B})\right)\right]^{1/4}}{  \left(  \frac{V_{\mathrm{sqrt}}(\rho_{AB})\big|_B + V_{\mathrm{sqrt}}(\rho_{B})}{2} \right)^{1/2}  }\,,
    \ee
where in the last line we exploited again the formula for the overlap between two Gaussian states and the fact that the two Gaussian states $\frac{\Tr_A\!\sqrt{\rho_{AB}}}{\Tr[\sqrt{\rho_{AB}}]}$ and $\frac{\sqrt{\rho_B}}{\Tr\sqrt{\rho_B}}$ have the same first moments. Consequently, we conclude that
\bb
H_{1/2}(A|B)_{\rho_{AB}}\coloneqq 2 \log_2 \left(\Tr\left[\sqrt{\rho_{AB}}\sqrt{\mathbb{1}_A\otimes \rho_B}\right]\right)=\log_2\left(\frac{\sqrt{\det\!\left(V_{\mathrm{sqrt}}(\rho_{AB})\right)\,\det\!\left( V_{\mathrm{sqrt}}(\rho_B)\right)}}{ \det\!\left(   
  \frac{ V_{\mathrm{sqrt}}(\rho_{AB})\big|_B + V_{\mathrm{sqrt}}(\rho_B)}{2}\right)}\right)\,.
\ee
\end{proof}

\newpage

\section{Non-asymptotic analysis of quantum communication across the pure loss channel}\label{sec_lower_bound_pure_loss_channel}
This section aims to establish easily computable lower bounds on the $n$-shot quantum capacity, two-way quantum capacity, and secret-key capacity of the pure loss channel~\cite{BUCCO} (see Definition~\ref{def_pure_loss} for the definition of pure loss channel). This is the most important Gaussian channel as it constitutes the most realistic model to describe optical links, such as optical fibres and free-space links. For each capacity, we will derive two distinct lower bounds: the first is based on the results we obtained in the previous Section~\ref{section_AEP0} exploiting the asymptotic equipartition property; the second is based on the analysis of the \emph{entropy variance}~\cite{Kaur_2017} presented in the forthcoming Subsection~\ref{section_entropy_cariance}.

\subsection{Asymptotic equipartition property approach}\label{subsecAsymptotic}
The goal of this section is to evaluate the lower bounds on the $n$-shot capacities proved in Theorem~\ref{thm_lower_bound} and Theorem~\ref{thm_lower_bound_2way} for the physically interesting case of the pure loss channel $\mathcal{E}_\lambda$.

Let us apply such theorems by choosing as input state $\Phi_{AA'}$ the two-mode squeezed vacuum state $\ketbra{\Psi_{N_s}}_{AA'}$ with local mean photon number equal to $N_s$, as defined in Eq.~\eqref{eq:2mode_squeezed}. Moreover, let us recall that the pure loss channel $\pazocal{E}_{\lambda}$, defined in Definition~\ref{def_pure_loss}, is a single-mode Gaussian channel that admits the following Stinespring dilation:
\bb
        \pazocal{E}_{\lambda}(\cdot)=\Tr_E\left[U_\lambda^{A'E\to BE} \big((\cdot) \otimes\ketbra{0}_E\big) (U_\lambda^{A'E\to BE})^\dagger\right]\,,
\ee
where $U_\lambda^{A'E\to BE}$ denotes the beam splitter unitary defined in Definition~\ref{def_beam_splitter}. In order to apply Theorem~\ref{thm_lower_bound} and Theorem~\ref{thm_lower_bound_2way}, we need to calculate the covariance matrix $V\!\left(\Psi_{ABE}^{(\lambda,N_s)}\right)$ of the tri-partite state
\bb
    \Psi_{ABE}^{(\lambda,N_s)}\coloneqq\left(\mathbb{1}_A\otimes U_\lambda^{BE}\right)\,\left( \ketbra{\Psi_{N_s}}_{AB}\otimes\ketbra{0}_{E}\right) \left(\mathbb{1}_A\otimes U_\lambda^{BE}\right)^\dagger\,,
\ee
e.g.~with respect the mode-ordering $(A,B,E)$. By exploiting Eq.~\eqref{law_sympl_unitary} and Eq.~\eqref{symplectic_beam_splitter}, it follows that $V\!\left(\Psi_{ABE}^{(\lambda,N_s)}\right)$ can be expressed as
\bb\label{cov_ABEE_comp}
V\!\left(\Psi_{ABE}^{(\lambda,N_s)}\right)=  
\left(\mathbb{1}_2\oplus {S}_\lambda\right) \left(V(\ketbra{\Psi_{N_s}}_{AB})\oplus  V(\ketbra{0})\right) \left(\mathbb{1}_2\oplus {S}_\lambda^{\intercal}\right) \,,
\ee
where:
\bb
S_\lambda &\coloneqq \left(\begin{matrix} \sqrt{\lambda}\,\mathbb{1}_2  & \sqrt{1-\lambda}\,\mathbb{1}_2 \\
-\sqrt{1-\lambda}\,\mathbb{1}_2 & \sqrt{\lambda}\,\mathbb{1}_2 \end{matrix}\right)\,, 
\ee
is the symplectic matrix associated with the beam splitter unitary, $V(\ketbra{\Psi_{N_s}})$ is the covariance matrix of the two-mode squeezed vacuum state given in Eq.~\eqref{cov_two_mode_squeezed}, and $V(\ketbra{0})=\mathbb{1}_2$ is the covariance matrix of the vacuum state. By explicitly performing the matrix multiplications, one obtains that 
\bb
{\fontsize{8.3}{8.3}\selectfont
    V\!\left(\Psi_{ABE}^{(\lambda,N_s)}\right)=
    \begin{bmatrix}
        1 + 2 N_s & 0 &  2 \sqrt{\lambda N_s (1 + N_s)} & 0 & -2 \sqrt{(1 - \lambda) N_s (1 + N_s)} & 0\\
        0 & 1 + 2 N_s & 0 & -2 \sqrt{ \lambda N_s (1 + N_s)} & 0 & 2 \sqrt{(1 - \lambda) N_s (1 + N_s)} \\
        2 \sqrt{ \lambda N_s (1 + N_s)} & 0 & 1 + 2\lambda N_s & 0 & -2 \sqrt{ (1 - \lambda) \lambda} N_s & 0 \\
        0 & -2 \sqrt{ \lambda N_s (1 + N_s)} & 0 & 1 + 2\lambda N_s & 0 & -2 \sqrt{ (1 - \lambda) \lambda} N_s \\    
        -2 \sqrt{(1 - \lambda) N_s (1 + N_s)} & 0 & -2 \sqrt{ (1 - \lambda) \lambda} N_s & 0 & 1 + 2(1-\lambda) N_s & 0 \\
        0 & 2 \sqrt{(1 - \lambda) N_s (1 + N_s)} & 0 & -2 \sqrt{ (1 - \lambda) \lambda} N_s & 0 & 1 + 2(1-\lambda) N_s \\
    \end{bmatrix}\,.}
\ee
Hence, by exploiting Lemma~\ref{thm_cond_petz_Gaussian}, one can explicitly calculate all the conditional Petz-Rényi entropies involved in Theorem~\ref{thm_lower_bound} and Theorem~\ref{thm_lower_bound_2way}. By a simple calculation performed on the Mathematica notebook attached to this manuscript, one obtains:
\bb\label{ingredient1_ec}
H_{1/2}(A|B)_{\Psi_{ABE}^{(\lambda,N_s)}}&=   \log_2\left(\frac{\left[1 + 2 (1 - \lambda) N_s + 2 \sqrt{(1 - \lambda) N_s [1 + (1 - \lambda) N_s]}\right]\left[1 + 2 \lambda N_s + 2 \sqrt{\lambda N_s (1 + \lambda N_s)}\right]}{\left[1 + \sqrt{\lambda N_s (1 + \lambda N_s)} + \lambda \left(2 N_s + \sqrt{\frac{(1 - \lambda) N_s^3}{1 + (1 - \lambda) N_s}}\right)\right]^2 } \right)\,,\\
    H_{1/2}(A|E)_{\Psi_{ABE}^{(\lambda,N_s)}}&=\log_2\left(\frac{\left[1 +2 (1 - \lambda) N_s + 2 \sqrt{(1 - \lambda) N_s (1 + (1 - \lambda) N_s)}\right]\left[1 + 2 \lambda N_s + 2 \sqrt{\lambda N_s (1 + \lambda N_s)}\right]}{\left[1  + \sqrt{(1 - \lambda) N_s (1 + (1 - \lambda) N_s)} + (1-\lambda)\left(2N_s+\sqrt{\frac{\lambda N_s^3}{1 + \lambda N_s}}\right)\right]^2 }\right)\,,\\
    H_{1/2}(B|A)_{\Psi_{ABE}^{(\lambda,N_s)}}&=\log_2\left(\frac{\left[1 + 2 N_s + 2 \sqrt{N_s (1 + N_s)}\right]\left[1 + 2 (1 - \lambda) N_s + 2 \sqrt{(1 - \lambda) N_s (1 + (1 - \lambda) N_s)}\right]}{\left[1 + 2 N_s + \sqrt{N_s (1 + N_s)} + (1 + N_s) \sqrt{1 - \frac{1}{1 + (1 - \lambda) N_s}}\right]^2}\right)\,,\\
    H_{1/2}(B|E)_{\Psi_{ABE}^{(\lambda,N_s)}}&= \log_2\left(\frac{\left(1 + 2 (1 - \lambda) N_s + 2 \sqrt{(1 - \lambda) N_s (1 + (1 - \lambda) N_s)}\right)(1 + 2 N_s + 2 \sqrt{N_s (1 + N_s)})}{\left[1 + (1- \lambda)\left(2 N_s + \sqrt{N_s (1 + N_s)}\right)  + \sqrt{(1 - \lambda) N_s (1 + (1 - \lambda) N_s)}\right]^2}\right)\,.
\ee
Moreover, thanks to Eq.~\eqref{q_ec_pure_loss} and Eq.~\eqref{q2_ec_pure_loss}, we can also obtain the coherent information quantities appearing in Theorem~\ref{thm_lower_bound} and Theorem~\ref{thm_lower_bound_2way}:
\bb\label{ingredient2_ec}
    I_c(A\,\rangle\, B)_{ \Psi_{ABE}^{(\lambda,N_s)}  } &=h(\lambda N_s)-h\!\left((1-\lambda)N_s\right)\,,\\
    I_c(B\,\rangle\, A)_{ \Psi_{ABE}^{(\lambda,N_s)}  } &=h(N_s)-h\!\left((1-\lambda)N_s\right)\,,
\ee
where $h(x)\coloneqq (x+1)\log_2(x+1)-x\log_2x$. By taking the limit of infinite energy $N_s\rightarrow\infty$, one obtains
\bb\label{ingredient1}
    \lim\limits_{N_s\rightarrow\infty}H_{1/2}(A|B)_{\Psi_{ABE}^{(\lambda,N_s)}}&=\log_2\!\left(\frac{1-\lambda}{\lambda}\right)\,,\\
    \lim\limits_{N_s\rightarrow\infty}H_{1/2}(A|E)_{\Psi_{ABE}^{(\lambda,N_s)}}&=-\log_2\!\left(\frac{1-\lambda}{\lambda}\right)\,,\\
    \lim\limits_{N_s\rightarrow\infty}H_{1/2}(B|A)_{\Psi_{ABE}^{(\lambda,N_s)}}&=\log_2\!\left(1-\lambda\right)\,,\\
    \lim\limits_{N_s\rightarrow\infty}H_{1/2}(B|E)_{\Psi_{ABE}^{(\lambda,N_s)}}&=-\log_2\!\left(1-\lambda\right)\,.
\ee
and, thanks to Eq.~\eqref{coh_1} and Eq.~\eqref{coh_2}, one also obtains that
\bb\label{ingredient2}
    \lim\limits_{N_s\rightarrow\infty}I_c(A\,\rangle\, B)_{ \Psi_{ABE}^{(\lambda,N_s)}  } &=\log_2\!\left(\frac{\lambda}{1-\lambda}\right)\,,\\
    \lim\limits_{N_s\rightarrow\infty}I_c(B\,\rangle\, A)_{ \Psi_{ABE}^{(\lambda,N_s)}  } &=\log_2\!\left(\frac{1}{1-\lambda}\right)\,.
\ee
By putting together Eq.~\eqref{ingredient1}, Eq.~\eqref{ingredient2}, and Theorem~\ref{thm_lower_bound}, we can obtain a lower bound on the $n$-shot quantum capacity of the pure loss channel.
\begin{thm}[(Lower bound on the $n$-shot quantum capacity of the pure loss channel)]\label{thm_one_shot_Q_pure}
    Let $\lambda\in(0,1)$, $\varepsilon\in(0,1)$, and $n\in\N$ such that $n\ge2\log_2\!\left(\frac{2}{\varepsilon^2}\right)$. The $n$-shot quantum capacity of the pure loss channel satisfies:
\bb\label{lower_bound_q_pure_loss}
    Q^{(\varepsilon,n)}(\mathcal{E}_\lambda)&\ge nQ(\mathcal{E}_\lambda)-\sqrt{n}4\log_2\!\left( \sqrt{\frac{1-\lambda}{\lambda}}+ \sqrt{\frac{\lambda}{1-\lambda}} +1 \right)\sqrt{\log_2\!\left(\frac{2^9}{\varepsilon^2}\right)}-\log_2\!\left(\frac{2^{18}}{3\varepsilon^4}\right)\,,
\ee
where 
\bb
    Q(\mathcal{E}_{\lambda})&=   
        \begin{cases}
        \log_2\!\left(\frac{\lambda}{1-\lambda}\right) &\text{if $\lambda\in(\frac{1}{2},1]$ ,} \\
        0 &\text{if $\lambda\in[0,\frac{1}{2}]$ ,}
    \end{cases} 
\ee
is the (asymptotic) quantum capacity of the pure loss channel.
\end{thm}
As a consequence of the above theorem, $n$ uses of the pure loss channel $\mathcal{E}_\lambda$ allows one to transmit a number of qubits equal to the right hand of Eq.~\eqref{lower_bound_q_pure_loss} with an error not exceeding $\varepsilon$. Here, the error is measured as in Eq.~\eqref{eq_def_Q_cap} in terms of the channel fidelity. 

By combining Eq.~\eqref{ingredient1}, Eq.~\eqref{ingredient2}, and Theorem~\ref{thm_lower_bound_2way}, we can obtain a lower bound on the $n$-shot two-way quantum capacity and $n$-shot secret-key capacity of the pure loss channel.
\begin{thm}[(Lower bound on the $n$-shot two-way quantum capacity and $n$-shot secret-key capacity of the pure loss channel)]\label{thm_one_shot_Q2_pure}
    Let $\lambda\in(0,1)$, $\varepsilon\in(0,1)$, and $n\in\N$ such that $n\ge2\log_2\!\left(\frac{2}{\varepsilon^2}\right)$. The $n$-shot two-way quantum capacity and the $n$-shot secret-key capacity of the pure loss channel satisfy:
\bb\label{lower_bound_q_pure_loss2}
K^{(\varepsilon,n)}(\mathcal{E}_\lambda)\ge Q_2^{(\varepsilon,n)}(\mathcal{E}_\lambda) &\ge nQ_2(\mathcal{E}_\lambda)-\sqrt{n}4\log_2\!\left(\sqrt{1-\lambda}+\sqrt{\frac{1}{1-\lambda}}+1\right)\sqrt{\log_2\!\left(\frac{8}{\varepsilon}\right)}-\log_2\!\left(\frac{16}{\varepsilon^2}\right)\,,
\ee
where $Q_2(\mathcal{E}_\lambda)\coloneqq \log_2\!\left(\frac{1}{1-\lambda}\right)$ is the (asymptotic) two-way quantum capacity of the pure loss channel.
\end{thm}
As a result of the above theorem, $n$ uses of the pure loss channel $\mathcal{E}_\lambda$ and arbitrary LOCCs allow one to generate a number of ebits and a number of secret-key bits equal to the right hand of Eq.~\eqref{lower_bound_q_pure_loss2} with an error not exceeding $\varepsilon$.

By using \eqref{ingredient1_ec}, \eqref{ingredient2_ec},Theorem~\ref{thm_lower_bound}, and Theorem~\ref{thm_lower_bound_2way}, we can also find a lower bound on the energy-constrained $n$-shot capacities of the pure loss channel, as stated in the following theorem.
\begin{thm}[(Lower bound on the energy-constrained $n$-shot capacities of the pure loss channel)]\label{thm_lower_ec_aep}
    Let $N_s>0$, $\lambda\in(0,1)$, $\varepsilon\in(0,1)$, and $n\in\N$ such that $n\ge2\log_2\!\left(\frac{2}{\varepsilon^2}\right)$. The energy-constrained $n$-shot quantum capacity of the pure loss channel satisfies:
    \bb
        Q^{(\varepsilon,n)}(\mathcal{E}_\lambda,N_s)\ge  n \left(h(\lambda N_s)-h\!\left((1-\lambda) N_s\right)\right)-\sqrt{n}4\log_2\left(\sqrt{2^{H_{1/2}(A|B)_{\Psi_{ABE}^{(\lambda,N_s)}}}}+\sqrt{2^{H_{1/2}(A|E)_{\Psi_{ABE}^{(\lambda,N_s)}}}}+1\right)\sqrt{\log_2\!\left(\frac{2^9}{\varepsilon^2}\right)}-\log_2\!\left(\frac{2^{18}}{3\varepsilon^4}\right)\,,
    \ee
    where $H_{1/2}(A|B)_{\Psi_{ABE}^{(\lambda,N_s)}}$ and $H_{1/2}(A|E)_{\Psi_{ABE}^{(\lambda,N_s)}}$ are reported in \eqref{ingredient1_ec}. Moreover, the energy-constrained $n$-shot two-way quantum and secret-key capacity of the pure loss channel satisfy:
    \bb
        K^{(\varepsilon,n)}(\mathcal{E}_\lambda,N_s)&\ge Q_2^{(\varepsilon,n)}(\mathcal{E}_\lambda,N_s)  \\
        &\ge n\left( h(N_s)-h\!\left((1-\lambda)N_s\right) 
 \right)-\sqrt{n}4\log_2\left(  \sqrt{2^{H_{1/2}(B|A)_{\Psi_{ABE}^{(\lambda,N_s)}}}}+\sqrt{2^{H_{1/2}(B|E)_{\Psi_{ABE}^{(\lambda,N_s)}}}}+1  \right)\sqrt{\log_2\!\left(\frac{8}{\varepsilon}\right)}-\log_2\!\left(\frac{16}{\varepsilon^2}\right)\,,
    \ee
    where $H_{1/2}(B|A)_{\Psi_{ABE}^{(\lambda,N_s)}}$ and $H_{1/2}(B|E)_{\Psi_{ABE}^{(\lambda,N_s)}}$ are reported in \eqref{ingredient1_ec}.
\end{thm}

\subsection{Entropy variance approach}\label{section_entropy_cariance}
In the previous section we found a lower bound on the $n$-shot capacities by bounding the smooth max relative entropy (see Definition~\ref{def_max_smooth}) via the infinite-dimensional asymptotic equipartition property (see Lemma~\ref{lemma_fawzi}). An alternative approach to lower bound the $n$-shot capacities can be obtained from a different bound on the smooth max relative entropy in terms of the so-called \emph{relative entropy variance}.
\begin{Def}[(Relative entropy variance and Conditional entropy variance)]
For finite-dimensional Hilbert spaces, the {relative entropy variance} between a state $\rho$ and a positive semi-definite operator $\sigma$ (in finite dimension) is defined by
\bb
V(\rho\|\sigma)\coloneqq \Tr[\rho(\log_2 \rho-\log_2\sigma)^2]-D(\rho\|\sigma)^2\,.
\ee
For separable Hilbert spaces, this expression has to be meant as a converging series. The {conditional entropy variance} of a bipartite state $\rho$ on the joint system $AB$ is defined as
\bb
V(A|B)_{\rho}\coloneqq V(\rho_{AB}\|\mathbb{1}_{A}\otimes \rho_B)\,.
\ee
\end{Def}
The forthcoming Lemma~\ref{lemma_maxhyp}, presented in~\cite{khatri2021secondorder} (previously stated in the finite-dimensional setting in~\cite{Anshu2019} and later generalised to von Neumann algebras in~\cite{fawzi2023asymptotic}), establishes a connection between the smooth max-relative entropy $D_{\max}^{\varepsilon}(\rho\|\sigma)$ and the \emph{$\varepsilon$-hypothesis testing relative entropy} $D^{\varepsilon}_{h}(\rho\|\sigma)$, defined by
\bb
    D^{\varepsilon}_{h}(\rho\|\sigma)\coloneqq-\log_2\inf_{E}\left\{ \Tr[E\sigma]\,:\, 0\le E\le \mathbb{1},\,\Tr[E\rho]\ge 1-\varepsilon \right\}\,.
\ee 
\begin{lemma}\label{lemma_maxhyp}
Let $\rho$ be a state and $\sigma$ a positive semi-definite operator on a separable Hilbert space. Let $\varepsilon,\delta>0$ such that $\varepsilon+\delta<1$ and $D_{\max}^{\sqrt{\varepsilon}}(\rho\|\sigma)<\infty$. Then, it holds that
\bb
D_{\max}^{\varepsilon}(\rho\|\sigma)&\leq D^{1-\varepsilon}_{h}(\rho\|\sigma)+\log_2\left(\frac{1}{1-\varepsilon}\right)\,,\\
D^{1-\delta-\varepsilon}_{h}(\rho\|\sigma)&\leq D^{\sqrt{\varepsilon}}_{\max}(\rho\|\sigma)+\log_2\left(\frac{4(1-\delta)}{1-\varepsilon}\right)\,.
\ee
\end{lemma}
The following is an immediate consequence of~\cite[Proposition 31]{MMMM}.
\begin{lemma}[(Upper bound for the hypothesis testing relative entropy)]\label{lemma_ChebyDh}
Let $\rho$ be a positive definite state and $\sigma$ be a positive definite operator with bounded trace on a separable Hilbert space. For any $\varepsilon\in(0,1)$ and $n\in\mathbb{N}$ we have that
\begin{align}
D^{\varepsilon}_{h}(\rho^{\otimes n}\|\sigma^{\otimes n})&\leq n D(\rho\|\sigma)+\sqrt{\frac{{nV(\rho\|\sigma)}}{1-\varepsilon}}+\log_2 6+ 2\log_2\frac{1+\varepsilon}{1-\varepsilon}. 
\end{align}
\end{lemma} 

Note that \cite[Proposition 31]{MMMM} is originally stated in the case in which $\sigma$ is a state. However, it can be readily extended to positive semi-definite operators by considering the definitions of $D^{\varepsilon}_{h}(\rho\|\sigma)$, $D(\rho\|\sigma)$, and $V(\rho\|\sigma)$. Additionally, note that the original statement employs the definitions of these quantities in terms of natural logarithms, but we have adjusted it to align with our convention of base-two logarithms.

As a consequence of Lemma~\eqref{lemma_maxhyp} and Lemma~\eqref{lemma_ChebyDh}, we have the following.
\begin{lemma} \label{lemma_low_variance_hmin}
Let $\rho$ be a positive definite state and $\sigma$ be a positive definite operator with bounded trace on a separable Hilbert space. For any $\varepsilon\in(0,1)$ and $n\in\mathbb{N}$ we have that
\begin{align}\label{lowerbound_cond_var}
D_{\max}^{\varepsilon}(\rho^{\otimes n}\|\sigma^{\otimes n})\leq n D(\rho\|\sigma)+\sqrt{\frac{{nV(\rho\|\sigma)}}{\varepsilon}}+\log_2\!\left(\frac{6(2-\varepsilon)^2}{\varepsilon^2(1-\varepsilon)}\right)\,.
\end{align}
\end{lemma}
Let us state the following lower bound on the $n$-shot quantum capacity of the pure loss channel.
\begin{lemma}\label{lem_variance-bound-k-Q}
    Let $\lambda\in(0,1)$ and $N_s>0$. For all $\varepsilon\in(0,1)$ and $n\in\mathbb{N}$, the $n$-shot quantum capacity of $\mathcal{E}_\lambda$ can be lower bounded in terms of the quantum relative entropy and the conditional entropy variance as follows:
    \bb\label{eq_statement_lem_variance}
        Q^{(\varepsilon,n)}(\mathcal{E}_{\lambda})\ge- n D(\Psi^{(k)}_{AE}\|\mathbb{1}_A \otimes \Psi^{(k)}_{E})-4\sqrt{\frac{{nV(A|E)_{\Psi^{(k)}}}}{\varepsilon}}-\log_2\!\left(\frac{2^{23}(32-\varepsilon)^2}{(16-\varepsilon)\varepsilon^6 }\right)\,,
    \ee 
    where $\Psi^{(k)}_{AE}$ is the reduced state on systems $A,E$ of the tri-partite state $\Psi_{ABE}^{(k)}\coloneqq \ketbra{\Psi^{(k)}}_{ABE}$ on systems $A,B,E$ defined as 
    \bb\label{def_psi_k_abe}
        \ket{\Psi^{(k)}}_{ABE}\coloneqq \frac{P_k^A \ket{\Psi^{(\lambda,N_s)}}_{ABE}}{\sqrt{\Tr[P_k \tau_{N_s}]}}\,.
    \ee
    Here, $\Psi^{(\lambda,N_s)}_{ABE}\coloneqq \ketbra{\Psi^{(\lambda,N_s)}}_{ABE}$ is a tri-partite state on systems $A,B,E$ defined as 
    \bb
        \ket{\Psi^{(\lambda,N_s)}}_{ABE}\coloneqq \mathbb{1}_A\otimes U_\lambda^{BE}\, \ket{\psi_{N_s}}_{AB}\otimes\ket{0}_{E}\,;
    \ee
    $U_\lambda^{BE}$ is the beam splitter unitary defined in Definition~\ref{def_beam_splitter}; $\ket{\psi_{N_s}}_{AB}$ is the two-mode squeezed vacuum state defined in Eq.~\eqref{eq:2mode_squeezed}; $\tau^A_{N_s}\coloneqq \Tr_{B}\ketbra{\psi_{N_s}}_{AB}$ is the thermal state with mean photon number $N_s$; $\ket{0}_E$ denoted the vacuum state on $E$; and $P_k^A\coloneqq \sum_{n=0}^{k}\ketbra{n}_A$ is the projector on the first $k+1$ Fock states of $A$.

    In addition, the right-hand side of \eqref{eq_statement_lem_variance} constitutes also a lower bound on the energy-constrained $n$-shot quantum capacity $Q^{(\varepsilon,n)}(\mathcal{E}_{\lambda},N_s)$.
\end{lemma}
\begin{proof}
Let us employ same notation used in the proof of Theorem~\ref{thm_lower_bound}, where  now the input state is chosen to be the two-mode squeezed vacuum state $\ket{\psi_{N_s}}_{AB}$. By exploiting \eqref{ineq_proof_thm_lower_bound}, we have that
\bb\label{eq_00001}
            Q^{(\varepsilon,n)}(\mathcal{E}_\lambda)&\ge -D_{\max}^{\frac{\varepsilon\sqrt{\delta}}{4}-\eta}\left((\Psi^{(k)}_{A_kEE'})^{\otimes n}\|\mathbb{1}_{A_k^n}\otimes (\Psi^{(k)}_{EE'})^{\otimes n}\right)+\log_2\!\left(\eta^4(1-\delta)\right)\,,
\ee
where, thanks to \eqref{eq_def_psik_proof}, the state $\Psi^{(k)}_{A_kB_kEE'}$ can be written as
\bb\label{eq_proof_001}
        \ket{\Psi_{A_kB_kEE'}^{(k)}}= V_{A\to A_k}\otimes V_{B\to B_k}\otimes \mathbb{1}_{EE'}\frac{P_k^A\otimes W_{B\to BE'}^{(k)}\ket{\Psi^{(\lambda,N_s)}_{ABE}}}{\sqrt{\Tr[P_k \tau_{N_s}]}}\,.
    \ee
We have that
\bb
    P_k^A\otimes \mathbb{1}\ket{\Psi^{(\lambda,N_s)}_{ABE}}=\mathbb{1}_A\otimes U_\lambda^{BE}\, (P_k^A\otimes \mathbb{1}_B\ket{\psi_{N_s}}_{AB})\otimes\ket{0}_{E}\,.
\ee
and that
\bb
    P_k^A\otimes \mathbb{1}_B\ket{\psi_{N_s}}_{AB}=P_k^A\otimes P_k^B\ket{\psi_{N_s}}_{AB}\,,
\ee
where the last equality follows from the definition of two-mode squeezed vacuum state in Eq.~\eqref{eq:2mode_squeezed}.
Now, we apply the fact that the beam splitter unitary can not increase the total number of photons. Specifically, from Eq.~\eqref{eq_beam_n0} it easily follows that 
\bb
    \mathbb{1}_A\otimes U_\lambda^{BE}\, (P_k^A\otimes P_k^B\otimes \mathbb{1}_E)\ket{\psi_{N_s}}_{AB}\otimes\ket{0}_{E}&=(P_k^A\otimes P_k^B\otimes \mathbb{1}_E)\mathbb{1}_A\otimes U_\lambda^{BE}\ket{\psi_{N_s}}_{AB}\otimes\ket{0}_{E}  \,,
\ee
i.e.
\bb
    P_k^A\otimes \mathbb{1}\ket{\Psi^{(\lambda,N_s)}_{ABE}}= (P_k^A\otimes P_k^B\otimes \mathbb{1}_E)\ket{\psi_{N_s}^{(\lambda,N_s)}}_{ABE}\,.
\ee
Consequently, from the definition of $W_{B\to BE'}^{(k)}$ in \eqref{def_w_bbe} it follows that
\bb
    P_k^A\otimes W_{B\to BE'}^{(k)}\ket{\Psi^{(\lambda,N_s)}_{ABE}}= P_k^A\otimes \mathbb{1}_{BEE'} \ket{\Psi^{(\lambda,N_s)}_{ABE}}\otimes\ket{0}_{E'}\,.
\ee
and hence \eqref{eq_proof_001} implies that
\bb
     \ket{\Psi_{A_kB_kEE'}^{(k)}}=V_{A\to A_k}\otimes V_{B\to B_k}\otimes \mathbb{1}_{EE'}\left(\frac{P_k^A\otimes \mathbb{1}_{BE} \ket{\Psi^{(\lambda,N_s)}_{ABE}}}{\sqrt{\Tr[P_k \tau_{N_s}]}}\right) \otimes\ket{0}_{E'}\,.
\ee
We are now ready to apply Lemma~\ref{lemma_low_variance_hmin}. Note that one of the hypotheses of Lemma~\ref{lemma_low_variance_hmin} is satisfied because the operator $\sigma\coloneqq\mathbb{1}_{A_k}\otimes \Psi^{(k)}_{EE'} $ has bounded trace (its trace equals $k+1$). The hypothesis that $\sigma$ is positive definite is not satisfied because here $\sigma$ is just positive semi-definite (and not positive definite). However, one can employ a completely analogous reasoning as in \cite[Theorem~24]{MMMM} to make the the statement of Lemma~\ref{lemma_low_variance_hmin} valid also in this case. Indeed, by following the reasoning of \cite[Theorem~24]{MMMM}, we can use that  $D^{\varepsilon}_{\max}(\rho\|\sigma)\leq D^{\varepsilon+\frac12\|\rho-\rho'\|_1}_{\max}(\rho'\|\sigma)$, which is valid for any state $\rho'$, and choose 
\bb
    \rho'=\rho_{k,N_{\mathrm{add}}}\coloneqq\frac{P^E_kP^{E'}_k\mathcal{E}_{N_{\mathrm{add}},EE'}((\Psi^{(k)}_{A_kEE'})^{\otimes n})P^{E}_k P^{E'}_k}{\Tr[P_k^E P^{E'}_k\mathcal{E}_{N_{\mathrm{add},EE'}}((\Psi^{(k)}_{A_kEE'})^{\otimes n})]}\,,
\ee
where $\mathcal{E}_{N_{\mathrm{add}},EE'}$ is an additive noise channel with parameter $N_{\mathrm{add}}$ acting on system $EE'$ (such that $N_{\mathrm{add}}=0$ corresponds to the identity channel), so that $\lim_{N_{\mathrm{add}}\rightarrow 0}\|\rho-\rho'\|_1=0$. 
In the achievability bounds we are allowed to choose any admissible $\sigma$, see Definition~\ref{def_smooth_min_ent} and how this is applied in~\eqref{ineq_proof_thm_lower_bound}; thus we can take it to be $\sigma_{k,N_{\mathrm{add}}}\coloneqq\mathbb{1}_{A_k}\otimes \Tr_{A_k}[\rho_{k,N_{\mathrm{add}}}]$. 
In this way $\sigma$ is full rank and Lemma~\ref{lemma_low_variance_hmin} can be applied. In the end, we can send $N_{\mathrm{add}}$ to zero, and we are left to check that the limits of the conditional entropy and of the variance of the conditional entropy converge to their value when $N_{\mathrm{add}}=0$. One can be convinced of this fact by noticing that to evaluate these expression on the considered class of states one only needs the diagonal density matrix elements in the Fock basis (since $\sigma_{k,N_{\mathrm{add}}}$ is diagonal in this basis), which are finitely many due to the projections $P_k^E$ and $P_k^{E'}$, and converge to non-zero numbers corresponding to the case $N_{\mathrm{add}}=0$. Hence, we conclude that the statement of Lemma~\ref{lemma_low_variance_hmin} is valid also when $\sigma$ is just positive semi-definite without being positive definite.

Consequently, by exploiting Lemma~\ref{lemma_low_variance_hmin} and Eq.~\eqref{eq_00001} with $\delta=\frac{1}{4}$ and $\eta=\frac{\varepsilon}{16}$, we obtain that
\bb\label{eq012345}
        Q^{(\varepsilon,n)}(\mathcal{E}_{\lambda}) &\geq  - n D(\Psi^{(k)}_{AE}\|\mathbb{1}_A \otimes \Psi^{(k)}_{E})-4\sqrt{\frac{{nV(A|E)_{\Psi^{(k)}}}}{\varepsilon}}-\log_2\!\left(\frac{2^{23}(32-\varepsilon)^2}{(16-\varepsilon)\varepsilon^6 }\right)\,.
    \ee  
Since the mean photon number on $B$ of the two-mode squeezed vacuum state $\ket{\psi_{N_s}}_{AB}$ is equal to $N_s$, it follows that the right-hand side of Eq.~\eqref{eq012345} constitutes also a lower bound on the energy-constrained $n$-shot quantum capacity $Q^{(\varepsilon,n)}(\mathcal{E}_{\lambda},N_s)$.
\end{proof}
Similarly, we obtain a lower bound on the $n$-shot two-way quantum capacity and secret-key capacity of the pure loss channel.
\begin{lemma}\label{lem_variance-bound-k-L}
    Let $\varepsilon\in(0,1)$ and $n\in\mathbb{N}$. By employing the same notation as in Lemma~\ref{lem_variance-bound-k-Q}, the $n$-shot two-way quantum capacity and the $n$-shot secret-key capacity of the pure loss channel $\mathcal{E}_\lambda$ can be lower bounded in terms of the quantum relative entropy and the conditional entropy variance as follows:
    \bb\label{b_lower_bound_CV_Q2_11}
        K^{(\varepsilon,n)}(\mathcal{E}_{\lambda})\ge Q_2^{(\varepsilon,n)}(\mathcal{E}_{\lambda})&\ge - n D(\Psi^{(k)}_{AE}\|\mathbb{1}_A \otimes \Psi^{(k)}_{E})-\sqrt{\frac{{2nV(A|E)_{\Psi^{(k)}}}}{\sqrt{\varepsilon}}}-\log_2\!\left(\frac{2^{6}\,3\,(4-\sqrt{\varepsilon})^2}{(2-\sqrt{\varepsilon})\varepsilon^3 }\right)\,,\\
        K^{(\varepsilon,n)}(\mathcal{E}_{\lambda})\ge Q_2^{(\varepsilon,n)}(\mathcal{E}_{\lambda}) &\ge  - n D(\Psi^{(k)}_{BE}\|\mathbb{1}_A \otimes \Psi^{(k)}_{E})-\sqrt{\frac{{2nV(B|E)_{\Psi^{(k)}}}}{\sqrt{\varepsilon}}}-\log_2\!\left(\frac{2^{6}\,3\,(4-\sqrt{\varepsilon})^2}{(2-\sqrt{\varepsilon})\varepsilon^3 }\right)\,.
    \ee
    In addition, the right-hand sides of \eqref{b_lower_bound_CV_Q2_11} constitute  also a lower bound on the energy-constrained $n$-shot two-way quantum capacity $Q_2^{(\varepsilon,n)}(\mathcal{E}_{\lambda},N_s)$ and energy-constrained $n$-shot two-way secret-key capacity $K^{(\varepsilon,n)}(\mathcal{E}_{\lambda},N_s)$.
\end{lemma}
\begin{proof}
    The proof is entirely analogous to the proof of Lemma~\ref{lem_variance-bound-k-Q}. In particular, it exploits Lemma~\ref{lemma_lower_2_way} with $\eta=\frac{\sqrt{\varepsilon}}{2}$ and Lemma~\ref{lemma_low_variance_hmin}.
\end{proof}
In the following, we will compute $V(B|E)_{\Psi^{(\lambda,N_s)}}$ and $V(A|E)_{\Psi^{(\lambda,N_s)}}$ and show that $V(B|E)_{\Psi^{(k)}}$ and $V(A|E)_{\Psi^{(k)}}$ converge to them. We have an explicit proof for the states at hand; it is possible that a more general statement can be made for states which have bounded second moments, but we leave it for future work.

Let us start by stating the following lemma.
\begin{lemma}\label{expre_out_stine_pure}
Let $\lambda\in(0,1),N_s>0$, and
    \bb\label{def_Psi_abe}
        \ket{\Psi}_{ABE}\coloneqq \mathbb{1}_A\otimes U_\lambda^{BE}\, \ket{\psi_{N_s}}_{AB}\otimes\ket{0}_{E}\,,
    \ee
where $U_\lambda^{BE}$ is the beam splitter unitary defined in Definition~\ref{def_beam_splitter}; $\ket{\psi_{N_s}}_{AB}$ is the two-mode squeezed vacuum state defined in Eq.~\eqref{eq:2mode_squeezed}; and $\ket{0}_E$ is the vacuum state on $E$. Then, it holds that
\bb\label{eq_main_0}
    \ket{\Psi}_{ABE}=\sum_{n=0}^\infty\sum_{l=0}^n\frac{1}{\sqrt{N_s+1}}\left(\frac{N_s}{N_s+1}\right)^{n/2}(-1)^l\sqrt{\binom{n}{l} \lambda^{n-l}(1-\lambda)^l}\ket{n}_A\otimes\ket{n-l}_B\otimes\ket{l}_E\,.
\ee
In particular, it holds that
\bb\label{eq_main_1}
    \Psi_{A}&=\tau_{N_s}\,,\\
    \Psi_{B}&=\tau_{\lambda N_s}\,,\\
    \Psi_{E}&=\tau_{(1-\lambda) N_s}\,,\\
\ee
\end{lemma}
\begin{proof}
    \eqref{eq_main_0} easily follows by exploiting the definition of two-mode squeezed vacuum state and by employing \eqref{eq_beam_n0}. Moreover, \eqref{eq_main_1} follows by taking the partial traces of \eqref{eq_main_0} and by exploiting the fact $\sum_{n=0}^\infty N\!B(n;r,p)=1$, where 
    \bb
        N\!B(n;r,p)\coloneqq\binom{n+r-1}{n}(1-p)^n p^r\quad\forall\, p\in[0,1], n\in\mathbb{N}, r\in\mathbb{N}\,,
    \ee
    is the \emph{negative binomial distribution}.
\end{proof}
The following lemma allows us to compute $V(B|E)_{\Psi^{(\lambda,N_s)}}$ and $V(A|E)_{\Psi^{(\lambda,N_s)}}$.
\begin{lemma}
Let $\lambda\in(0,1),N_s>0$, and
    \bb
        \ket{\Psi}_{ABE}\coloneqq \mathbb{1}_A\otimes U_\lambda^{BE}\, \ket{\psi_{N_s}}_{AB}\otimes\ket{0}_{E}\,,
    \ee
be the state defined in \eqref{def_Psi_abe}. Then, by denoting $\Psi_{ABE}\coloneqq \ketbra{\Psi}_{ABE}$, it holds that
\bb\label{eq_prove_entropy_variance}
V(A|E)_{\Psi}=V(A|B)_{\Psi}&=\lambda N_s(1+\lambda N_s)\left(\log_2\frac{\lambda N_s}{1+\lambda N_s}\right)^2\\
&\quad+(1-\lambda)N_s(1+(1-\lambda)N_s)\left(\log_2\frac{(1-\lambda)N_s}{1+(1-\lambda)N_s}\right)^2\\
&\quad-2(1-\lambda)\lambda N_s^2\log_2\frac{(1-\lambda)N_s}{1+(1-\lambda)N_s}\log_2\frac{\lambda N_s}{1+\lambda N_s}\,,\\
V(B|E)_{\Psi }=V(B|A)_{\Psi}&=N_s(1+N_s)\left(\log_2\frac{N_s}{1+N_s}\right)^2\\
&\quad+(1-\lambda)N_s(1+(1-\lambda)N_s)\left(\log_2\frac{(1-\lambda)N_s}{1+(1-\lambda)N_s}\right)^2\\
&\quad-2(1-\lambda)N_s(1+N_s)\left(\log_2\frac{(1-\lambda)N_s}{1+(1-\lambda)N_s}\right)\left(\log_2\frac{N_s}{1+N_s}\right),
\ee
Consequently, the following limits hold:
\bb\label{limits_var}
    \lim_{N_s\rightarrow \infty }V(A|E)_{\Psi }&=\lim_{N_s\rightarrow \infty }V(A|B)_{\Psi }=0\,,\\
    \lim_{N_s\rightarrow \infty }V(B|E)_{\Psi }&=\lim_{N_s\rightarrow \infty }V(B|A)_{\Psi }=0\,.
\ee
\end{lemma}

\begin{proof}
The expression of $V(B|E)_{\Psi_{AB}}$ in \eqref{eq_prove_entropy_variance} has already been proved in \cite{Kaur_2017}. Let us sketch the proof presented in \cite{Kaur_2017}, as it will be useful for the rest of our proof. As proved in \cite{Kaur_2017}, from the definition of conditional entropy variance one can easily see that
\bb\label{expansionV}
V(B|E)_{\Psi_{ABE}}&=\Tr[\Psi_{A}(\log_2 \Psi_{A})^2]-\Tr[\Psi_{A}\log_2 \Psi_{A}]^2\\
&\quad+\Tr[\Psi_{E}(\log_2 \Psi_{E})^2]-\Tr[\Psi_{E}\log_2 \Psi_{E}]^2\\
&\quad-2\Tr[\Psi_{AE}(\log_2 \Psi_{A})(\log_2 \Psi_{E})]+2\Tr[\Psi_{E}\log_2 \Psi_{E}]\Tr[\Psi_{A}\log_2 \Psi_{A}]\,.
\ee
Since this expression is symmetric under exchanging $E$ and $A$, it follows that $V(B|E)_{\Psi_{ABE}}=V(B|A)_{\Psi_{ABE}}$. Additionally, Lemma~\ref{expre_out_stine_pure} establishes that
\bb
    \Psi_{A}&=\tau_{N_s}\,,\\
    \Psi_{E}&=\tau_{(1-\lambda)N_s}\,.
\ee
Hence, thanks to \eqref{entropy_thermal}, we have that
\bb
    -\Tr[\Psi_{A}\log_2 \Psi_{A}]&= h(N_s)\,,\\
    -\Tr[\Psi_{E}\log_2 \Psi_{E}]&= h\!\left((1-\lambda)N_s\right)\,,
\ee
where $h(x)\coloneqq (x+1)\log_2(x+1) - x\log_2 x$. Furthermore, the first two terms in \ref{expansionV} are the entropy variances of $\Psi_{A}$ and $\Psi_{E}$, which are both thermal states. Using~\eqref{ent_var_th_state}, these terms can be expressed as
\bb\label{relvariances0}
\Tr[\Psi_{A}(\log_2 \Psi_{A})^2]-\Tr[\Psi_{A}\log_2 \Psi_{A}]^2&=N_s(N_s+1)\log_2\left(1+\frac{1}{N_s}\right)^2\,,\\
\Tr[\Psi_{E}(\log_2 \Psi_{E})^2]-\Tr[\Psi_{E}\log_2 \Psi_{E}]^2&=(1-\lambda)N_s\left[(1-\lambda)N_s+1\right]\log_2\left(1+\frac{1}{(1-\lambda)N_s}\right)^2\,.
\ee
Hence, we have that
\bb\label{expansionV2}
    V(B|E)_{\Psi_{ABE}}&=N_s(N_s+1)\log_2\left(1+\frac{1}{N_s}\right)^2\\
    &\quad+(1-\lambda)N_s\left[(1-\lambda)N_s+1\right]\log_2\left(1+\frac{1}{(1-\lambda)N_s}\right)^2\\
    &\quad -2\left[\Tr[\Psi_{AE}(\log_2 \Psi_{A})(\log_2 \Psi_{E})]-h(N_s)h((1-\lambda)N_s)\right]
\ee
Hence, to complete the calculation of $V(B|E)_{\Psi_{ABE}}$, we only need to calculate the term $\Tr[\Psi_{AE}(\log_2 \Psi_{A})(\log_2 \Psi_{E})]$. To this end, note that the operator $\log_2\tau_{N_s}$ can be written in terms of the photon number operator $\hat{N}_n$ as
\bb
    \log_2\tau_{N_s}=\log_2\left(\frac{N_s}{N_s+1}\right)\hat{N} -\log_2(1+N_s)\mathbb{1}\,.
\ee
Consequently, by denoting as $\hat{N}_A$ and $\hat{N}_E$ the photon number operators on $A$ and $E$, we have that
\bb
    \log_2\Psi_A&=\log_2\left(\frac{N_s}{N_s+1}\right)\hat{N}_A -\log_2(1+N_s)\mathbb{1}_A\,,\\
    \log_2\Psi_E&=\log_2\left(\frac{(1-\lambda)N_s}{(1-\lambda)N_s+1}\right)\hat{N}_E -\log_2(1+(1-\lambda)N_s)\mathbb{1}_E\,,    
\ee 
and thus
\bb\label{eq_prove_ent_varia_pur}
    \Tr[\Psi_{AE}(\log_2 \Psi_{A})(\log_2 \Psi_{E})]&= \log_2\left(\frac{N_s}{N_s+1}\right) \log_2\left(\frac{(1-\lambda)N_s}{(1-\lambda)N_s+1}\right)\Tr[\Psi_{AE}\hat{N}_A\otimes\hat{N}_E]\\
    &\quad-\log_2\left(\frac{N_s}{N_s+1}\right)\log_2\left(1+(1-\lambda)N_s\right)\Tr[\Psi_{A}\hat{N}_A]\\
    &\quad -\log_2\left(\frac{(1-\lambda)N_s}{(1-\lambda)N_s+1}\right)\log_2\left(1+N_s\right)\Tr[\Psi_{E}\hat{N}_E]\\
    &\quad+\log_2\left(1+N_s\right)\log_2\left(1+(1-\lambda)N_s\right)\\
    &= \log_2\left(\frac{N_s}{N_s+1}\right) \log_2\left(\frac{(1-\lambda)N_s}{(1-\lambda)N_s+1}\right)\Tr[\Psi_{AE}\hat{N}_E\otimes\hat{N}_A]\\
    &\quad-\log_2\left(\frac{N_s}{N_s+1}\right)\log_2\left(1+(1-\lambda)N_s\right)N_s\\
    &\quad -\log_2\left(\frac{(1-\lambda)N_s}{(1-\lambda)N_s+1}\right)\log_2\left(1+N_s\right)(1-\lambda)N_s\\
    &\quad+\log_2\left(1+N_s\right)\log_2\left(1+(1-\lambda)N_s\right)\,.
\ee
Hence, we just need to evaluate the term $\Tr[\Psi_{AE}\hat{N}_A\otimes\hat{N}_E]$. We have that
\bb
    \Tr[\Psi_{AE}\hat{N}_A\otimes\hat{N}_E]&=\sum_{n=0}^\infty \sum_{l=0}^\infty nl \Tr[\Psi_{AE}\ketbra{n}_A\otimes \ketbra{l}_E]\\
    &\eqt{(i)}\frac{1}{N_s+1}\sum_{n=0}^\infty \sum_{l=0}^n nl \left(\frac{N_s}{N_s+1}\right)^n\binom{n}{l} \lambda^{n-l}(1-\lambda)^l\\
    &\eqt{(ii)}\frac{1-\lambda}{N_s+1}\sum_{n=0}^\infty n^2\left(\frac{N_s}{N_s+1}\right)^n\\
    &\eqt{(iii)}(1-\lambda)N_s(2N_s+1)\,,
\ee
where in (i) we employed Lemma~\ref{expre_out_stine_pure}, in (ii) we used that the known expression for the mean value of the binomial distribution, and in (iii)
we exploited that 
\bb
    \sum_{n=0}^\infty n^2\left(\frac{N_s}{N_s+1}\right)^n=N_s(N_s+1)(2N_s+1)\,.
\ee
Hence, by substituting in \eqref{eq_prove_ent_varia_pur}, one can easily show that
\bb
    \Tr[\Psi_{AE}(\log_2 \Psi_{A})(\log_2 \Psi_{E})]-h(N_s)h((1-\lambda)N_s)= (1-\lambda)N_s(1+N_s)\log_2\left(\frac{(1-\lambda)N_s}{1+(1-\lambda)N_s}\right)\log_2\left(\frac{N_s}{1+N_s}\right)\,.
\ee
Hence, by employing \eqref{expansionV2}, one can prove the expression of $V(B|E)_{\Psi_{ABE}}$ in \eqref{eq_prove_entropy_variance}.

Now, let us prove the expression of $V(A|E)_{\Psi_{ABE}}$ in \eqref{eq_prove_entropy_variance}. As seen above in \eqref{expansionV}, it holds that $V(A|E)_{\Psi_{ABE}}=V(A|B)_{\Psi_{ABE}}$ and that
\bb\label{expansionV3}
    V(A|E)_{\Psi_{ABE}}&=\Tr[\Psi_{B}(\log_2 \Psi_{B})^2]-\Tr[\Psi_{B}\log_2 \Psi_{B}]^2\\
    &\quad+\Tr[\Psi_{E}(\log_2 \Psi_{E})^2]-\Tr[\Psi_{E}\log_2 \Psi_{E}]^2\\
    &\quad-2\Tr[\Psi_{BE}(\log_2 \Psi_{B})(\log_2 \Psi_{E})]+2\Tr[\Psi_{E}\log_2 \Psi_{E}]\Tr[\Psi_{B}\log_2 \Psi_{B}]\,.
\ee
Since
\bb
    \Psi_{B}&=\tau_{\lambda N_s}\,,\\
    \Psi_{E}&=\tau_{(1-\lambda)N_s}\,,
\ee
as established by Lemma~\ref{expre_out_stine_pure}, we can \eqref{entropy_thermal} to obtain that
\bb
    -\Tr[\Psi_{B}\log_2 \Psi_{B}]&= h(\lambda N_s)\,,\\
    -\Tr[\Psi_{E}\log_2 \Psi_{E}]&= h\!\left((1-\lambda)N_s\right)\,,
\ee
where $h(x)\coloneqq (x+1)\log_2(x+1) - x\log_2 x$. Furthermore, the first two terms in \ref{expansionV3} are the entropy variances of $\Psi_{B}$ and $\Psi_{E}$, which are both thermal states. Using~\eqref{ent_var_th_state}, these terms can be expressed as
\bb\label{relvariances}
\Tr[\Psi_{B}(\log_2 \Psi_{B})^2]-\Tr[\Psi_{B}\log_2 \Psi_{B}]^2&=\lambda N_s(\lambda N_s+1)\log_2\left(1+\frac{1}{\lambda N_s}\right)^2\,,\\
\Tr[\Psi_{E}(\log_2 \Psi_{E})^2]-\Tr[\Psi_{E}\log_2 \Psi_{E}]^2&=(1-\lambda)N_s\left[(1-\lambda)N_s+1\right]\log_2\left(1+\frac{1}{(1-\lambda)N_s}\right)^2\,.
\ee
Hence, we have that
\bb\label{expansionV4}
    V(A|E)_{\Psi_{ABE}}&=\lambda N_s(\lambda N_s+1)\log_2\left(1+\frac{1}{\lambda N_s}\right)^2\\
    &\quad (1-\lambda)N_s\left[(1-\lambda)N_s+1\right]\log_2\left(1+\frac{1}{(1-\lambda)N_s}\right)^2  \\
    &\quad -2\left[\Tr[\Psi_{BE}(\log_2 \Psi_{B})(\log_2 \Psi_{E})]-h(\lambda N_s)h((1-\lambda)N_s)\right]
\ee
To complete the calculation of all the terms in \eqref{expansionV2}, we only need to calculate the term $\Tr[\Psi_{BE}(\log_2 \Psi_{B})(\log_2 \Psi_{E})]$. As done above, we have that
\bb
    \log_2\Psi_B&=\log_2\left(\frac{\lambda N_s}{\lambda N_s+1}\right)\hat{N}_B -\log_2(1+\lambda N_s)\mathbb{1}_B\,,\\
    \log_2\Psi_E&=\log_2\left(\frac{(1-\lambda) N_s}{(1-\lambda) N_s+1}\right)\hat{N}_E -\log_2(1+(1-\lambda) N_s)\mathbb{1}_E\,,
\ee
and hence
\bb\label{eq_prove_ent_varia_pur2}
    \Tr[\Psi_{BE}(\log_2 \Psi_{B})(\log_2 \Psi_{E})]&= \log_2\left(\frac{\lambda N_s}{\lambda N_s+1}\right) \log_2\left(\frac{(1-\lambda)N_s}{(1-\lambda)N_s+1}\right)\Tr[\Psi_{BE}\hat{N}_B\otimes\hat{N}_E]\\
    &\quad-\log_2\left(\frac{\lambda N_s}{\lambda N_s+1}\right)\log_2\left(1+(1-\lambda)N_s\right)\Tr[\Psi_{B}\hat{N}_B]\\
    &\quad -\log_2\left(\frac{(1-\lambda)N_s}{(1-\lambda)N_s+1}\right)\log_2\left(1+\lambda N_s\right)\Tr[\Psi_{E}\hat{N}_E]\\
    &\quad+\log_2\left(1+\lambda N_s\right)\log_2\left(1+(1-\lambda)N_s\right)\\
&= \log_2\left(\frac{\lambda N_s}{\lambda N_s+1}\right) \log_2\left(\frac{(1-\lambda)N_s}{(1-\lambda)N_s+1}\right)\Tr[\Psi_{BE}\hat{N}_B\otimes\hat{N}_E]\\
    &\quad-\log_2\left(\frac{\lambda N_s}{\lambda N_s+1}\right)\log_2\left(1+(1-\lambda)N_s\right)\lambda N_s\\
    &\quad -\log_2\left(\frac{(1-\lambda)N_s}{(1-\lambda)N_s+1}\right)\log_2\left(1+\lambda N_s\right)(1-\lambda)N_s\\
    &\quad+\log_2\left(1+\lambda N_s\right)\log_2\left(1+(1-\lambda)N_s\right)\,.
\ee
Hence, we just need to evaluate the term $\Tr[\Psi_{BE}\hat{N}_B\otimes\hat{N}_E]$. By employing Lemma~\ref{expre_out_stine_pure}, we have that
\bb\label{eq_enbe}
    \Tr[\Psi_{BE}\hat{N}_B\otimes\hat{N}_E]&= \sum_{l=0}^\infty  l \Tr[\Psi_{AE}\hat{N}_B\otimes \ketbra{l}_E]\\
    &\eqt{(iv)}\frac{1}{N_s+1}\sum_{n=0}^\infty \sum_{l=0}^n (n-l)l \left(\frac{N_s}{N_s+1}\right)^n\binom{n}{l} \lambda^{n-l}(1-\lambda)^l\\
    &\eqt{(v)}\frac{1}{N_s+1}\sum_{n=0}^\infty \left(\frac{N_s}{N_s+1}\right)^n \left[n^2 (1-\lambda)-\left(n\lambda(1-\lambda)+n^2(1-\lambda)^2\right) \right]\\
    &\eqt{(vi)}\frac{\lambda(1-\lambda)}{N_s+1}\sum_{n=0}^\infty \left(\frac{N_s}{N_s+1}\right)^n(n^2-n)\\
    &= 2 \lambda(1-\lambda) N_s^2\,,
\ee
where in (v) and in (vi) we used standard moment evaluations for binomial and geometric distributions. Hence, by substituting in \eqref{eq_prove_ent_varia_pur2}, one can easily show that
\bb
    \Tr[\Psi_{BE}(\log_2 \Psi_{B})(\log_2 \Psi_{E})]-h(\lambda N_s)h((1-\lambda)N_s)=  (1-\lambda)\lambda N_s^2\log_2\!\left(\frac{(1-\lambda)N_s}{1+(1-\lambda)N_s}\right)\log_2\!\left(\frac{\lambda N_s}{1+\lambda N_s}\right)\,.
\ee
Hence, by employing \eqref{expansionV4} one can prove the expression of $V(A|E)_{\Psi_{ABE}}$ in \eqref{eq_prove_entropy_variance}. Both limits in \eqref{limits_var} are then easy to obtain.
\end{proof}
The following lemma establishes that $V(B|E)_{\Psi^{(k)}}$ and $V(A|E)_{\Psi^{(k)}}$ converge to $V(B|E)_{\Psi^{(\lambda,N_s)}}$ and $V(A|E)_{\Psi^{(\lambda,N_s)}}$, respectively.
\begin{lemma}\label{lemma_cont_cond_var}
    Let $\lambda\in(0,1),N_s>0$, and let $\Psi_{ABE}\coloneqq \ketbra{\Psi}_{ABE}$, where
    \bb
        \ket{\Psi}_{ABE}\coloneqq \mathbb{1}_A\otimes U_\lambda^{BE}\, \ket{\psi_{N_s}}_{AB}\otimes\ket{0}_{E}\,,
    \ee
    is the state defined in \eqref{def_Psi_abe}. Moreover, let $\Psi_{ABE}^{(k)}\coloneqq \ketbra{\Psi^{(k)}}_{ABE}$, where
    \bb
        \ket{\Psi^{(k)}}_{ABE}\coloneqq \frac{P_k^A \ket{\Psi^{(\lambda,N_s)}}_{ABE}}{\sqrt{\Tr[P_k \tau_{N_s}]}}\,.
    \ee
    is the state defined in \eqref{def_psi_k_abe}. Then, it holds that
    \bb
        \lim_{k\rightarrow \infty}V(A|E)_{\Psi^{(k)}}&=V(A|E)_{\Psi}\,,\\
        \lim_{k\rightarrow \infty}V(B|E)_{\Psi^{(k)}}&=V(B|E)_{\Psi}\,.
    \ee
\end{lemma}
\begin{proof}
    From the definition of conditional entropy variance, it follows that
    \bb\label{vae_psi}
        V(A|E)_{\Psi}&=\Tr\left[\Psi_{AE}\left(\log_2\Psi_{AE}-\log_2\Psi_{E}\right)^2\right]-[S(A|E)_{\Psi}]^2\\
        &=\Tr[\Psi_{AE}\log_2^2\Psi_{AE}]+\Tr[\Psi_{E}\log_2^2\Psi_{E}]-2\Tr[\Psi_{AE}\log_2\Psi_{AE}\log_2\Psi_{E}]-[S(A|E)_{\Psi}]^2
    \ee
    and that
    \bb\label{vae_psi_k}
        V(A|E)_{\Psi^{(k)}}=\Tr[\Psi^{(k)}_{AE}\log_2^2\Psi^{(k)}_{AE}]+\Tr[\Psi^{(k)}_{E}\log_2^2\Psi^{(k)}_{E}]-2\Tr[\Psi^{(k)}_{AE}\log_2\Psi^{(k)}_{AE}\log_2\Psi^{(k)}_{E}]-[S(A|E)_{\Psi^{(k)}}]^2\,.
    \ee
    Hence, in order to show that $\lim\limits_{k\to\infty}V(A|E)_{\Psi^{(k)}}=V(A|E)_{\Psi}$, it suffices to show that 
    \bb\label{limits_k_to_prove}
        \lim\limits_{k\rightarrow\infty} S(A|E)_{\Psi^{(k)}}&=S(A|E)_{\Psi}\,, \\
        \lim\limits_{k\rightarrow\infty} \Tr[\Psi^{(k)}_{AE}\log_2^2\Psi^{(k)}_{AE}]&=\Tr[\Psi_{AE}\log_2^2\Psi_{AE}]\,,\\
        \lim\limits_{k\rightarrow\infty} \Tr[\Psi^{(k)}_{E}\log_2^2\Psi^{(k)}_{E}]&=\Tr[\Psi_{E}\log_2^2\Psi_{E}]\,,\\    
        \lim\limits_{k\rightarrow\infty} \Tr[\Psi^{(k)}_{AE}\log_2\Psi^{(k)}_{AE}\log_2\Psi^{(k)}_{E}]&=\Tr[\Psi_{AE}\log_2\Psi_{AE}\log_2\Psi_{E}]\,.
    \ee
    The first limit in \eqref{limits_k_to_prove} follows from the same reasoning used to prove \eqref{lim_sab_k_inf}, which leveraged the convergence of $\Psi^{(k)}$ to $\Psi$ in trace norm and the continuity of conditional entropy with respect to the trace norm over the set of states with bounded mean photon number~\cite{tightuniform}.

    Now, let us prove the second limit in \eqref{limits_k_to_prove}, i.e.~$\lim\limits_{k\rightarrow\infty} \Tr[\Psi^{(k)}_{AE}\log_2^2\Psi^{(k)}_{AE}]=\Tr[\Psi_{AE}\log_2^2\Psi_{AE}]$. This is equivalent to show that
    \bb
        \lim\limits_{k\rightarrow\infty} \Tr[\Psi^{(k)}_{B}\log_2^2\Psi^{(k)}_{B}]=\Tr[\Psi_{B}\log_2^2\Psi_{B}]\,.
    \ee
    Thanks to Lemma~\ref{expre_out_stine_pure}, we have that
    \bb
        \Psi_{B}=\tau_{\lambda N_s}=\sum_{l=0}^\infty q_{\infty, l}\ketbra{l}\,,
    \ee
    where we introduced
    \bb
        q_{\infty, l}\coloneqq \frac{1}{\lambda N_s+1}\left(\frac{\lambda N_s}{\lambda N_s+1}\right)^l\,.
    \ee
    Similarly to the proof of Lemma~\ref{expre_out_stine_pure}, we have that
    \bb\label{eq_psi_k_abe_formula}
        \ket{\Psi^{(k)}}_{ABE}=\frac{1}{\sqrt{\Tr[P_k\tau_{N_s}]}}\sum_{n=0}^k\sum_{l=0}^n\frac{1}{\sqrt{N_s+1}}\left(\frac{N_s}{N_s+1}\right)^{n/2}(-1)^l\sqrt{\binom{n}{l} \lambda^{n-l}(1-\lambda)^l}\ket{n}_A\otimes\ket{n-l}_B\otimes\ket{l}_E\,.
    \ee
    Hence, we have that
    \bb
        \Psi^{(k)}_{B}=\frac{1}{p_{k}}\sum_{l=0}^k q_{k, l}\ketbra{l}\,,
    \ee
    where we introduced
    \bb
        q_{k, l}&\coloneqq \sum_{n=l}^k \frac{1}{N_s+1}\left(\frac{N_s}{N_s+1}\right)^n\binom{n}{l}\lambda^l(1-\lambda)^{n-l}\,,\\
        p_{k}&\coloneqq \Tr[p_k \tau_{N_s}]=\frac{1}{N_s+1}\sum_{n=0}^{k}\left(\frac{N_s}{N_s+1}\right)^n=1-\left(\frac{N_s}{N_s+1}\right)^{k+1}\,.
    \ee    
    Note that
    \bb\label{simplification_without_pk}
        \Tr[\Psi^{(k)}_{B}\log_2^2\Psi^{(k)}_{B}]&=\sum_{l=0}^\infty \frac{q_{k,l}}{p_k}\left(\log_2\frac{q_{k,l}}{p_k}\right)^2\\
        &=\frac{1}{{p_{k}}}\sum_{l=0}^{k}q_{k,l}\left(\log_2^2 q_{k,l}-2\log_2 p_{k}\log_2 q_{k,l}+\log_2^2p_{k}\right)\\
        &=\frac{1}{{p_{k}}}\sum_{l=0}^{k}q_{k,l}\log_2^2 q_{k,l}-2\log_2 p_{k}\sum_{l=0}^{k}\frac{q_{k,l}}{p_{k}}\log_2 \frac{q_{k,l}}{p_{k}}- \log_2^2p_{k}\,,\\
        &=\frac{1}{{p_{k}}}\sum_{l=0}^{k}q_{k,l}\log_2^2 q_{k,l}+2(\log_2 p_{k} )S(\Psi^{(k)}_B)- \log_2^2p_{k}\,.
    \ee
    By exploiting the facts that (i) $\lim\limits_{k\rightarrow\infty}p_k=1$, (ii) the state $\Psi^{(k)}_B$ converges to $\Psi_{B}$ in trace norm (as proved in the proof of Theorem~\ref{thm_lower_bound}), which has finite entropy (specifically, $S( \Psi_{B})=S(\tau_{\lambda N_s})=h(\lambda N_s)$), and (iii) the von Neumann entropy is continuous with respect the trace norm over the set of states having bounded mean photon number~\cite{tightuniform}, it follows that the term $2(\log_2 p_{k} )S(\Psi^{(k)}_B)- \log_2^2p_{k}$ converges to zero as $k\rightarrow\infty$. Hence, to complete the proof of $\lim\limits_{k\rightarrow\infty} \Tr[\Psi^{(k)}_{B}\log_2^2\Psi^{(k)}_{B}]=\Tr[\Psi_{B}\log_2^2\Psi_{B}]$, we just need to prove that
    \bb
        \lim\limits_{k\rightarrow\infty}\sum_{l=0}^{k}q_{k,l}\log_2^2 q_{k,l}=\sum_{l=0}^{\infty}q_{\infty,l}\log_2^2 q_{\infty,l}\,.
    \ee
    Let us fix $l$ and $k$ such that $l\le k$. Note that 
    \bb\label{eq244}
        \frac{q_{k,l}}{q_{\infty,l}}&=\frac{\lambda N_s+1}{ N_s+1}\left(\frac{\lambda N_s}{\lambda N_s+1}\right)^{-l}\sum_{n=l}^k\binom{n}{l}\lambda^l(1-\lambda)^{n-l}\left(\frac{N_s}{N_s+1}\right)^n\\
        &=\frac{\lambda N_s+1}{ N_s+1}\left(\frac{\lambda N_s}{\lambda N_s+1}\right)^{-l}\sum_{n=0}^{k-l}\binom{n+l}{l}\lambda^l(1-\lambda)^{n}\left(\frac{N_s}{N_s+1}\right)^{n+l}\\
        &= \sum_{n=0}^{k-l} \binom{n+l}{l}\left( \frac{(1-\lambda)N_s}{N_s+1}  \right)^{n}\left(\frac{ \lambda N_s+1}{N_s+1}\right)^{l+1}\\
        &= \sum_{n=0}^{k-l} N\!B\left( n; l+1, \frac{ \lambda N_s+1}{N_s+1} \right)\\
        &= 1- \sum_{n=k-l+1}^{\infty} N\!B\left( n; l+1, \frac{ \lambda N_s+1}{N_s+1} \right)\,,
    \ee
    where we introduced the \emph{negative binomial distribution} defined as 
    \bb
        N\!B(n;r,p)\coloneqq\binom{n+r-1}{n}(1-p)^n p^r
    \ee
    for all $p\in[0,1], n\in\mathbb{N}, r\in\mathbb{N}$.
    Its moment generating function is known to be 
    \bb
        \sum_{n=0}^\infty N\!B(n;r,p) e^{tn}   =\left(\frac{p}{1-(1-p)e^t}\right)^{r}\,.
    \ee
    for $t< \ln \frac{1}{1-p}$. Consequently, by applying a Chernoff bound approach, we have that 
    \bb
        \sum_{ n=m }^\infty N\!B(n;r,p) \le e^{-t m} \sum_{ n=m }^\infty e^{tn} N\!B(n;r,p)\le e^{-t m} \sum_{ n=0 }^\infty e^{tn} N\!B(n;r,p)=e^{-t m} \left(\frac{p}{1-(1-p)e^t}\right)^{r}\,
    \ee
    for all $t$ such that $e^t<\frac{1}{1-p}$. Let $m\ge1$ and let us choose $t$ such that $e^t=\frac{m}{m+r}\frac{1}{1-p}$. Hence, it holds that
    \bb
        \sum_{ n=m }^\infty N\!B(n;r,p) &\le \frac{p^r (1-p)^m}{(\frac{r}{m+r})^r (1-\frac{r}{m+r})^m  }\\
        &=e^{  r\left(\ln p - \ln(\frac{r}{m+r}) \right) + m\left(\ln (1-p) - \ln(1-\frac{r}{m+r}) \right)    }\\
        &=e^{  - (m+r)\left[\frac{r}{m+r}\left(\ln(\frac{r}{m+r})-\ln p  \right) + (1-\frac{r}{m+r})\left(\ln(1-\frac{r}{m+r}) -\ln (1-p) \right) \right]   }\\
        &\leqt{(i)} e^{  - 2(m+r)(p-\frac{r}{m+r})^2   } \,,
    \ee
    where in (i) we exploited the elementary inequality
    \bb
        p\ln\left(\frac{p}{q}\right)+(1-p)\ln\left(\frac{1-p}{1-q}\right)\ge 2(p-q)^2\quad \forall\, p,q\in(0,1)\,,
    \ee
    which can be proved e.g.~exploiting Pinsker's inequality. Hence, by employing \eqref{eq244}, we obtain that
    \bb\label{ineq_ratio_q_qinf}
        \frac{q_{k,l}}{q_{\infty,l}}\ge 1 -e^{-2 (k+2)(p-\frac{l+1}{k+2})^2   }\,,
    \ee 
    where we defined $p\coloneqq \frac{\lambda N_s+1}{N_s+1}$.  Additionally, note that it always holds that $\frac{q_{k,l}}{q_{\infty,l}}\le 1$ (this can be proved e.g.~exploiting \eqref{eq244}).

    Now, let us proceed with the proof of $\lim\limits_{k\rightarrow\infty}\sum_{l=0}^{k}q_{k,l}\log_2^2 q_{k,l}=\sum_{l=0}^{\infty}q_{\infty,l}\log_2^2 q_{\infty,l}$. Let us split the sum as
    \bb
        \sum_{l=0}^{k}q_{k,l}\log_2^2 q_{k,l}=\sum_{l=0}^{\lfloor\frac{p}{2}(k+2)\rfloor+1}q_{k,l}\log_2^2 q_{k,l}+\sum_{l=\lfloor\frac{p}{2}(k+2)\rfloor+2}^{k}q_{k,l}\log_2^2 q_{k,l}\,,
    \ee
    In this way, in the first sum we can assume that $(p-\frac{l+1}{k+2})^2\ge p^2/4$ and thus the following lower bound holds
    \bb\label{ineq_useful0}
        \frac{q_{k,l}}{q_{\infty,l}}\ge 1 -e^{-(k+2)p^2/2   }\,.
    \ee 
    Let us prove that 
    \bb\label{eq_00101}
        \lim\limits_{k\rightarrow\infty}\sum_{l=0}^{\lfloor\frac{p}{2}(k+2)\rfloor+1}q_{k,l}\log_2^2 q_{k,l}&=\sum_{l=0}^{\infty}q_{\infty,l}\log_2^2 q_{\infty,l}\,,\\
        \lim\limits_{k\rightarrow\infty}\sum_{l=\lfloor\frac{p}{2}(k+2)\rfloor+2}^{k}q_{k,l}\log_2^2 q_{k,l}&=0\,.
    \ee
    Let us begin with the first limit. Let us write
    \bb
        \sum_{l=0}^{\lfloor\frac{p}{2}(k+2)\rfloor+1}q_{k,l}\log_2^2 q_{k,l}=\sum_{l=0}^{\lfloor\frac{p}{2}(k+2)\rfloor+1}q_{k,l}\log_2^2 \frac{q_{k,l}}{q_{\infty,l}} +2\sum_{l=0}^{\lfloor\frac{p}{2}(k+2)\rfloor+1}q_{k,l}\log_2 \frac{q_{k,l}}{q_{\infty,l}}\log_2 q_{\infty,l}+ \sum_{l=0}^{\lfloor\frac{p}{2}(k+2)\rfloor+1}q_{k,l}\log_2^2 q_{\infty,l}\,,
    \ee
    For the first term, we have that
    \bb
        0\le \sum_{l=0}^{\lfloor\frac{p}{2}(k+2)\rfloor+1}q_{k,l}\log_2^2 \frac{q_{k,l}}{q_{\infty,l}}&\le \sum_{l=0}^{\lfloor\frac{p}{2}(k+2)\rfloor+1}q_{\infty,l}\log_2^2 \frac{q_{k,l}}{q_{\infty,l}}\\
        &\le \log_2^2\left(1 -e^{-(k+2)p^2/2   }\right) \sum_{l=0}^{\lfloor\frac{p}{2}(k+2)\rfloor+1}q_{\infty,l}\\
        &\le \log_2^2\left(1 -e^{-(k+2)p^2/2   }\right)\xrightarrow{k\rightarrow\infty}0\,,
    \ee
    For the second term, we obtain that
    \bb
        0\le\sum_{l=0}^{\lfloor\frac{p}{2}(k+2)\rfloor+1}q_{k,l}\log_2 \frac{q_{k,l}}{q_{\infty,l}}\log_2 q_{\infty,l}&\le \sum_{l=0}^{\lfloor\frac{p}{2}(k+2)\rfloor+1}q_{\infty,l}\log_2 \frac{q_{k,l}}{q_{\infty,l}}\log_2 q_{\infty,l}\\
        &\le \log_2 \left(1 -e^{-(k+2)p^2/2   }\right)\sum_{l=0}^{\lfloor\frac{p}{2}(k+2)\rfloor+1}q_{\infty,l}\log_2 q_{\infty,l}\\
        &\le \log_2 \left(1 -e^{-(k+2)p^2/2   }\right)\sum_{l=0}^{\infty}q_{\infty,l}\log_2 q_{\infty,l}\\
        &= -\log_2 \left(1 -e^{-(k+2)p^2/2   }\right)S(\Psi_B)\xrightarrow{k\rightarrow\infty}0\,,
    \ee
    where in the last line we used that $S(\Psi_B)=h(\lambda N_s)$ is finite. Finally, for the third term, we observe that
    \bb
        \sum_{l=0}^{\lfloor\frac{p}{2}(k+2)\rfloor+1}q_{k,l}\log_2^2 q_{\infty,l}= \sum_{l=0}^{\lfloor\frac{p}{2}(k+2)\rfloor+1}q_{\infty,l}\log_2^2 q_{\infty,l}- \sum_{l=0}^{\lfloor\frac{p}{2}(k+2)\rfloor+1}(q_{\infty,l}-q_{k,l})\log_2^2 q_{\infty,l}\,
    \ee
    and that
    \bb
        \lim\limits_{k\rightarrow\infty}\sum_{l=0}^{\lfloor\frac{p}{2}(k+2)\rfloor+1}q_{\infty,l}\log_2^2 q_{\infty,l}&=\sum_{l=0}^{\infty}q_{\infty,l}\log_2^2 q_{\infty,l}\,.
    \ee
    Additionally, it holds that
    \bb
    0\le \sum_{l=0}^{\lfloor\frac{p}{2}(k+2)\rfloor+1}(q_{\infty,l}-q_{k,l})\log_2^2 q_{\infty,l}&=\sum_{l=0}^{\lfloor\frac{p}{2}(k+2)\rfloor+1}q_{\infty,l}(1-\frac{q_{k,l}}{q_{\infty,l}})\log_2^2 q_{\infty,l}\\
    &\le e^{-(k+2)p^2/2   }\sum_{l=0}^{\lfloor\frac{p}{2}(k+2)\rfloor+1}q_{\infty,l}\log_2^2 q_{\infty,l}\\
    &\le e^{-(k+2)p^2/2   }\sum_{l=0}^{\infty}q_{\infty,l}\log_2^2 q_{\infty,l}\\
    &= e^{-(k+2)p^2/2   }\Tr[\Psi_B\log_2^2\Psi_B]\xrightarrow{k\rightarrow\infty}0\,,   
    \ee
    where in the last line we used that $\Tr[\Psi_B\log_2^2\Psi_B]$ is finite (see e.g.~\ref{relvariances}). Hence, we proved that the third term satisfies
    \bb 
        \lim\limits_{k\rightarrow\infty}\sum_{l=0}^{\lfloor\frac{p}{2}(k+2)\rfloor+1}q_{k,l}\log_2^2 q_{\infty,l}= \sum_{l=0}^{\infty}q_{\infty,l}\log_2^2 q_{\infty,l}\,,
    \ee
    and thus the proof of the first limit in \eqref{eq_00101} is complete. Let us now proceed with the proof of the second one, i.e.~ $\lim\limits_{k\rightarrow\infty}\sum_{l=\lfloor\frac{p}{2}(k+2)\rfloor+2}^{k}q_{k,l}\log_2^2 q_{k,l}=0$. Similarly as above, it holds that
    \bb
        \sum_{l=\lfloor\frac{p}{2}(k+2)\rfloor+2}^{k}q_{k,l}\log_2^2 q_{k,l}=\sum_{l=\lfloor\frac{p}{2}(k+2)\rfloor+2}^{k}q_{k,l}\log_2^2 \frac{q_{k,l}}{q_{\infty,l}} +2\sum_{l=\lfloor\frac{p}{2}(k+2)\rfloor+2}^{k}q_{k,l}\log_2 \frac{q_{k,l}}{q_{\infty,l}}\log_2 q_{\infty,l}+ \sum_{l=\lfloor\frac{p}{2}(k+2)\rfloor+2}^{k}q_{k,l}\log_2^2 q_{\infty,l}\,,
    \ee
    To show that each of these terms goes to zero, we use the fact that 
    \bb 
        \frac{q_{k,l}}{q_{\infty,l}}&= \sum_{n=0}^{k-l} \binom{n+l}{l}\left( \frac{(1-\lambda)N_s}{N_s+1}  \right)^{n}\left(\frac{ \lambda N_s+1}{N_s+1}\right)^{l+1} \ge \left(\frac{ \lambda N_s+1}{N_s+1}\right)^{l+1}=p^{l+1}\,,
    \ee
   where the first inequality follows from \eqref{eq244}.
   For the first term, we have that
   \bb
        0\le\sum_{l=\lfloor\frac{p}{2}(k+2)\rfloor+2}^{k}q_{k,l}\log_2^2 \frac{q_{k,l}}{q_{\infty,l}}&\leq \sum_{l=\lfloor\frac{p}{2}(k+2)\rfloor+2}^{k}q_{\infty,l}\log_2^2 \frac{q_{k,l}}{q_{\infty,l}}\\
        &\le ( \log_2 p)^2 \sum_{l=\lfloor\frac{p}{2}(k+2)\rfloor+2}^{k}q_{\infty,l}(l+1)^2\xrightarrow{k\rightarrow\infty}0\,,
   \ee
   where in the last line we used that $\Tr[\Psi_B (a^\dagger a)^2]$ is finite (recall that $\Psi_B=\tau_{\lambda N_s}$ is a thermal state). Similarly, for the second term, we obtain that
   \bb
        0\le \sum_{l=\lfloor\frac{p}{2}(k+2)\rfloor+2}^{k}q_{k,l}\log_2 \frac{q_{k,l}}{q_{\infty,l}}\log_2 q_{\infty,l}&\le   \sum_{l=\lfloor\frac{p}{2}(k+2)\rfloor+2}^{k}q_{\infty,l}\log_2 \frac{q_{k,l}}{q_{\infty,l}}\log_2 q_{\infty,l}\\
        &\le  ( \log_2 p)\sum_{l=\lfloor\frac{p}{2}(k+2)\rfloor+2}^{k}q_{\infty,l}(l+1)\log_2 q_{\infty,l}\\
        &=  ( \log_2 p)\sum_{l=\lfloor\frac{p}{2}(k+2)\rfloor+2}^{k}q_{\infty,l}(l+1)\log_2 \left[\frac{1}{\lambda N_s+1}\left(\frac{\lambda N_s}{\lambda N_s+1}\right)^l\right]\xrightarrow{k\rightarrow\infty} 0\,,
   \ee
   where in the last line we used that both $\Tr[\Psi_B (a^\dagger a)^2]$ and $\Tr[\Psi_B a^\dagger a]$ are finite, as $\Psi_{B}=\tau_{\lambda N_s}$ is a thermal state. Finally, for the third term, we have that
   \bb
        0\le \sum_{l=\lfloor\frac{p}{2}(k+2)\rfloor+2}^{k}q_{k,l}\log_2^2 q_{\infty,l}&\le \sum_{l=\lfloor\frac{p}{2}(k+2)\rfloor+2}^{k}q_{\infty,l}\log_2^2 q_{\infty,l}\xrightarrow{k\rightarrow\infty} 0\,,
   \ee
   where in the last line we exploited that $\Tr[\Psi_B\log_2^2\Psi_B]$ is finite (see e.g.~\ref{relvariances}). Consequently, the proof of the second limit in \eqref{eq_00101} is complete. To summarise, we have proved that $
        \lim\limits_{k\rightarrow\infty}\sum_{l=0}^{k}q_{k,l}\log_2^2 q_{k,l}=\sum_{l=0}^{\infty}q_{\infty,l}\log_2^2 q_{\infty,l}$, which implies the validity of the second limit in \eqref{limits_k_to_prove}.

    The proof of the third limit in \eqref{limits_k_to_prove}, i.e.~$\lim\limits_{k\rightarrow\infty}\Tr[\Psi^{(k)}_{E}\log_2^2\Psi^{(k)}_{E}]=\Tr[\Psi_{E}\log_2^2\Psi_{E}]$, is completely analogous to the proof of $\lim\limits_{k\rightarrow\infty}\Tr[\Psi^{(k)}_{B}\log_2^2\Psi^{(k)}_{B}]=\Tr[\Psi_{B}\log_2^2\Psi_{B}]$.

    Let us now show that the fourth limit in \eqref{limits_k_to_prove}, i.e.~$\lim\limits_{k\rightarrow\infty} \Tr[\Psi^{(k)}_{AE}\log_2\Psi^{(k)}_{AE}\log_2\Psi^{(k)}_{E}]=\Tr[\Psi_{AE}\log_2\Psi_{AE}\log_2\Psi_{E}]$. First, note that
    \bb 
        \Tr[\Psi_{AE}\log_2\Psi_{AE}\log_2\Psi_{E}]&=\Tr[\Psi_{ABE}\log_2\Psi_{AE}\log_2\Psi_{E}]\\
        &=\Tr[\Psi_{ABE}\log_2\Psi_{B}\log_2\Psi_{E}]\\
        &=\Tr[\Psi_{BE}\log_2\Psi_{B}\log_2\Psi_{E}]\,,
    \ee
    where in the last line we used the general fact that for any bipartite state $\ket{\psi}_{S_1S_2}$ on systems $S_1,S_2$ it holds that
    \bb
        \left(\Tr_{S_1}\ketbra{\psi}_{S_1S_2}\otimes \mathbb{1}_{S_2}\right)\ket{\psi}_{S_1S_2}=  \left(\mathbb{1}_{S_1}\otimes\Tr_{S_2}\ketbra{\psi}_{S_1S_2}\right)\ket{\psi}_{S_1S_2}\,,
    \ee
    which can be shown e.g.~via Schmidt decomposition. Analogously, it holds that 
    \bb
        \Tr[\Psi^{(k)}_{AE}\log_2\Psi^{(k)}_{AE}\log_2\Psi^{(k)}_{E}]=\Tr[\Psi^{(k)}_{BE}\log_2\Psi^{(k)}_{B}\log_2\Psi^{(k)}_{E}]\,.
    \ee
    Hence, we need to prove that 
    \bb
    \lim\limits_{k\rightarrow\infty}\Tr[\Psi^{(k)}_{BE}\log_2\Psi^{(k)}_{B}\log_2\Psi^{(k)}_{E}]=\Tr[\Psi_{BE}\log_2\Psi_{B}\log_2\Psi_{E}]\,.
    \ee
    From \eqref{eq_psi_k_abe_formula} it follows that
    \bb
		\Psi^{(k)}_{B}&=\frac{1}{p_{k}}\sum_{l=0}^k q_{k, l}\ketbra{l}\,,\\
		\Psi^{(k)}_{E}&=\frac{1}{p_{k}}\sum_{l=0}^k q'_{k, l}\ketbra{l}\,,
	\ee
where we introduced
\bb
q_{k, l}&\coloneqq \sum_{n=l}^k \frac{1}{N_s+1}\left(\frac{N_s}{N_s+1}\right)^n\binom{n}{l}\lambda^l(1-\lambda)^{n-l}\,,\\
q'_{k, l}&\coloneqq \sum_{n=l}^k \frac{1}{N_s+1}\left(\frac{N_s}{N_s+1}\right)^n\binom{n}{l}(1-\lambda)^l\lambda^{n-l}\,,\\
p_{k}&\coloneqq \Tr[p_k \tau_{N_s}]=\frac{1}{N_s+1}\sum_{n=0}^{k}\left(\frac{N_s}{N_s+1}\right)^n=1-\left(\frac{N_s}{N_s+1}\right)^{k+1}\,.
\ee    
Moreover, it holds that
\bb
	       \Psi^{(k)}_{BE}=\frac{1}{p_k}\sum_{n=0}^k\sum_{l_1,l_2=0}^n\frac{1}{N_s+1}\left(\frac{N_s}{N_s+1}\right)^{n}(-1)^{l_1+l_2}\sqrt{\binom{n}{l_1}\binom{n}{l_2} \lambda^{2n-l_1-l_2}(1-\lambda)^{l_1+l_2}}\ketbraa{n-l_1}{n-l_2}_B\otimes\ketbraa{l_1}{l_2}_E\,.
\ee
Consequently, we obtain that
\bb
	\Tr[\Psi^{(k)}_{BE}\log_2\Psi^{(k)}_{B}\log_2\Psi^{(k)}_{E}]&=\sum_{l=0}^{\infty} \log_2\frac{q'_{k,l}}{p_k}\Tr[\Psi^{(k)}_{BE}\log_2\Psi^{(k)}_{B}\ketbra{l}_E]\\
	&= \sum_{l=0}^{k}\sum_{n=l}^{k} \frac{\tilde{q}_{l,n}}{p_k} \log_2\frac{q_{k,n-l}}{p_k}\log_2\frac{q'_{k,l}}{p_k}\,,
\ee
where we defined
\bb
		\tilde{q}_{l,n}\coloneqq  \frac{1}{N_s+1}\left(\frac{N_s}{N_s+1}\right)^{n}\binom{n}{l} \lambda^{n-l}(1-\lambda)^{l}\,.
\ee
Moreover, it can be easily shown that (e.g.~by considering the partial traces of $\Psi_{BE}$)
\bb
\sum_{n=l}^k \tilde{q}_{n-l,n}&= q_{k,l}\,,\\
\sum_{n=l}^k \tilde{q}_{l,n}&=q'_{k,l}\,.
\ee
Analogously, it holds that
\bb
		\Tr[\Psi_{BE}\log_2\Psi_{B}\log_2\Psi_{E}]= \sum_{l=0}^{\infty}\sum_{n=l}^{\infty} \tilde{q}_{l,n}  \log_2q_{\infty,n-l}\log_2q'_{\infty,l}\,,
\ee
where
\bb
	q_{\infty,l}&\coloneqq\frac{1}{\lambda N_s+1}\left(\frac{\lambda N_s}{\lambda N_s+1}\right)^l\,,\\
	q'_{\infty,l}&\coloneqq\frac{1}{(1-\lambda) N_s+1}\left(\frac{(1-\lambda )N_s}{(1-\lambda) N_s+1}\right)^l\,.
\ee
Hence, we have to show that
\bb
	\lim\limits_{k\rightarrow\infty} \sum_{l=0}^{k}\sum_{n=l}^{k} \frac{\tilde{q}_{l,n}}{p_k} \log_2\frac{q_{k,n-l}}{p_k}\log_2\frac{q'_{k,l}}{p_k}=\sum_{l=0}^{\infty}\sum_{n=l}^{\infty} \tilde{q}_{l,n}  \log_2q_{\infty,n-l}\log_2q'_{\infty,l}\,.
\ee
Similarly to what we did  in \eqref{simplification_without_pk}, by exploiting the facts that (i) $\lim\limits_{k\rightarrow\infty}p_k=1$, (ii) the state $\Psi^{(k)}_B$ (resp.~$\Psi^{(k)}_E$) converges to $\Psi_{B}$ (resp. $\Psi^{(k)}_E$) in trace norm (as proved in the proof of Theorem~\ref{thm_lower_bound}), which has finite entropy, and (iii) the von Neumann entropy is continuous with respect the trace norm over the set of states having bounded mean photon number~\cite{tightuniform}, it follows that we just need to prove that
\bb
	\lim\limits_{k\rightarrow\infty} \sum_{l=0}^{k}\sum_{n=l}^{k}  \tilde{q}_{l,n}  \log_2q_{k,n-l} \log_2q'_{k,l}=\sum_{l=0}^{\infty}\sum_{n=l}^{\infty} \tilde{q}_{l,n}  \log_2q_{\infty,n-l}\log_2q'_{\infty,l}\,.
\ee
Moreover, we have that
\bb
 	 \sum_{l=0}^{k}\sum_{n=l}^{k}  \tilde{q}_{l,n}  \log_2q_{k,n-l} \log_2q'_{k,l}&=\sum_{l=0}^{k}\sum_{n=l}^{k}  \tilde{q}_{l,n}  \log_2\frac{q_{k,n-l}}{q_{\infty,n-l}} \log_2q'_{k,l} \\&\quad+ \sum_{l=0}^{k}\sum_{n=l}^{k}  \tilde{q}_{l,n}  \log_2q_{\infty,n-l} \log_2\frac{q'_{k,l}}{q'_{\infty,l}}+  \sum_{l=0}^{k}\sum_{n=l}^{k}  \tilde{q}_{l,n}  \log_2q_{\infty,n-l} \log_2q'_{\infty,l}\,.
\ee
Since the last term converges to $\sum_{l=0}^{\infty}\sum_{n=l}^{\infty} \tilde{q}_{l,n}  \log_2q_{\infty,n-l}\log_2q'_{\infty,l}$ as $k\rightarrow\infty$, we just need to show that  
\bb\label{lim_zero_toprove}
	\lim\limits_{k\rightarrow\infty} \sum_{l=0}^{k}\sum_{n=l}^{k}  \tilde{q}_{l,n}  \log_2\frac{q_{k,n-l}}{q_{\infty,n-l}} \log_2q'_{k,l} &=0\,,\\
	\lim\limits_{k\rightarrow\infty}\sum_{l=0}^{k}\sum_{n=l}^{k}  \tilde{q}_{l,n}  \log_2q_{\infty,n-l} \log_2\frac{q'_{k,l}}{q'_{\infty,l}}&=0\,.
\ee
Let us prove the first limit. As we saw above, by defining $p\coloneqq\frac{\lambda N_s+1}{N_s+1}$, we have that  
\bb\label{hand1}
\frac{q_{k,l}}{q_{\infty,l}}&\ge 1 -e^{-(k+2)p^2/2   }\quad  \forall\,l\le\left\lfloor\frac{p}{2}(k+2)\right\rfloor+1\,,\\
\frac{q_{k,l}}{q_{\infty,l}}&\ge p^{l+1}\,,\\
q_{k,l}&\le q_{\infty,l}\,,\\
\sum_{n=l}^k \tilde{q}_{n-l,n}&= q_{k,l}\,.
\ee 
Analogously, by defining $p'\coloneqq\frac{(1-\lambda) N_s+1}{N_s+1}$, we have that  
\bb\label{hand2}
\frac{q'_{k,l}}{q'_{\infty,l}}&\ge 1 -e^{-(k+2)p'^2/2   }\quad  \forall\,l\le\left\lfloor\frac{p'}{2}(k+2)\right\rfloor+1\,,\\
\frac{q'_{k,l}}{q_{\infty,l}}&\ge p'^{l+1}\,,\\
q'_{k,l}&\le q'_{\infty,l}\,,\\
\sum_{n=l}^k \tilde{q}_{l,n}&=q'_{k,l}\,.
\ee 
With the relations in \eqref{hand1} and \eqref{hand2} at hand, we obtain that
\bb
0&\le \sum_{l=0}^{k}\sum_{n=l}^{k}  \tilde{q}_{l,n}  \log_2\frac{q_{k,n-l}}{q_{\infty,n-l}} \log_2q'_{k,l} \\
&=  \sum_{n=0}^{k} \sum_{l=0}^{n} \tilde{q}_{l,n}  \log_2\frac{q_{k,n-l}}{q_{\infty,n-l}} \log_2q'_{k,l} \\
&= \sum_{n=0}^{   \lfloor\frac{p}{2}(k+2)\rfloor+1   } \sum_{l=0}^{n} \tilde{q}_{l,n}  \log_2\frac{q_{k,n-l}}{q_{\infty,n-l}} \log_2q'_{k,l}+   \sum_{n=\lfloor\frac{p}{2}(k+2)\rfloor+2}^{k} \sum_{l=0}^{n} \tilde{q}_{l,n}  \log_2\frac{q_{k,n-l}}{q_{\infty,n-l}} \log_2q'_{k,l}\\
&\le\log_2\left(1 -e^{-(k+2)p^2/2   } \right) \sum_{n=0}^{   \lfloor\frac{p}{2}(k+2)\rfloor+1   } \sum_{l=0}^{n} \tilde{q}_{l,n}   \log_2q'_{k,l}+  \sum_{n=\lfloor\frac{p}{2}(k+2)\rfloor+2}^{k} \sum_{l=0}^{n} \tilde{q}_{l,n}  \log_2\frac{q_{k,n-l}}{q_{\infty,n-l}} \log_2q'_{k,l}\\
&\le\log_2\left(1 -e^{-(k+2)p^2/2   } \right) \sum_{n=0}^{k} \sum_{l=0}^{n} \tilde{q}_{l,n}   \log_2q'_{k,l}+  \sum_{n=\lfloor\frac{p}{2}(k+2)\rfloor+2}^{k} \sum_{l=0}^{n} \tilde{q}_{l,n}  \log_2\frac{q_{k,n-l}}{q_{\infty,n-l}} \log_2q'_{k,l}\\
&=\log_2\left(1 -e^{-(k+2)p^2/2   } \right)  \sum_{l=0}^{k}q'_{k,l}   \log_2q'_{k,l}+  \sum_{n=\lfloor\frac{p}{2}(k+2)\rfloor+2}^{k} \sum_{l=0}^{n} \tilde{q}_{l,n}  \log_2\frac{q_{k,n-l}}{q_{\infty,n-l}} \log_2q'_{k,l}
\ee
The first term goes to zero as a consequence of the facts that $\log_2\left(1 -e^{-(k+2)p^2/2   } \right)  \xrightarrow{k\rightarrow\infty} 0$ and that $\sum_{l=0}^{k}q'_{k,l}   \log_2q'_{k,l}$ converges to $S(\Psi_E)$, which is finite. Let us analyse the second term:
\bb
	0&\le\sum_{n=\lfloor\frac{p}{2}(k+2)\rfloor+2}^{k} \sum_{l=0}^{n} \tilde{q}_{l,n}  \log_2\frac{q_{k,n-l}}{q_{\infty,n-l}} \log_2q'_{k,l}\\
	&\le \log_2p \sum_{n=\lfloor\frac{p}{2}(k+2)\rfloor+2}^{k} \sum_{l=0}^{n} \tilde{q}_{l,n}  (n-l+1) \log_2q'_{k,l}\\
	&= \log_2p \sum_{n=\lfloor\frac{p}{2}(k+2)\rfloor+2}^{k} \sum_{l=0}^{n} \tilde{q}_{l,n}  (n-l+1) \log_2\frac{q'_{k,l}}{q'_{\infty,l}}+ \log_2p \sum_{n=\lfloor\frac{p}{2}(k+2)\rfloor+2}^{k} \sum_{l=0}^{n} \tilde{q}_{l,n}  (n-l+1) \log_2 q'_{\infty,l}\\
	&= \log_2p \sum_{n=\lfloor\frac{p}{2}(k+2)\rfloor+2}^{k} \sum_{l=0}^{n} \tilde{q}_{l,n}  (n-l+1) \log_2\frac{q'_{k,l}}{q'_{\infty,l}}+ \log_2p \sum_{n=\lfloor\frac{p}{2}(k+2)\rfloor+2}^{k} \sum_{l=0}^{n} \tilde{q}_{l,n}  (n-l+1) \log_2\left( \frac{1}{(1-\lambda) N_s+1}\left(\frac{(1-\lambda )N_s}{(1-\lambda) N_s+1}\right)^l \right)\\
	&\le  (\log_2p)^2 \sum_{n=\lfloor\frac{p}{2}(k+2)\rfloor+2}^{k} \sum_{l=0}^{n} \tilde{q}_{l,n}  (n-l+1) (l+1)+ \log_2p \sum_{n=\lfloor\frac{p}{2}(k+2)\rfloor+2}^{k} \sum_{l=0}^{n} \tilde{q}_{l,n}  (n-l+1) \log_2\left( \frac{1}{(1-\lambda) N_s+1}\left(\frac{(1-\lambda )N_s}{(1-\lambda) N_s+1}\right)^l \right)\,.
\ee
These two sums converge to zero because the series 
\bb
	\Tr[\Psi_{BE} \hat{N}_B\otimes \hat{N}_E  ]=\sum_{n=0}^\infty \sum_{l=n}^\infty \tilde{q}_{l,n}(n-l)l
\ee
 converges (specifically, it equals $2 \lambda(1-\lambda) N_s^2$), as a consequence of \eqref{eq_enbe}. Hence, we have proved the first limit in \eqref{lim_zero_toprove}. Now, let us prove the second limit, i.e.~$	\lim\limits_{k\rightarrow\infty}\sum_{l=0}^{k}\sum_{n=l}^{k}  \tilde{q}_{l,n}  \log_2q_{\infty,n-l} \log_2\frac{q'_{k,l}}{q'_{\infty,l}}=0$. Similarly to what we did with the first limit, by exploiting \eqref{hand1} and \eqref{hand2}, we obtain that
 \bb
 		0&\le \sum_{l=0}^{k}\sum_{n=l}^{k}  \tilde{q}_{l,n}  \log_2q_{\infty,n-l} \log_2\frac{q'_{k,l}}{q'_{\infty,l}}\\
 		& =\sum_{n=0}^{k} \sum_{l=0}^{n} \tilde{q}_{l,n}  \log_2q_{\infty,n-l} \log_2\frac{q'_{k,l}}{q'_{\infty,l}}\\
 		 & =\sum_{n=0}^{  \lfloor\frac{p'}{2}(k+2)\rfloor+1  } \sum_{l=0}^{n} \tilde{q}_{l,n}  \log_2q_{\infty,n-l} \log_2\frac{q'_{k,l}}{q'_{\infty,l}}+\sum_{n=  \lfloor\frac{p'}{2}(k+2)\rfloor+2 }^{k} \sum_{l=0}^{n} \tilde{q}_{l,n}  \log_2q_{\infty,n-l} \log_2\frac{q'_{k,l}}{q'_{\infty,l}}\\
 		 &\le  \log_2\left(    1 -e^{-(k+2)p'^2/2}  \right) \sum_{n=0}^{  \lfloor\frac{p'}{2}(k+2)\rfloor+1  } \sum_{l=0}^{n} \tilde{q}_{l,n}  \log_2q_{\infty,n-l} +\log_2 p'\sum_{n=  \lfloor\frac{p'}{2}(k+2)\rfloor+2 }^{k} \sum_{l=0}^{n} \tilde{q}_{l,n}  \log_2q_{\infty,n-l} (l+1)\\
 		 &= \log_2\left(    1 -e^{-(k+2)p'^2/2}  \right) \sum_{n=0}^{  \lfloor\frac{p'}{2}(k+2)\rfloor+1  } \sum_{l=0}^{n} \tilde{q}_{l,n} \log_2\left( \frac{1}{\lambda N_s+1}\left(\frac{\lambda N_s}{\lambda N_s+1}\right)^{n-l}   \right)\\
 		 &\quad+\log_2 p'\sum_{n=  \lfloor\frac{p'}{2}(k+2)\rfloor+2 }^{k} \sum_{l=0}^{n} \tilde{q}_{l,n}\log_2\left( \frac{1}{\lambda N_s+1}\left(\frac{\lambda N_s}{\lambda N_s+1}\right)^{n-l}   \right)(l+1),.
 \ee
Analogously as above, both these sums converge to zero because the state $\Psi_{BE}$ has finite moments. Consequently, we have proved that $    \lim\limits_{k\rightarrow\infty}\Tr[\Psi^{(k)}_{BE}\log_2\Psi^{(k)}_{B}\log_2\Psi^{(k)}_{E}]=\Tr[\Psi_{BE}\log_2\Psi_{B}\log_2\Psi_{E}]$. This concludes the proof of $\lim\limits_{k\to\infty}V(A|E)_{\Psi^{(k)}}=V(A|E)_{\Psi}$.
\end{proof}
We are now ready to bound the $n$-shot capacities of the pure loss channel. Let us start from the energy-constrained case.
\begin{thm}[(Lower bound on the energy-constrained $n$-shot capacities of the pure loss channel via entropy variance approach)]\label{thm_lower_bound_twoway_rel_ec}
Let $N_s>0$, $\lambda\in[0,1]$, $\varepsilon\in(0,1)$, and $n\in\mathbb{N}$. The energy-constrained $n$-shot quantum capacity, the  energy-constrained $n$-shot two-way quantum capacity, and the energy-constrained $n$-shot secret-key capacity of the pure loss channel $\mathcal{E}_\lambda$ can be lower bounded as follows:
\bb\label{EQ_ec_n_shot_ec}
    Q^{(\varepsilon,n)}(\mathcal{E}_{\lambda},N_s)&
        \geq  n \left(  h\!\left(\lambda N_s\right)-h\!\left((1-\lambda)N_s\right)
  \right)-4\sqrt{\frac{nV(A|E)_{\Psi^{(\lambda,N_s)}}}{\varepsilon}}-\log_2\!\left(\frac{2^{23}(32-\varepsilon)^2}{(16-\varepsilon)\varepsilon^6 }\right)\,,\\
        K^{(\varepsilon,n)}(\mathcal{E}_{\lambda},N_s)&\ge Q_2^{(\varepsilon,n)}(\mathcal{E}_{\lambda},N_s) \\
        &\ge n \left(h\!\left( N_s\right)-h\!\left((1-\lambda)N_s\right)\right)-\sqrt{\frac{2nV(B|E)_{\Psi^{(\lambda,N_s)}}}{\sqrt{\varepsilon}}}-\log_2\!\left(\frac{2^{6}\,3\,(4-\sqrt{\varepsilon})^2}{(2-\sqrt{\varepsilon})\varepsilon^3 }\right)\,,
    \ee
    where $h(x)\coloneqq (x+1)\log_2(x+1)-x\log_2x$ and where the conditional entropy variances are given by
\bb\label{eq_variance_thm_ec}
V(A|E)_{\Psi^{(\lambda,N_s)}} &=\lambda N_s(1+\lambda N_s)\left(\log_2\frac{\lambda N_s}{1+\lambda N_s}\right)^2 +(1-\lambda)N_s(1+(1-\lambda)N_s)\left(\log_2\frac{(1-\lambda)N_s}{1+(1-\lambda)N_s}\right)^2\\
&\quad-2(1-\lambda)\lambda N_s^2\log_2\frac{(1-\lambda)N_s}{1+(1-\lambda)N_s}\log_2\frac{\lambda N_s}{1+\lambda N_s}\,,\\
V(B|E)_{ \Psi^{(\lambda,N_s)} } &=N_s(1+N_s)\left(\log_2\frac{N_s}{1+N_s}\right)^2 +(1-\lambda)N_s(1+(1-\lambda)N_s)\left(\log_2\frac{(1-\lambda)N_s}{1+(1-\lambda)N_s}\right)^2\\
&\quad-2(1-\lambda)N_s(1+N_s)\left(\log_2\frac{(1-\lambda)N_s}{1+(1-\lambda)N_s}\right)\left(\log_2\frac{N_s}{1+N_s}\right)\,.
\ee
\end{thm}
\begin{proof}
By exploiting Lemma~\ref{lem_variance-bound-k-Q} and Lemma~\ref{lem_variance-bound-k-L}, we have that
\bb
Q^{(\varepsilon,n)}(\mathcal{E}_{\lambda},N_s)&\ge- n D(\Psi^{(k)}_{AE}\|\mathbb{1}_A \otimes \Psi^{(k)}_{E})-4\sqrt{\frac{{nV(A|E)_{\Psi^{(k)}}}}{\varepsilon}}-\log_2\!\left(\frac{2^{23}(32-\varepsilon)^2}{(16-\varepsilon)\varepsilon^6 }\right)\,,\\
K^{(\varepsilon,n)}(\mathcal{E}_{\lambda},N_s)&\ge Q_2^{(\varepsilon,n)}(\mathcal{E}_{\lambda},N_s) \\
&\ge  - n D(\Psi^{(k)}_{BE}\|\mathbb{1}_A \otimes \Psi^{(k)}_{E})-\sqrt{\frac{{2nV(B|E)_{\Psi^{(k)}}}}{\sqrt{\varepsilon}}}-\log_2\!\left(\frac{2^{6}\,3\,(4-\sqrt{\varepsilon})^2}{(2-\sqrt{\varepsilon})\varepsilon^3 }\right)\,.
\ee
The lower bounds in \eqref{EQ_ec_n_shot_ec} follow by taking the limit for $k\rightarrow \infty$ of the above lower bounds. Specifically, we do this by exploiting Lemma~\ref{lemma_cont_cond_var} and the fact, proved in the proof of Theorem~\ref{thm_lower_bound} (see the first equation of \eqref{limit_to_prove}), that
\bb
    \lim_{k\rightarrow\infty}D(\Psi^{(k)}_{AE}\|\mathbb{1}_A \otimes \Psi^{(k)}_{E})&=D(\Psi^{(\lambda,N_s)}_{AE}\|\mathbb{1}_A \otimes \Psi^{(\lambda,N_s)}_{E})=-I_c(A\,\rangle\, B)_{ \Psi^{(\lambda,N_s)} }=-h(\lambda N_s)+h((1-\lambda) N_s)\,,\\
    \lim_{k\rightarrow\infty}D(\Psi^{(k)}_{BE}\|\mathbb{1}_B \otimes \Psi^{(k)}_{E})&=D(\Psi^{(\lambda,N_s)}_{BE}\|\mathbb{1}_B \otimes \Psi^{(\lambda,N_s)}_{E})=-I_c(B\,\rangle\, A)_{ \Psi^{(\lambda,N_s)} }=-h( N_s)+h((1-\lambda) N_s)\,,
\ee
where the expression of the coherent information quantities are reported in \eqref{q_ec_pure_loss} and \eqref{q2_ec_pure_loss}. 
\end{proof}
By taking the limit $N_s\rightarrow\infty$ of \eqref{EQ_ec_n_shot_ec}, we can find a lower bound on the (unconstrained) $n$-shot capacities of the pure loss channel.
\begin{thm}[(Lower bound on the $n$-shot capacities of the pure loss channel via entropy variance approach)]\label{thm_lower_bound_twoway_rel}
Let $\lambda\in[0,1]$, $\varepsilon\in(0,1)$, and $n\in\mathbb{N}$. The $n$-shot quantum capacity, the $n$-shot two-way quantum capacity, and the $n$-shot secret-key capacity of the pure loss channel $\mathcal{E}_\lambda$ can be lower bounded as follows:
\bb\label{EQ_n_shot_nonEC}
    Q^{(\varepsilon,n)}(\mathcal{E}_{\lambda}) 
        &\ge n Q(\mathcal{E}_{\lambda}) - \log_2\left( \frac{2^{23}(32-\varepsilon)^2}{ (16-\varepsilon)\varepsilon^6}\right)      
        \,,\\
        K^{(\varepsilon,n)}(\mathcal{E}_{\lambda})&\ge Q_2^{(\varepsilon,n)}(\mathcal{E}_{\lambda}) \ge n Q_2(\mathcal{E}_{\lambda})- \log_2\!\left(\frac{2^{6}\,3\,(4-\sqrt{\varepsilon})^2}{(2-\sqrt{\varepsilon})\varepsilon^3 }\right) \,,
\ee
where
\bb
    Q(\mathcal{E}_{\lambda})&=   
        \begin{cases}
        \log_2\!\left(\frac{\lambda}{1-\lambda}\right) &\text{if $\lambda\in(\frac{1}{2},1]$ ,} \\
        0 &\text{if $\lambda\in[0,\frac{1}{2}]$ ,}
    \end{cases} 
\ee
and $Q_2(\mathcal{E}_{\lambda})= \log_2\!\left(\frac{1}{1-\lambda}\right)$ are the (asymptotic) quantum capacity and two-way quantum capacity of the pure loss channel, respectively.
\end{thm}

\subsection{Best lower bounds}\label{sec_best_low}
Above we have found two different lower bounds on the $n$-shot capacities of the pure loss channel, derived by exploiting two different approaches: the asymptotic equipartition property approach (see Section~\ref{subsecAsymptotic}) and the entropy variance approach (see subsection~\ref{section_entropy_cariance}). Let us compare these two lower bounds.

The lower bounds derived exploiting the asymptotic equipartition property approach are stated in Theorem~\ref{thm_one_shot_Q_pure} and in Theorem~\ref{thm_one_shot_Q2_pure}, and they read
\bb
    Q^{(\varepsilon,n)}(\mathcal{E}_\lambda)&\ge nQ(\mathcal{E}_\lambda)-\sqrt{n}4\log_2\!\left( \sqrt{\frac{1-\lambda}{\lambda}}+ \sqrt{\frac{\lambda}{1-\lambda}} +1 \right)\sqrt{\log_2\!\left(\frac{2^9}{\varepsilon^2}\right)}-\log_2\!\left(\frac{2^{18}}{3\varepsilon^4}\right)\,,\\
    K^{(\varepsilon,n)}(\mathcal{E}_\lambda)&\ge Q_2^{(\varepsilon,n)}(\mathcal{E}_\lambda) \ge nQ_2(\mathcal{E}_\lambda)-\sqrt{n}4\log_2\!\left(\sqrt{1-\lambda}+\sqrt{\frac{1}{1-\lambda}}+1\right)\sqrt{\log_2\!\left(\frac{8}{\varepsilon}\right)}-\log_2\!\left(\frac{16}{\varepsilon^2}\right)\,,
\ee
while the lower bounds derived exploiting the relative entropy approach are stated in Theorem~\ref{thm_lower_bound_twoway_rel} and they read
\bb\label{eq_bound_2}
Q^{(\varepsilon,n)}(\mathcal{E}_{\lambda}) 
        &\ge n Q(\mathcal{E}_{\lambda}) - \log_2\left( \frac{2^{23}(32-\varepsilon)^2}{ (16-\varepsilon)\varepsilon^6}\right)      
        \,,\\
        K^{(\varepsilon,n)}(\mathcal{E}_{\lambda})&\ge Q_2^{(\varepsilon,n)}(\mathcal{E}_{\lambda}) \ge n Q_2(\mathcal{E}_{\lambda})- \log_2\!\left(\frac{2^{6}\,3\,(4-\sqrt{\varepsilon})^2}{(2-\sqrt{\varepsilon})\varepsilon^3 }\right) \,,
    \ee
where $Q(\mathcal{E}_{\lambda})=\max\!\left(0,\log_2\!\left(\frac{\lambda}{1-\lambda}\right)\right)$ is the (asymptotic) quantum capacity and $Q_2(\mathcal{E}_{\lambda})= \log_2\!\left(\frac{1}{1-\lambda}\right)$ is the (asymptotic) two-way quantum capacity. By inspection, we can observe that the best lower bounds are the ones based on the relative entropy approach reported in \eqref{eq_bound_2}, due to the absence of the term proportional to $\sqrt{n}$. Let us report the best lower bound in a quotable theorem.
\begin{thm}[(Best lower bounds on the $n$-shot capacities of the pure loss channel)]\label{best_bounds}
    Let $\lambda\in[0,1]$, $\varepsilon\in(0,1)$, and $n\in\mathbb{N}$. The $n$-shot quantum capacity, the $n$-shot two-way quantum capacity, and the $n$-shot secret-key capacity of the pure loss channel $\mathcal{E}_\lambda$ can be lower bounded as follows:
    \bb\label{ineq_best_bounds}
Q^{(\varepsilon,n)}(\mathcal{E}_{\lambda}) 
        &\ge n Q(\mathcal{E}_{\lambda}) - \log_2\left( \frac{2^{23}(32-\varepsilon)^2}{ (16-\varepsilon)\varepsilon^6}\right)      
        \,,\\
        K^{(\varepsilon,n)}(\mathcal{E}_{\lambda})&\ge Q_2^{(\varepsilon,n)}(\mathcal{E}_{\lambda}) \ge n Q_2(\mathcal{E}_{\lambda})- \log_2\!\left(\frac{2^{6}\,3\,(4-\sqrt{\varepsilon})^2}{(2-\sqrt{\varepsilon})\varepsilon^3 }\right) \,,
    \ee
    where $Q(\mathcal{E}_{\lambda})=\max\!\left(0,\log_2\!\left(\frac{\lambda}{1-\lambda}\right)\right)$ and $Q_2(\mathcal{E}_{\lambda})= \log_2\!\left(\frac{1}{1-\lambda}\right)$ are the (asymptotic) quantum capacity and the (asymptotic) two-way quantum capacity, respectively.
\end{thm}
The lower bound on the $n$-shot two-way quantum capacity and secret-key capacity in \eqref{ineq_best_bounds} have to be compared with the upper bound established in \cite{MMMM}:
\bb\label{upper_bound_MMMM}
    Q_2^{(\varepsilon,n)}(\mathcal{E}_{\lambda})\le K^{(\varepsilon,n)}(\mathcal{E}_{\lambda})\le n\log_2\!\left(\frac{1}{1-\lambda}\right)+\log_26+2\log_2\!\left(\frac{1+\varepsilon}{1-\varepsilon}\right)\,.
\ee
Consequently, by leveraging \eqref{ineq_best_bounds} and \eqref{upper_bound_MMMM}, we conclude that for constant $\varepsilon$,  the following holds:
\bb
    Q_2^{(\varepsilon,n)}(\mathcal{E}_{\lambda})=K^{(\varepsilon,n)}(\mathcal{E}_{\lambda})= n\log_2\!\left(\frac{1}{1-\lambda}\right) + O(1)\,.
\ee
This result proves the conjecture made in \cite{Kaur_2017}, demonstrating that the upper bound in \eqref{upper_bound_MMMM} is nearly optimal.

\subsection{Channel complexity of the pure loss channel}\label{sec_channel_compl}
Traditionally, the analysis of quantum communication over Gaussian channels has been conducted in the \emph{asymptotic} setting, where the number of channel uses approaches infinity and the error in the quantum communication task vanishes. However, due to technological limitations and experimental imperfections, it has become essential to also examine the \emph{non-asymptotic} setting, where the number of channel uses is a fixed (finite) number $n$ and the goal is to achieve a \emph{non-zero} error that it is required to be below a certain threshold $\varepsilon$. In this non-asymptotic setting, we address the following questions:
\begin{itemize}
    \item How many uses of the pure loss channel $\mathcal{E}_\lambda$ are sufficient in order to transmit $k$ qubits with an error of at most $\varepsilon$?
    \item How many uses of the pure loss channel $\mathcal{E}_\lambda$, assisted by arbitrary LOCCs, are sufficient in order to generate $k$ ebits (and hence $k$ secret-key bits) with an error of at most $\varepsilon$?
\end{itemize}
Answering these questions is crucial since the pure loss channel is the most realistic model for optical links, including optical fibres and free-space links.

In the non-asymptotic setting, it is crucial to determine the minimum number $n$  of channel uses required to perform a given quantum communication task with an error of at most $\varepsilon$. Determining this minimum $n$ is important because it indicates how long one needs to wait in order to complete the quantum communication task. This minimum $n$ can be referred to as \emph{channel complexity}, analogous to the concepts of "sample complexity" or "query complexity" well-studied in the quantum learning theory literature~\cite{anshu2023survey}. The channel complexity for qubit distribution, LOCC-assisted entanglement distribution, and LOCC-assisted secret-key distribution over the pure loss channel will be denoted as $n_{Q}(\lambda,\varepsilon,k)$, $n_{Q_2}(\lambda,\varepsilon,k)$, and $n_{K}(\lambda,\varepsilon,k)$, respectively. Below, we establish bounds on these channel complexities.
\begin{thm}[(Upper bound on the channel complexity of qubit distribution)]\label{thm_activation_complexityQ}
    Let $\lambda\in(\frac{1}{2},1)$, $\varepsilon\in(0,1)$, and $k\in\mathbb{N}$. It holds that
    \bb
        n_{Q}(\lambda,\varepsilon,k)\le \frac{k+\log_2\!\left( \frac{2^{23}(32-\varepsilon)^2}{ (16-\varepsilon)\varepsilon^6}\right)}{\log_2\!\left(\frac{\lambda}{1-\lambda}\right)}\,.
    \ee
    More explicitly, exploiting a number $\frac{k+\log_2\!\left( \frac{2^{23}(32-\varepsilon)^2}{ (16-\varepsilon)\varepsilon^6}\right)}{\log_2\!\left(\frac{\lambda}{1-\lambda}\right)}$ of uses of the pure loss channel $\mathcal{E}_\lambda$ is sufficient in order to transmit $k$ qubits with an error of at most $\varepsilon$. Here, the error of at most $\varepsilon$ is measured in terms of the channel fidelity between the encoded channel and the identity map as defined in Eq.~\eqref{eq_def_Q_cap}.  
\end{thm}
\begin{proof}
In order to transmit $k$ qubits with error at most $\varepsilon$ across $n$ uses of the pure loss channel $\mathcal{E}_\lambda$, it is sufficient to exploit a number $n$ of uses of the pure loss channel $\mathcal{E}_\lambda$ such that the $n$-shot quantum capacity satisfies $Q^{(\varepsilon,n)}(\mathcal{E}_\lambda)\ge k$. Hence, by exploiting Theorem~\ref{best_bounds}, it thus suffices to require that
\bb
    n\log_2\!\left(\frac{\lambda}{1-\lambda}\right)-\log_2\left( \frac{2^{23}(32-\varepsilon)^2}{ (16-\varepsilon)\varepsilon^6}\right)\ge k\,.
\ee
By solving the above inequality with respect $n$, we conclude the proof.
\end{proof}
Let us now state an analogous result for the task of entanglement distribution and secret key distribution.
\begin{thm}[(Upper bound on the channel complexity of LOCC-assisted entanglement distribution and secret-key distribution)]\label{thm_activation_complexityQ2}
Let $\lambda\in(0,1)$, $\varepsilon\in(0,1)$, and $k\in\mathbb{N}$. It holds that
\bb
    n_{K}(\lambda,\varepsilon,k)\le n_{Q_2}(\lambda,\varepsilon,k)\le \frac{k+\log_2\!\left(\frac{2^{6}\,3\,(4-\sqrt{\varepsilon})^2}{(2-\sqrt{\varepsilon})\varepsilon^3 }\right) }{\log_2\!\left(\frac{1}{1-\lambda}\right)}\,,
\ee
More explicitly, exploiting a number $\frac{k+\log_2\!\left(\frac{2^{6}\,3\,(4-\sqrt{\varepsilon})^2}{(2-\sqrt{\varepsilon})\varepsilon^3 }\right) }{\log_2\!\left(\frac{1}{1-\lambda}\right)}$ of uses of the pure loss channel $\mathcal{E}_\lambda$, together with LOCCs, is sufficient in order to generate $k$ ebits (and hence $k$ secret-key bits) with an error of at most $\varepsilon$. Here, the error of at most $\varepsilon$ is measured in terms of the fidelity between the output state of the  protocol and the maximally entangled state as defined in Eq.~\eqref{def_eq_two_way_q}. 
\end{thm}
\begin{proof}
In order to transmit $k$ ebits with error at most $\varepsilon$ across $n$ uses of the pure loss channel $\mathcal{E}_\lambda$, it is sufficient to exploit a number $n$ of uses of the pure loss channel $\mathcal{E}_\lambda$ such that the $n$-shot two-way quantum capacity satisfies $Q^{(\varepsilon,n)}(\mathcal{E}_\lambda)\ge k$. Hence, by exploiting Theorem~\ref{best_bounds}, it thus suffices to require that
\bb
    n\log_2\!\left(\frac{1}{1-\lambda}\right)-\log_2\!\left(\frac{2^{6}\,3\,(4-\sqrt{\varepsilon})^2}{(2-\sqrt{\varepsilon})\varepsilon^3 }\right) \ge k\,,
\ee
By solving the above inequality with respect $n$, we conclude the proof.
\end{proof}
By exploiting the upper bound on the $n$-shot capacities reported in \eqref{upper_bound_MMMM} and proved in \cite{MMMM}, one can also find a lower bound on the channel complexity of LOCC-assisted entanglement distribution and secret-key distribution over the pure loss channel.
\begin{thm}[(Lower bound on the channel complexity of LOCC-assisted entanglement distribution and secret-key distribution)] 
Let $n\in\mathbb{N}$, $\lambda\in(0,1)$, $\varepsilon\in(0,1)$, and $k\in\mathbb{N}$. It holds that
\bb
    n_{Q_2}(\lambda,\varepsilon,k)\ge n_{K}(\lambda,\varepsilon,k)\ge 
     \frac{k-\log_2\left(\frac{6(1+\varepsilon)^2}{(1-\varepsilon)^2}\right)}{ \log_2\!\left(\frac{1}{1-\lambda}\right)  }\,.  
\ee
In other words, if $n$ uses of the pure loss channel $\mathcal{E}_\lambda$, together with LOCCs, allows one to generate $k$ ebits (and hence $k$ secret-key bits) with an error of at most $\varepsilon$, then it must hold that
\bb
    n\ge \frac{k-\log_2\!\left(\frac{6(1+\varepsilon)^2}{(1-\varepsilon)^2}\right)}{ \log_2\!\left(\frac{1}{1-\lambda}\right)  }\,.
\ee
\end{thm}
\begin{proof}
By definition of $n$-shot capacities, it must hold that $k\le K^{(\varepsilon,n)}(\mathcal{E}_\lambda)$. Moreover, by exploiting the upper bound on the $n$-shot capacities reported in \eqref{upper_bound_MMMM} and proved in \cite{MMMM}, we have that
\bb
    k\le n\log_2\!\left(\frac{1}{1-\lambda}\right)+\log_26+2\log_2\!\left(\frac{1+\varepsilon}{1-\varepsilon}\right)\,.
\ee
By solving the above inequality with respect $n$, we conclude the proof.
\end{proof}
By summarising, the channel complexity of LOCC-assisted entanglement distribution and secret-key distribution over the pure loss channel satisfy
\bb
    \frac{k-\log_2\left(\frac{6(1+\varepsilon)^2}{(1-\varepsilon)^2}\right)}{ \log_2\!\left(\frac{1}{1-\lambda}\right)  }\le n_{K}(\lambda,\varepsilon,k)\le n_{Q_2}(\lambda,\varepsilon,k)\le \frac{k+\log_2\!\left(\frac{2^{6}\,3\,(4-\sqrt{\varepsilon})^2}{(2-\sqrt{\varepsilon})\varepsilon^3 }\right) }{\log_2\!\left(\frac{1}{1-\lambda}\right)}\,.
\ee

\section{Non-asymptotic analysis of quantum communication across the pure amplifier channel}\label{sec_lower_bound_pure_ampl_channel}
This section aims to establish easily computable lower bounds on the $n$-shot quantum capacity, two-way quantum capacity, and secret-key capacity of the pure amplifier channel~\cite{BUCCO} (see Definition~\ref{def_pure_ampl} for the definition of pure amplifier channel). To establish such lower bounds, we exploit the asymptotic equipartition property approach developed in Section~\ref{section_AEP0}. Specifically, we will evaluate the lower bounds on the $n$-shot capacities proved in Theorem~\ref{thm_lower_bound} and Theorem~\ref{thm_lower_bound_2way} for the physically interesting case of the pure amplifier channel $\Phi_g$.

Let us apply such theorems by choosing as input state $\Phi_{AA'}$ the two-mode squeezed vacuum state $\ketbra{\Psi_{N_s}}_{AA'}$ with local mean photon number equal to $N_s$, as defined in Eq.~\eqref{eq:2mode_squeezed}. Moreover, let us recall that the pure amplifier channel $\Phi_{g}$, defined in Definition~\ref{def_pure_ampl}, is a single-mode Gaussian channel that admits the following Stinespring dilation:
\bb
        \Phi_{g}(\cdot)=\Tr_E\left[U_g^{A'E\to BE} \big((\cdot) \otimes\ketbra{0}_E\big) (U_g^{A'E\to BE})^\dagger\right]\,,
\ee
where $U_g^{A'E\to BE}$ denotes the two-mode squeezing unitary defined in Definition~\ref{def_two_sq}. In order to apply Theorem~\ref{thm_lower_bound} and Theorem~\ref{thm_lower_bound_2way}, we need to calculate the covariance matrix $V\!\left(\Psi_{ABE}^{(g,N_s)}\right)$ of the tri-partite state
\bb
    \Psi_{ABE}^{(g,N_s)}\coloneqq\left(\mathbb{1}_A\otimes U_g^{BE}\right)\,\left( \ketbra{\Psi_{N_s}}_{AB}\otimes\ketbra{0}_{E}\right) \left(\mathbb{1}_A\otimes U_g^{BE}\right)^\dagger\,,
\ee
e.g.~with respect the mode-ordering $(A,B,E)$. By exploiting Eq.~\eqref{law_sympl_unitary} and Eq.~\eqref{symplectic_squez}, it follows that $V\!\left(\Psi_{ABE}^{(g,N_s)}\right)$ can be expressed as
\bb\label{cov_ABEE_comp2}
V\!\left(\Psi_{ABE}^{(g,N_s)}\right)=  
\left(\mathbb{1}_2\oplus {S}_g\right) \left(V(\ketbra{\Psi_{N_s}}_{AB})\oplus  V(\ketbra{0})\right) \left(\mathbb{1}_2\oplus {S}_g^{\intercal}\right) \,,
\ee
where:
\bb
S_g\coloneqq	\begin{pmatrix}
		\sqrt{g}\,\mathbb{1}_2\, & \,\sqrt{g-1}\,\sigma_z \\
		\sqrt{g-1}\,\sigma_z\, &\, \sqrt{g}\,\mathbb{1}_2
	\end{pmatrix}\,
\ee
is the symplectic matrix associated with the two-mode squeezing unitary, $V(\ketbra{\Psi_{N_s}})$ is the covariance matrix of the two-mode squeezed vacuum state given in Eq.~\eqref{cov_two_mode_squeezed}, and $V(\ketbra{0})=\mathbb{1}_2$ is the covariance matrix of the vacuum state. By explicitly performing the matrix multiplications, one obtains that 
\bb
{\fontsize{8.2}{8.2}\selectfont
    V\!\left(\Psi_{ABE}^{(g,N_s)}\right)=
\begin{bmatrix}
1 + 2 N_s & 0 & 2 \sqrt{g N_s (1 + N_s)} & 0 & 2 \sqrt{(g - 1) N_s (1 + N_s)} & 0 \\
0 & 1 + 2 N_s & 0 & -2 \sqrt{g N_s (1 + N_s)} & 0 & 2 \sqrt{(g - 1) N_s (1 + N_s)} \\
2 \sqrt{g N_s (1 + N_s)} & 0 & -1 + g (2 N_s + 1) & 0 & 2 \sqrt{g (g - 1)} (N_s + 1) & 0 \\
0 & -2 \sqrt{g N_s (1 + N_s)} & 0 & -1 + g (2 N_s + 1) & 0 & -2 \sqrt{g (g - 1)} (N_s + 1) \\
2 \sqrt{(g - 1) N_s (1 + N_s)} & 0 & 2 \sqrt{g (g - 1)} (N_s + 1) & 0 & 1 + 2 N_s (g - 1) & 0 \\
0 & 2 \sqrt{(g - 1) N_s (1 + N_s)} & 0 & -2 \sqrt{g (g - 1)} (N_s + 1) & 0 & 1 + 2 N_s (g - 1)
\end{bmatrix}\,.}
\ee
Hence, by exploiting Lemma~\ref{thm_cond_petz_Gaussian}, one can explicitly calculate all the conditional Petz-Rényi entropies involved in Theorem~\ref{thm_lower_bound} and Theorem~\ref{thm_lower_bound_2way}. By a simple calculation performed on the Mathematica notebook attached to this manuscript, one obtains:

{\small
    \bb\label{ingredient2_ec_amp}
H_{1/2}(A|B)_{\Psi_{ABE}^{(g,N_s)}}&=   \log_2\left(
\frac{
\left[ 
2 g (1 + N_s) - 1 + 2 \sqrt{g (1 + N_s) \big(-1 + g (1 + N_s)\big)}
\right]
\left[
2 g (1 + N_s) - 1 - 2 N_s + 2 \sqrt{(g - 1) (1 + N_s) \big(g + (g - 1) N_s\big)}
\right]
}{
\left[
g(2N_s+1)+\sqrt{g N_s (g N_s+1)}+g (N_s+1)\sqrt{\frac{(g-1)(N_s+1)}{(g-1)N_s+g}}
\right]^2
}
\right)\,,\\
    H_{1/2}(A|E)_{\Psi_{ABE}^{(g,N_s)}}&=\log_2 \left( 
\frac{
\left[ -1 - 2N_s + 2g(1 + N_s) + 2\sqrt{(g - 1)(1 + N_s)\big(g + (g - 1)N_s\big)} \right] 
\left[ -1 + 2g(1 + N_s) + 2\sqrt{g(1 + N_s)(g - 1 + gN_s)}\right]
}{
\left[
g+(g-1)(2N_s+1) +\sqrt{(g-1) (N_s+1) ((g-1) N_s+g)}+(g-1) (N_s+1) \sqrt{\frac{g N_s+g}{g N_s+g-1}}  
\right]^2}\right)
\,,\\
\ee}
Moreover, thanks to Eq.~\eqref{q_ec_pure_ampl}, we can also obtain the coherent information quantity appearing in Theorem~\ref{thm_lower_bound} and Theorem~\ref{thm_lower_bound_2way}:
\bb\label{ingredient1_ec_amp}
    I_c(A\,\rangle\, B)_{ \Psi_{ABE}^{(g,N_s)}  } &=h\!\left( gN_s + g-1\right)-h\!\left( (g-1)(N_s+1) \right)\,,
\ee
where $h(x)\coloneqq (x+1)\log_2(x+1)-x\log_2x$. By taking the limit of infinite energy $N_s\rightarrow\infty$, one obtains
\bb\label{ingredient1_amp}
    \lim\limits_{N_s\rightarrow\infty}H_{1/2}(A|B)_{\Psi_{ABE}^{(g,N_s)}}&=\log_2\!\left(\frac{g-1}{g}\right)\,,\\
    \lim\limits_{N_s\rightarrow\infty}H_{1/2}(A|E)_{\Psi_{ABE}^{(g,N_s)}}&= \log_2\!\left(\frac{g}{g-1}\right)\,,\\
\ee
and, thanks to Eq.~\eqref{coh_ampl}, one also obtains that
\bb\label{ingredient2_amp}
    \lim\limits_{N_s\rightarrow\infty}I_c(A\,\rangle\, B)_{ \Psi_{ABE}^{(g,N_s)}  } &=\log_2\!\left(\frac{g}{g-1}\right)\,.
\ee
By putting together Eq.~\eqref{ingredient1_amp}, Eq.~\eqref{ingredient2_amp}, Theorem~\ref{thm_lower_bound}, and Theorem~\ref{thm_lower_bound_2way}, we obtain a lower bound on the $n$-shot capacities of the pure amplifier channel.
\begin{thm}[(Lower bound on the $n$-shot capacities of the pure amplifier channel)]\label{thm_one_shot_Q_pure_ampl}
    Let $g>1$, $\varepsilon\in(0,1)$, and $n\in\N$ such that $n\ge2\log_2\!\left(\frac{2}{\varepsilon^2}\right)$. The $n$-shot quantum capacity, the $n$-shot two-way quantum capacity, and the $n$-shot secret key capacity of the pure amplifier channel satisfies:
\bb\label{lower_bound_q_pure_amp}
    Q^{(\varepsilon,n)}(\Phi_g)&\ge nQ(\Phi_g)-\sqrt{n}4\log_2\!\left( \sqrt{\frac{g-1}{g}}+ \sqrt{\frac{g}{g-1}} +1 \right)\sqrt{\log_2\!\left(\frac{2^9}{\varepsilon^2}\right)}+\log_2\!\left(\frac{3\varepsilon^4}{2^{18}}\right)\,,\\
    K^{(\varepsilon,n)}(\Phi_g)&\ge Q_2^{(\varepsilon,n)}(\Phi_g) \ge nQ_2(\Phi_g)-\sqrt{n}4\log_2\!\left(\sqrt{\frac{g-1}{g}}+ \sqrt{\frac{g}{g-1}} +1 \right)\sqrt{\log_2\!\left(\frac{8}{\varepsilon}\right)}-\log_2\!\left(\frac{16}{\varepsilon^2}\right)\,,
\ee
where 
\bb
    Q(\Phi_{g})&=Q_2(\Phi_{g})= \log_2\!\left(\frac{g}{g-1}\right)
\ee
is the quantum capacity of the pure amplifier channel, which is equal to its two-way quantum capacity.
\end{thm}
By using \eqref{ingredient2_ec_amp}, \eqref{ingredient1_ec_amp}, Theorem~\ref{thm_lower_bound}, and Theorem~\ref{thm_lower_bound_2way}, we can also find a lower bound on the energy-constrained $n$-shot capacities of the pure amplifier channel, as stated in the following theorem.
\begin{thm}[(Lower bound on the energy-constrained $n$-shot capacities of the pure amplifier channel)]\label{thm_lower_ec_amp}
    Let $N_s>0$, $g>1$, $\varepsilon\in(0,1)$, and $n\in\N$ such that $n\ge2\log_2\!\left(\frac{2}{\varepsilon^2}\right)$. The energy-constrained $n$-shot quantum, two-way quantum, and secret-key capacity of the pure amplifier channel satisfy
    \bb
        Q^{(\varepsilon,n)}(\Phi_g,N_s)&\ge  n \left(h\!\left( gN_s + g-1\right)-h\!\left( (g-1)(N_s+1) \right)\right)\\
        &\qquad-\sqrt{n}4\log_2\left(\sqrt{2^{H_{1/2}(A|B)_{\Psi_{ABE}^{(g,N_s)}}}}+\sqrt{2^{H_{1/2}(A|E)_{\Psi_{ABE}^{(g,N_s)}}}}+1\right)\sqrt{\log_2\!\left(\frac{2^9}{\varepsilon^2}\right)}+\log_2\!\left(\frac{3\varepsilon^4}{2^{18}}\right)\,,\\
        K^{(\varepsilon,n)}(\Phi_g,N_s)&\ge Q_2^{(\varepsilon,n)}(\Phi_g,N_s)  \\
        &\ge  n \left(h\!\left( gN_s + g-1\right)-h\!\left( (g-1)(N_s+1) \right)\right)\\
        &\qquad\qquad-\sqrt{n}4\log_2\left(\sqrt{2^{H_{1/2}(A|B)_{\Psi_{ABE}^{(g,N_s)}}}}+\sqrt{2^{H_{1/2}(A|E)_{\Psi_{ABE}^{(g,N_s)}}}}+1\right)\sqrt{\log_2\!\left(\frac{8}{\varepsilon}\right)}-\log_2\!\left(\frac{16}{\varepsilon^2}\right)\,\,,
    \ee
    where $H_{1/2}(A|B)_{\Psi_{ABE}^{(g,N_s)}}$ and $H_{1/2}(A|E)_{\Psi_{ABE}^{(g,N_s)}}$ are reported in \eqref{ingredient1_ec_amp}. 
\end{thm}

\end{document}

%% file: newdef.tex
\usepackage{newpxtext,newpxmath}

\let\coloneqq\relax

\usepackage{amsthm}
\usepackage{amssymb}
\usepackage{amsmath}
\usepackage{bbold}
\usepackage{bbm}
\usepackage{nonfloat}
\usepackage[pdftex]{hyperref}
\usepackage{braket}
\usepackage{dsfont}
\usepackage{mathdots}
\usepackage{mathtools}
\usepackage{enumerate}
\usepackage[shortlabels]{enumitem}
\usepackage{csquotes}
\usepackage{stmaryrd}
\usepackage[cal=boondox]{mathalfa}
\usepackage{graphicx}
\usepackage{stackengine}
\usepackage{scalerel}
\usepackage{xr}
\usepackage{array}
\usepackage{makecell}
\newcolumntype{x}[1]{>{\centering\arraybackslash}p{#1}}
\usepackage{tikz}
\usepackage{pgfplots}
\usetikzlibrary{shapes.geometric, shapes.misc, positioning, arrows, arrows.meta, decorations.pathreplacing, decorations.pathmorphing, patterns, angles, quotes, calc}
\usepackage{booktabs}
\usepackage{xfrac}
\usepackage{siunitx}
\usepackage{centernot}
\usepackage{comment}
\usepackage{chngcntr}

\newtheorem{thm}{Theorem}
\newtheorem*{thm*}{Theorem}

\newtheorem*{prop*}{Proposition}
\newtheorem{lemma}[thm]{Lemma}
\newtheorem*{lemma*}{Lemma}

\newtheorem{cor}[thm]{Corollary}
\newtheorem*{cor*}{Corollary}

\newtheorem*{cj*}{Conjecture}
\newtheorem{Def}[thm]{Definition}
\newtheorem*{Def*}{Definition}

\newtheorem{remark}{Remark}

\makeatletter
\def\thmhead@plain#1#2#3{%
  \thmname{#1}\thmnumber{\@ifnotempty{#1}{ }\@upn{#2}}%
  \thmnote{ {\the\thm@notefont#3}}}
\let\thmhead\thmhead@plain
\makeatother

\theoremstyle{definition}

\newcommand{\bb}{\begin{equation}\begin{aligned}\hspace{0pt}}
\newcommand{\bbb}{\begin{equation*}\begin{aligned}}
\newcommand{\ee}{\end{aligned}\end{equation}}
\newcommand{\eee}{\end{aligned}\end{equation*}}
\newcommand*{\coloneqq}{\mathrel{\vcenter{\baselineskip0.5ex \lineskiplimit0pt \hbox{\scriptsize.}\hbox{\scriptsize.}}} =}

\newcommand{\eqt}[1]{\stackrel{\mathclap{\scriptsize \mbox{#1}}}{=}}
\newcommand{\leqt}[1]{\stackrel{\mathclap{\scriptsize \mbox{#1}}}{\leq}}

\newcommand{\geqt}[1]{\stackrel{\mathclap{\scriptsize \mbox{#1}}}{\geq}}
\newcommand{\ketbra}[1]{\ket{#1}\!\!\bra{#1}}
\newcommand{\ketbraa}[2]{\ket{#1}\!\!\bra{#2}}

\newcommand{\N}{\mathds{N}}

\DeclareMathOperator{\Tr}{Tr}

\DeclareMathAlphabet{\pazocal}{OMS}{zplm}{m}{n}

\DeclareMathOperator{\Id}{Id}

\newcommand{\HH}{\pazocal{H}}

\newcommand{\NN}{\mathcal{N}}

\newcommand{\lsmatrix}{\left(\begin{smallmatrix}}
\newcommand{\rsmatrix}{\end{smallmatrix}\right)}

\stackMath

\stackMath

\makeatletter
\newcommand*\rel@kern[1]{\kern#1\dimexpr\macc@kerna}
\newcommand*\widebar[1]{%
  \begingroup
  \def\mathaccent##1##2{%
    \rel@kern{0.8}%
    \overline{\rel@kern{-0.8}\macc@nucleus\rel@kern{0.2}}%
    \rel@kern{-0.2}%
  }%
  \macc@depth\@ne
  \let\math@bgroup\@empty \let\math@egroup\macc@set@skewchar
  \mathsurround\z@ \frozen@everymath{\mathgroup\macc@group\relax}%
  \macc@set@skewchar\relax
  \let\mathaccentV\macc@nested@a
  \macc@nested@a\relax111{#1}%
  \endgroup
}

\counterwithin*{equation}{part}
\counterwithin*{thm}{part}
\counterwithin*{figure}{part}

\tikzset{meter/.append style={draw, inner sep=10, rectangle, font=\vphantom{A}, minimum width=30, line width=.8, path picture={\draw[black] ([shift={(.1,.3)}]path picture bounding box.south west) to[bend left=50] ([shift={(-.1,.3)}]path picture bounding box.south east);\draw[black,-latex] ([shift={(0,.1)}]path picture bounding box.south) -- ([shift={(.3,-.1)}]path picture bounding box.north);}}}
\tikzset{roundnode/.append style={circle, draw=black, fill=gray!20, thick, minimum size=10mm}}
\tikzset{squarenode/.style={rectangle, draw=black, fill=none, thick, minimum size=10mm}}

\definecolor{Blues5seq1}{RGB}{239,243,255}
\definecolor{Blues5seq2}{RGB}{189,215,231}
\definecolor{Blues5seq3}{RGB}{107,174,214}
\definecolor{Blues5seq4}{RGB}{49,130,189}
\definecolor{Blues5seq5}{RGB}{8,81,156}

\definecolor{Greens5seq1}{RGB}{237,248,233}
\definecolor{Greens5seq2}{RGB}{186,228,179}
\definecolor{Greens5seq3}{RGB}{116,196,118}
\definecolor{Greens5seq4}{RGB}{49,163,84}
\definecolor{Greens5seq5}{RGB}{0,109,44}

\definecolor{Reds5seq1}{RGB}{254,229,217}
\definecolor{Reds5seq2}{RGB}{252,174,145}
\definecolor{Reds5seq3}{RGB}{251,106,74}
\definecolor{Reds5seq4}{RGB}{222,45,38}
\definecolor{Reds5seq5}{RGB}{165,15,21}

\usepackage{pgfplots}

\DeclareMathOperator{\arccosh}{arccosh}
\pgfplotsset{width=10cm,compat=1.9}
\usetikzlibrary{decorations.markings}

\definecolor{marco}{rgb}{.4,.2,.6}